\documentclass[a4paper,UKenglish,thm-restate]{lipics-v2021}

\usepackage{amsmath}

\usepackage{amsfonts}
\usepackage{amsthm}
\usepackage{amsmath}
\usepackage{array}
\usepackage{cite}
\usepackage[dvipsnames,svgnames,x11names]{xcolor}
\usepackage{empheq}
\usepackage{graphicx}
\usepackage{latexsym}
\usepackage{mathtools}
\usepackage{multicol}
\usepackage{paralist}
\usepackage{relsize}
\usepackage{stmaryrd}
\usepackage{tabularx}
\usepackage{thmtools}
\usepackage[]{units}
\usepackage{url}
\usepackage{verbatim}
\usepackage{wasysym}

\allowdisplaybreaks

\nolinenumbers

\newcommand{\lmerge}{\mathbin{
                  \setlength{\unitlength}{1ex}
                  \begin{picture}(1,1.75)
                  \put(0,0){\line(1,0){1}}
                  \put(0,0){\line(0,1){1.75}}
                  \put(0.45,0){\line(0,1){1.75}}
                  \end{picture}
                 }}
\newcommand{\cmerge}{~|~}

\newcommand{\textl}{~\!\lmerge}
\newcommand{\textc}{|}

\newcommand{\ccslc}{\text{CCS}_{\text{LC}}}

\newcommand{\nil}{\mathbf{0}}
\newcommand{\Act}{\mathcal{A}}
\newcommand{\Acttau}{\mathcal{A}_\tau}
\newcommand{\Var}{\mathcal{V}}

\newcommand{\size}{\mathrm{size}}
\newcommand{\depth}{\mathrm{depth}}
\newcommand{\rdepth}{\mathrm{Rdepth}}

\newcommand{\var}{\mathrm{var}}
\newcommand{\init}{\mathrm{init}}
\newcommand{\der}{\mathrm{der}}

\newcommand{\trans}[1][]{\xrightarrow{\, {#1} \, }}
\newcommand{\ntrans}[1][]{\mathrel{{\trans[#1]}\makebox[0em][r]{$\not$\hspace{2ex}}}{\!}}

\newcommand{\xitrans}{\overset{\xi}\twoheadrightarrow}
\newcommand{\wtrans}[1][]{\overset{#1}\Longrightarrow}

\newcommand{\SOSrule}[2]{\frac{\displaystyle #1}{\displaystyle #2}}

\newcommand{\rel}{\,{\mathcal R}\,}

\newcommand{\Proc}{\mathbf{P}}
\newcommand{\C}{\mathcal{C}}

\newcommand{\dd}{\mathrm{d}}
\newcommand{\e}{\mathfrak{e}}
\newcommand{\E}{\mathcal{E}}

\newcommand{\N}{\mathbb{N}}

\newcommand{\B}{\mathtt{B}}

\newcommand{\bb}{\mathtt{BB}}

\newcommand{\rbb}{\mathtt{RBB}}
\newcommand{\reb}{\mathtt{R}\eta\mathtt{B}}
\newcommand{\rdb}{\mathtt{RDB}}
\newcommand{\rwb}{\mathtt{RWB}}

\usepackage[normalem]{ulem}

\title{On the Axiomatisation of Branching Bisimulation Congruence over CCS}
\titlerunning{On the axiomatisation of branching bisimulation congruence over CCS}

\author{Luca Aceto}{Reykjavik University, Iceland 
\and Gran Sasso Science Institute (GSSI), Italy}{}{https://orcid.org/0000-0002-2197-3018}{}

\author{Valentina Castiglioni}{Reykjavik University, Iceland}{}{https://orcid.org/0000-0002-8112-6523}{}

\author{Anna Ing{\'o}lfsd{\'o}ttir}{Reykjavik University, Iceland}{}{https://orcid.org/0000-0001-8362-3075}{}

\author{Bas Luttik}{Eindhoven University of Technology, The Netherlands}{}{https://orcid.org/0000-0001-6710-8436}{}

\authorrunning{L. Aceto, V. Castiglioni, A. Ing\'olfsd\'ottir, and B. Luttik} 

\Copyright{Luca Aceto, Valentina Castiglioni, Anna Ing\'olfsd\'ottir, and Bas Luttik} 



\funding{This work has been partially supported by the project ``\emph{Open Problems in the Equational Logic of Processes}'' (OPEL) of the Icelandic Research Fund (grant No.~196050-051). 
V. Castiglioni has been supported by the project ``\emph{Programs in the wild: Uncertainties, adaptabiLiTy and veRificatiON}'' (ULTRON) of the Icelandic Research Fund (grant No.~228376-051).}

\EventEditors{Bartek Klin, S\l{}awomir Lasota, and Anca Muscholl}
\EventNoEds{3}
\EventLongTitle{33rd International Conference on Concurrency Theory (CONCUR 2022)}
\EventShortTitle{CONCUR 2022}
\EventAcronym{CONCUR}
\EventYear{2022}
\EventDate{September 12--16, 2022}
\EventLocation{Warsaw, Poland}
\EventLogo{}
\SeriesVolume{243}
\ArticleNo{28}

\hideLIPIcs

\begin{document}

\maketitle

\begin{abstract}
In this paper we investigate the equational theory of (the restriction, relabelling, and recursion free fragment of) CCS modulo rooted branching bisimilarity, which is a classic, bisimulation-based notion of equivalence that abstracts from internal computational steps in process behaviour.
Firstly, we show that CCS is not finitely based modulo the considered congruence.
As a key step of independent interest in the proof of that negative result, we prove that each CCS process has a unique parallel decomposition into indecomposable processes modulo branching bisimilarity.
As a second main contribution, we show that, when the set of actions is finite, rooted branching bisimilarity has a finite equational basis over CCS enriched with the left merge and communication merge operators from ACP.
\end{abstract}


\section{Introduction}

This paper is a new chapter in the saga of the axiomatisation of the \emph{parallel composition operator} $\mathbin{\|}$ (also known as ``\emph{full}'' \emph{merge} \cite{BK84b,BK85}) of the Calculus of Communicating Systems (CCS) \cite{Mi80}.
The saga has its roots in the works \cite{HM80,HM85}, in which Hennessy and Milner studied the \emph{equational  theory} of (recursion free) CCS and proposed a \emph{ground-complete axiomatisation} for it modulo \emph{strong bisimilarity} and \emph{observational congruence}, two classic notions of behavioural \emph{congruence} (i.e., an equivalence relation that is compositional with respect to the language operators) that allow one to establish whether two processes have the same \emph{observable behaviour} \cite{Pa81}.
That axiomatisation included infinitely many axioms, which were instances of the \emph{expansion law} used to ``simulate equationally'' the operational semantics of $\mathbin{\|}$.
Then, Bergstra and Klop showed, in \cite{BK84b}, that a \emph{finite} ground-complete axiomatisation modulo bisimilarity can be obtained by enriching CCS with two auxiliary operators, i.e., the \emph{left merge} $\textl$ and the \emph{communication merge} $\textc$, expressing one step in the pure interleaving and the synchronous behaviour of $\|$, respectively.
Their result was strengthened by Aceto et al.\ in \cite{AFIL09}, where it is proved that, over the fragment of CCS without recursion, restriction and relabelling, the auxiliary operators $\textl$ and $\textc$ allow for finitely axiomatising $\|$ modulo bisimilarity also when CCS terms with variables are considered.
Moreover, in \cite{AILT08} that result is extended to the fragment of CCS with relabelling and restriction, but without communication.
From those studies, we can infer that $\textl$ and $\textc$ are \emph{sufficient} to finitely axiomatise $\mathbin{\|}$ over CCS modulo bisimilarity. 
(\emph{Henceforth, we only consider the recursion, restriction and relabelling free fragment of CCS.})
Moller showed, in \cite{Mo89,Mo90}, that they are also \emph{necessary}.
He considered a minimal fragment of CCS, including only the inactive process, action prefixing, nondeterministic choice and interleaving, and proved that, even in the presence of a single action, bisimilarity does not afford a finite ground-complete axiomatisation over that language. 
Moller's proof technique was then used to show that the same negative result holds if we replace $\textl$ and $\textc$ with the so called \emph{Hennessy's merge} \cite{He88}, which denotes an asymmetric interleaving with communication, or, more generally, with a single binary auxiliary operator satisfying three assumptions given in \cite{ACFIL21}.

The aforementioned works considered equational characterisations of $\mathbin{\|}$ modulo strong bisimilarity.
However, a plethora of behavioural congruences have been proposed in the literature, corresponding to different levels of abstraction from the information on process execution.
Hence, another chapter in the saga consisted in extending the studies recalled above to the behavioural congruences in van Glabbeek's linear time-branching time spectrum \cite{vG90}.
The work \cite{ACILP20} delineated the \emph{boundary} between finite and non-finite axiomatisability of $\mathbin{\|}$ modulo all the congruences in the spectrum.


\subparagraph{Our contribution: branching bisimulation congruence.}

Some information on process behaviour can either be considered irrelevant or be unavailable to an external observer.
\emph{Weak behavioural semantics} have been introduced to study the effects of these unobservable (or \emph{silent}) actions, usually denoted by $\tau$, on the observable behaviour of processes, each semantics considering a different level of abstraction.
A taxonomy of weak semantics is given in \cite{vG93}, and studies on the equational theories of various of these semantics have been carried out over the algebra BCCSP, which consists of the basic operators from CCS and CSP \cite{Ho85} but does not include $\mathbin{\|}$ (see, among others, \cite{AdFGI14,CFvG08,dNH83,vGW96,HM85}).
A finite, ground-complete axiomatisation of parallel composition modulo \emph{rooted weak bisimilarity} (also known as \emph{observational congruence} \cite{HM85}) is provided by Bergstra and Klop in \cite{BK85} over the algebra $\text{ACP}_\tau$ that includes the auxiliary operators $\textl$ and $\textc$.
To the best of our knowledge, the only study on the axiomatisability of CCS's $\mathbin{\|}$ over open terms modulo weak congruences is the negative result from \cite{AACIL21}, which shows that a class of weak congruences (including rooted weak bisimilarity) does not afford a finite, complete axiomatisation over the open terms of the minimal fragment of CCS with interleaving.

In this paper we focus on \emph{branching bisimilarity} \cite{vGW89}, which generalises strong bisimilarity to abstract away from $\tau$-steps of terms while preserving their \emph{branching structure} \cite{vGW89,vGW96}, and its \emph{rooted} version, which is a congruence with respect to CCS operators.

As a first main contribution, we show that \emph{rooted branching bisimilarity affords no finite ground-complete axiomatisation over CCS}.
To this end, we adapt the proof-theoretic technique used by Moller to prove the corresponding negative result for strong bisimilarity.
We remark that, even though the general proof strategy is a natural extension of Moller's, our proof requires a number of original, non-trivial technical results on (rooted) branching bisimilarity.
In particular, we observe that equational proofs of $\tau$-free equations might involve terms having occurrences of $\tau$ in some intermediate steps (see, e.g., page 175 of Moller’s thesis \cite{Mo89}), and our proof of the negative result for rooted branching bisimilarity will account for those uses of $\tau$, thus making our results special for the considered weak congruence.
Moreover, as an intermediate step in our proof, we establish a result of independent interest: we show that \emph{each CCS process has a unique decomposition into indecomposable processes modulo branching bisimilarity}.
A similar result was proven in \cite{Bas16}, but only for interleaving parallel composition.
Here, we extend this result to the full merge operator, including thus the possibility of communication between the parallel components.

Having established the negative result, a natural question is whether the use of the auxiliary operators from \cite{BK84b} can help us to obtain an equational basis for rooted branching bisimilarity.
Hence, as our second main contribution, we consider the language $\ccslc$, namely CCS enriched with $\textl$ and $\textc$, and \emph{we provide a complete axiomatisation for rooted branching bisimilarity over $\ccslc$ that is finite when so is the set of actions over which terms are defined}.
This axiomatisation is obtained by extending the complete axiom system for strong bisimilarity over $\ccslc$ from \cite{AFIL09} with axioms expressing the behaviour of $\textl$ and $\textc$ in the presence of $\tau$-actions (from \cite{BK85}), and with the suitable $\tau$-laws (from \cite{HM85,vGW96}) necessary to deal with rooted branching bisimilarity.
Specifically, we will see that we can express equationally the fact that left merge and communication merge distribute over choice (left merge in one argument, communication merge in both), thus allowing us to expand the behaviour of the parallel components using only a finite number of axioms, regardless of their size.
A key step in the proof of the completeness result consists in another intermediate original contribution of this work: the definition of the semantics of \emph{open} $\ccslc$ terms.

Our contribution can then be summarised as follows:
\begin{enumerate}
\item We show that every branching equivalence class of CCS processes has a unique parallel decomposition into indecomposables. 
\item We prove that rooted branching bisimilarity admits no finite equational axiomatisation over CCS.
\item We define the semantics of open $\ccslc$ terms.
\item We provide a (finite) complete axiomatisation for $\sim_\rbb$ over $\ccslc$.
\end{enumerate}


\section{Background}
\label{sec:background}

\subparagraph{Labelled transition systems}

As semantic model we consider classic \emph{labelled transition systems} \cite{Ke76}.
We assume a non-empty set of action names $\Act$, and we let $\overline{\Act}$ denote the set of action co-names, i.e., $\overline{\Act}=\{\overline{a} \mid a \in \Act\}$.
As usual, we postulate that $\overline{\overline{a}}=a$ and $a \neq \overline{a}$ for all $a \in \Act$.
Then, we define $\Acttau = \Act \cup \overline{\Act} \cup \{\tau\}$, where $\tau \not \in \Act\cup\overline{\Act}$.
Henceforth, we let $\mu,\nu,\dots$ range over actions in $\Acttau$, and $\alpha,\beta,\dots$ range over actions in $\Act \cup \overline{\Act}$.

\begin{definition}
[Labelled Transition System]
\label{Def:lts}
A {\sl labelled transition system} (LTS) is a triple $(\Proc,\Acttau,\trans[])$, where $\Proc$ is a set of \emph{processes} (or \emph{states}), $\Acttau$ is a set of {\sl actions}, and ${\trans[]} \subseteq \Proc \times \Acttau \times \Proc$ is a ({\sl labelled}) {\sl transition relation}. 
\end{definition}

As usual, we use $p \trans[\mu] p'$ in lieu of $(p,\mu,p') \in {\trans[]}$. 
For each $p \in \Proc$ and $\mu \in \Act$, we write $p \trans[\mu]$ if $p \trans[\mu] p'$ holds for some $p'$, and $p \ntrans[\mu]$ otherwise. 
The \emph{initials} of $p$ are the actions that label the outgoing transitions of $p$, that is, $\init(p) = \{\mu \in \Acttau \mid p \trans[\mu] \}$.


\subparagraph{The language CCS}

We consider the recursion, relabelling and restriction free fragment of Milner's CCS~\cite{Mi89}, which for simplicity we still call CCS, given by the following grammar:
\[
t ::=\; \nil \;|\; 
x \;|\; 
\mu.t \;|\; 
t+t \;|\; 
t \mathbin{\|} t 
\enspace ,
\]
where $x$ is a variable drawn from a countably infinite set $\Var$ disjoint from $\Acttau$, and $\mu \in \Acttau$.
We use the \emph{Structural Operational Semantics} (SOS) framework \cite{Pl81} to equip processes with an operational semantics. 
The SOS rules (or inference rules) for the CCS operators given above are reported in Table~\ref{tab:sos_rules} (symmetric rules for $+$ and $\mathbin{\|}$ are omitted).

\begin{table}[t]
\begin{gather*}
\SOSrule{}{\mu.t \trans[\mu] t} 
\qquad
\SOSrule{t \trans[\mu] t'}{t + u \trans[\mu] t'} 
\qquad 
\SOSrule{t \trans[\mu] t'}{t \mathbin{\|} u \trans[\mu] t' \mathbin{\|} u} 
\qquad
\SOSrule{t \trans[\alpha] t' \quad u \trans[\overline{\alpha}] u'}{t \mathbin{\|} u \trans[\tau] t' \mathbin{\|} u'} 
\end{gather*}
\caption{\label{tab:sos_rules} The SOS rules for CCS operators
($\mu \in \Acttau$, $\alpha \in \Act\cup\overline{\Act}$).
} 
\end{table}

We shall use the meta-variables $t,u,v,w$ to range over process terms, and write $\var(t)$ for the collection of variables occurring in the term $t$.
We use a {\em summation} $\sum_{i\in\{1,\ldots,k\}}t_i$ to abbreviate $t_1+\cdots+t_k$, where the empty sum represents $\nil$.
We call the term $t_j$ ($j \in \{1,\dots,k\}$) a \emph{summand} of $t = \sum_{i \in \{1,\dots k\}} t_i$ if it does not have $+$ as head operator.
The {\sl size} of a term $t$, denoted by $\size(t)$, is the number of operator symbols in $t$. 
A term is {\em closed} if it does not contain any variables.  
Closed terms, or {\sl processes}, will be denoted by $p,q,r$. 
Moreover, we omit trailing $\nil$'s from terms.
A {\sl (closed) substitution} is a mapping from process variables to (closed) terms.
Substitutions are extended from variables to terms, transitions, and rules in the usual way.
Note that $\sigma(t)$ is closed, if so is $\sigma$. 
We let $\sigma[x\mapsto p]$ denote the substitution that maps the variable $x$ into process $p$ and behaves like $\sigma$ on all other variables.
In particular, $[x \mapsto p]$ denotes the substitution that maps the variable $x$ into process $p$ and behaves like the identity on all other variables.

The inference rules in Table~\ref{tab:sos_rules} allow us to derive valid transitions between closed terms.
The operational semantics for our language is then modelled by the LTS whose processes are the closed terms, and whose labelled transitions are those that are provable from the SOS rules.
Henceforth, we let $\Proc$ denote the set of CCS processes.
We remark that whenever $p \trans[\mu] p'$, then $\size(p) > \size(p')$.


\subparagraph{Branching bisimilarity}

\emph{Branching bisimilarity} is a bisimulation-based behavioural equivalence that abstracts away from computation steps in processes that are deemed unobservable, while preserving their \emph{branching structure}.
The abstraction is achieved by labelling these computation steps with $\tau$, and giving $\tau$-labelled transitions a special treatment in the definition of the behavioural equivalence.
Preservation of the branching structure is mainly due to the \emph{stuttering} nature of branching bisimulation, which guarantees that the behaviour of a term is preserved in the execution of a sequence of silent steps \cite{vGW89,vGW96}.

Let $\trans[\varepsilon]$ denote the reflexive and transitive closure of the transition $\trans[\tau]$. 

\begin{definition}
[Branching bisimilarity]
\label{def:bb}
Let $(\Proc,\Acttau,\trans[])$ be a LTS.
\emph{Branching bisimilarity}, denoted by $\sim_\bb$, is the largest symmetric relation over $\Proc$ such that, whenever $p \sim_\bb q$, if $p \trans[\mu] p'$, then either:
\begin{itemize}
\item $\mu = \tau$ and $p' \sim_\bb q$, or
\item there are processes $q',q''$ such that $q \trans[\varepsilon] q'' \trans[\mu] q'$, $p \sim_\bb q''$, and $p' \sim_\bb q'$.
\end{itemize}
\end{definition}

Branching bisimilarity satisfies the \emph{stuttering property} \cite[Lemma 2.5]{vGW96}: 
\emph{Assume that $p \sim_\bb q$.
Whenever $p \trans[\tau] p_1 \trans[\tau] \dots \trans[\tau] p_n$ and $p_n \sim_\bb q$, for some $n \ge 1$, then $p_i \sim_\bb q$ for all $i=1,\dots,n-1$.}

To guarantee compositional reasoning over a process language, we require a behavioural equivalence $\sim$ to be a \emph{congruence} with respect to all language operators.
This consists in verifying whether, for all $n$-ary operators $f$
\[
\text{if } t_i \sim t_i' \text{ for all } i = 1,\dots,n, \text{ then } f(t_1,\dots,t_n) \sim f(t_1',\dots,t_n').
\]
It is well known that branching bisimilarity is an equivalence relation \cite{vGW96,Ba96}.
Moreover, action prefixing and parallel composition satisfy the \emph{simple BB cool rule format} \cite{vG11} and hence $\sim_\bb$ is compositional with respect to those operators. 
However, $\sim_\bb$ is not a congruence with respect to nondeterministic choice.
To remedy this inconvenience, the \emph{root condition} is introduced:
\emph{rooted branching bisimilarity} behaves like strong bisimilarity on the initial transitions, and like branching bisimilarity on subsequent transitions.

\begin{definition}
[Rooted branching bisimilarity]
\label{def:rbb}
\emph{Rooted branching bisimilarity}, denoted by $\sim_\rbb$, is the symmetric relation over $\Proc$ such that, whenever $p \sim_\rbb q$, if $p \trans[\mu] p'$, then there is a process $q'$ such that $q \trans[\mu] q'$ and $p' \sim_\bb q'$.
\end{definition}

It is well known that rooted branching bisimilarity is an equivalence relation \cite{vGW96,Ba96}, and that $\sim_\rbb$ is a congruence over CCS (see, e.g., \cite{vG11}).


\subparagraph{Equational Logic}

An \emph{axiom system} $\E$ is a collection of (\emph{process}) \emph{equations} $t \approx u$ over the considered language, thus CCS in this paper.
An equation $t \approx u$ is \emph{derivable} from an axiom system $\E$, notation $\E \vdash t \approx u$, if there is an \emph{equational proof} for it from $\E$, namely if $t \approx u$ can be inferred from the axioms in $\E$ using the \emph{rules} of \emph{equational logic}.
The rules over CCS are reported in Table~\ref{tab:equational_logic}. 

\begin{table}[t]
\begin{gather*}
\scalebox{0.9}{($e_1$)}\; t \approx t 
\qquad
\scalebox{0.9}{($e_2$)}\; \frac{t \approx u}{u \approx t} 
\qquad
\scalebox{0.9}{($e_3$)}\; \frac{{t \approx u ~~ u \approx v}}{{t \approx v}} 
\qquad
\scalebox{0.9}{($e_4$)}\; \frac{{t \approx u}}{{\sigma(t) \approx \sigma(u)}} \\[.2cm]
\scalebox{0.9}{($e_4$)}\; \frac{t \approx u}{\mu. t \approx \mu. u}
\qquad
\scalebox{0.9}{($e_5$)}\; \frac{t \approx  u~~ t' \approx u'}{t+t' \approx u+u'}
\qquad
\scalebox{0.9}{($e_6$)}\; \frac{t \approx  u~~ t' \approx u'}{t\mathbin{\|} t' \approx u\mathbin{\|} u'}
\enspace .
\end{gather*}
\caption{\label{tab:equational_logic} The rules of equational logic} 
\end{table}

We assume, without loss of generality, that the substitution rule is only applied on equations $(t \approx u) \in \E$.  
In this case, $\sigma(t) \approx \sigma(u)$ is called a {\em substitution instance} of an axiom in $\E$.
Moreover, by postulating that for each axiom in $\E$ also its symmetric counterpart is present in $\E$, one may assume that the symmetry rule is never used in equational proofs.  

\begin{table*}[t]
\setlength{\tabcolsep}{15pt}
\centering
\begin{tabular}{llll}
\multicolumn{3}{l}{Some axioms for bisimilarity over CCS:} \\[.2cm]
\multicolumn{2}{l}{$\scalebox{0.85}{A0}\quad x + \nil \approx x$}
&
\multicolumn{2}{l}{$\scalebox{0.85}{P0}\quad x \mathbin{\|} \nil \approx x$}
\\
\multicolumn{2}{l}{$\scalebox{0.85}{A1}\quad x+y \approx y+x$}
&
\multicolumn{2}{l}{$\scalebox{0.85}{P1}\quad x \mathbin{\|} y \approx y \mathbin{\|} x$}\\
\multicolumn{2}{l}{$\scalebox{0.85}{A2}\quad (x+y)+z \approx x+(y+z)$} 
&
\multicolumn{2}{l}{$\scalebox{0.85}{P2}\quad (x \mathbin{\|} y) \mathbin{\|} z \approx x \mathbin{\|} (y \mathbin{\|} z)$}\\
\multicolumn{2}{l}{$\scalebox{0.85}{A3}\quad x + x \approx x$}
\\
\hline
\hline\\
\multicolumn{3}{l}{Additional axioms for rooted branching bisimilarity over CCS:} \\[.2cm]
\multicolumn{2}{l}{$\scalebox{0.85}{TB}\quad \mu (\tau(x+y) + y) \approx \mu (x +y)$}
&
\multicolumn{2}{l}{$\scalebox{0.85}{T1}\quad \mu \tau x \approx \mu x$} 
\end{tabular}
\caption{\label{tab:axioms_b} Some axioms for rooted branching bisimilarity.}
\end{table*}

We are interested in equations that are valid modulo some congruence relation $\sim$ over terms.
The equation $t \approx u$ is said to be \emph{sound} modulo $\sim$ if $\sigma(t) \sim \sigma(u)$ for all closed substitutions $\sigma$.
For simplicity, if $t \approx u$ is sound, then we write $t \sim u$.
An axiom system is \emph{sound} modulo $\sim$ if, and only if, all of its equations are sound modulo $\sim$. 
Conversely, we say that $\E$ is \emph{complete} modulo $\sim$ if $t \sim u$ implies $\E \vdash t \approx u$ for all terms $t,u$.
If we restrict ourselves to consider only equations over closed terms then $\E$ is said to be \emph{ground-complete} modulo $\sim$.
We say that $\sim$ has a finite, (ground) complete axiomatisation, if there is a finite axiom system $\E$ that is sound and (ground) complete for $\sim$.

Henceforth, we exploit the associativity and commutativity of $+$ and $\mathbin{\|}$ modulo the relevant behavioural equivalences.
The symbol $=$ will then denote equality modulo A1-A2 and P1-P2 in Table~\ref{tab:axioms_b}.


\section{The main results}
\label{sec:main_results}

Our aim is to study the axiomatisability of rooted branching bisimilarity over CCS.
Our investigations produced, as main outcomes, a negative result (Theorem~\ref{thm:rbb_negative}) and a positive one (Theorem~\ref{thm:rbb_complete_ccslc}).
In detail, in the first part of the paper we prove the following theorem:

\begin{restatable}{theorem}{thmrbbnegative}
\label{thm:rbb_negative}
Rooted branching bisimilarity has no finite equational ground-complete axiomatisation over CCS.
\end{restatable}

Given the negative result, it is natural to wonder whether an equational basis for rooted branching bisimilarity can be obtained if we enrich CCS with some auxiliary operators.
Considering the similarities between $\sim_\rbb$ and strong bisimilarity, the principal candidates for this role are the left merge $\textl$ and the communication merge $\textc$ from \cite{BK84b}.
Indeed, we show that if we add those two operators to the syntax of CCS, then we can obtain a complete axiomatisation of rooted branching bisimilarity over the new language, denoted by $\ccslc$.
The desired equational basis is given by the axiom system $\E_\rbb$, which is presented fully in Table~\ref{tab:axioms_rbb} in Section~\ref{sec:completeness}.
$\E_\rbb$ is an extension of the complete axiom system for strong bisimilarity over $\ccslc$ from \cite{AFIL09} with axioms expressing the behaviour of left merge and communication merge in the presence of $\tau$-actions (taken from \cite{BK85}), and with the suitable $\tau$-laws necessary to deal with rooted branching bisimilarity (taken from \cite{HM85,vGW96}).

Formally, our second main contribution consists in a proof of the following theorem:

\begin{restatable}[Completeness]{theorem}{thmrbbcompleteccslc}
\label{thm:rbb_complete_ccslc}
Let $t,u$ be $\ccslc$ terms.
If $t \sim_\rbb u$, then $\E_\rbb \vdash t \approx u$.
\end{restatable}

We will also argue that this axiomatisation is finite when so is the set of actions.
Hence, when $\Act$ is finite, $\ccslc$ modulo $\sim_\rbb$ is finitely based, unlike CCS.

Considering the amount of technical results that we will need to fulfil our objectives, we devote Section~\ref{sec:roadmap} to a presentation of the proof strategy that we will apply to obtain Theorem~\ref{thm:rbb_negative}.
Sections~\ref{sec:decomposition}--\ref{sec:negative_result} then present the formalisation of the ideas discussed in that section.
Similarly, in Section~\ref{sec:completeness_roadmap} we give a high-level description of the approach that we will follow to prove Theorem~\ref{thm:rbb_complete_ccslc}.
The technical development of the proof is then reported in Sections~\ref{sec:configurations_bis}--\ref{sec:completeness}.

All the complementary results needed to prove the two theorems, are reported in the Appendix.


\section{Proof strategy for Theorem~\ref{thm:rbb_negative}}
\label{sec:roadmap}

In this section we present the proof strategy we will apply to obtain Theorem~\ref{thm:rbb_negative}.

Our proof follows the so-called \emph{proof-theoretic approach} to non-finite-axiomatisability results, whose use in the field of process algebra stems from \cite{Mo89,Mo90,Mo90a}, where Moller proved that CCS modulo strong bisimilarity is not finitely based.
In the proof-theoretic approach, the idea is to identify a specific property of terms parametric in $n \ge 0$, say $\mathbb{P}_n$, and show that if $\E$ is an arbitrary finite axiom system that is sound with respect to $\sim_\rbb$, then $\mathbb{P}_n$ is preserved by provability from $\mathcal{E}$ when $n$ is ``\emph{large enough}''.
Next, we exhibit an infinite family of equations $\{\e_n \mid n \ge 0\}$ over closed terms that are all sound modulo $\sim_\rbb$, but are such that only one side of $\e_n$ satisfies $\mathbb{P}_n$, for each $n \ge 0$. 
In particular, this implies that whenever $n$ is ``large enough'' then the sound equation $\e_n$ cannot be proved from $\E$. 
Since $\E$ is an arbitrary finite sound axiom system, it follows that no finite sound axiomatisation can prove all the equations in the family $\{\e_n \mid n \ge 0\}$ and therefore no finite sound axiomatisation is ground complete for CCS modulo modulo $\sim_\rbb$. 


\subparagraph{The choice of $\mathbb{P}_n$ and the family of equations}

In \cite{Mo89,Mo90,Mo90a} Moller applied the proof method sketched above to prove that strong bisimilarity has no finite, complete axiomatisation over CCS.
The key idea underlying this result is that, since $\mathbin{\|}$ does not distribute over $+$ in either of its arguments modulo strong bisimilarity, then no finite, sound axiom system can ``\emph{expand}'' the initial behaviour of process $a \mathbin{\|} \sum_{i = 1}^n a^i$ (where $a^i = aa^{i-1}$ for each $i = 1,\dots,n$, with $a^0 = \nil$) when $n$ is large. 

Since, by definition, rooted branching bisimilarity behaves exactly like strong bisimilarity on the first step, and parallel composition does not distribute over choice in either of its arguments, modulo $\sim_\rbb$, it is natural to exploit a similar strategy to prove Theorem~\ref{thm:rbb_negative}.
In detail, we will consider, for each $n \ge 2$, the process $p_n = \sum_{i = 2}^n aa^{\le i}$, where $a^{\le i} = \sum_{j = 1}^i a^j$ for each $i = 2,\dots,n$.
Then, for each $n \ge 2$, the property $\mathbb{P}_n$ will consist in having a summand rooted branching bisimilar to the process $a \mathbin{\|} p_n$, and we will show that, when $n$ is large enough, $\mathbb{P}_n$ is an invariant under provability from an arbitrary finite, sound axiom system (Theorem~\ref{thm:rbb_preserves_property}).
Hence, the sound equation
$
\e_n \;\colon\; a \mathbin{\|} p_n \,\approx\, ap_n + \sum_{i = 2}^n a(a \mathbin{\|} a^{\le i})
$
cannot be derived from $\E$ because its right-hand side has no summand that is rooted branching bisimilar to $a \mathbin{\|} p_n$, unlike its left-hand side.
Therefore no finite sound axiom system can prove the infinite family of equations $\{\e_n \mid n \ge 2\}$, yielding the desired negative result.

In proving that $\mathbb{P}_n$ is invariant under provability, one pivotal ingredient will be the fact that processes $p_n$ and $a^{\le i}$, for $n \ge 2$ and $i \in \{2,\dots,n\}$, are \emph{indecomposable}.
The existence of a unique parallel decomposition into indecomposable processes modulo \emph{branching bisimilarity} over CCS with \emph{interleaving parallel composition} was studied in \cite{Bas16}.
In Section~\ref{sec:indecomposables}, we extend the result from \cite{Bas16} to the full merge operator, thus including communication (Proposition~\ref{prop:unique_par}). 


\subparagraph{The choice of $n$}

The choice of a sufficiently large $n$ plays a crucial
role in proving that $\mathbb{P}_n$ is an invariant under provability from a finite, sound axiom system $\E$ (Theorem~\ref{thm:rbb_preserves_property}).
The key step in that proof deals with the case in which $p \approx q$ is a substitution instance of an equation in $\E$ (Proposition~\ref{prop:rbb_substitution_case}), i.e., $p = \sigma(t)$, $q = \sigma(u)$, and $t \approx u \in \E$ for some terms $t, u$ and closed substitution $\sigma$. 
In this case, assuming that $n > \size(t)$, we can prove that if $p = \sigma(t)$ satisfies $\mathbb{P}_n$ then this is due to the behaviour of $\sigma(x)$ for some variable $x$.
In order to reach this conclusion, in Section~\ref{sec:decomposition}, we study how the behaviour of closed instances of terms may depend on the behaviour of the closed instances of variables occurring in them. 
Moreover, we show that if $t \approx u$ is sound modulo rooted branching bisimilarity and $x$ occurs in $t$, then it occurs also in $u$ (Proposition~\ref{prop:same_var}).
Hence, we can infer that $\sigma(x)$ triggers in $\sigma(u)$ the same behaviour that it induced in $\sigma(t)$, and thus that $q = \sigma(u)$ satisfies $\mathbb{P}_n$. 
All the additional properties of process $a \mathbin{\|} p_n$ used to achieve this conclusion are presented in Appendix~\ref{sec:preliminaries}.


\section{Decomposing the semantics of terms}
\label{sec:decomposition}

In the proofs to follow, we shall sometimes need to establish a correspondence between the behaviour of open terms and that of their closed instances.
In detail, we are interested in the correspondence between a transition $\sigma(t) \trans[\mu] p$, for some term $t$, closed substitution $\sigma$, action $\mu$, and process $p$, and the behaviour of $t$ and that of $\sigma(x)$, for each variable $x$ occurring in $t$.
The simplest case is a direct application of the operational semantics in Table~\ref{tab:sos_rules}.

\begin{lemma}
\label{lem:substitution}
For all terms $t,t'$, substitution $\sigma$, and $\mu \in \Acttau$, if $t \trans[\mu] t'$ then $\sigma(t) \trans[\mu] \sigma(t')$.
\end{lemma}

Let us focus now on the role of variables.
A transition $\sigma(t) \trans[\mu] p$ may also derive from the initial behaviour of some closed term $\sigma(x)$, provided that the collection of initial moves of $\sigma(t)$ depends, in some formal sense, on that of the closed term substituted for the variable $x$. 
In this case, we say that $x$ \emph{triggers the behaviour} of $t$.
To fully describe this situation, we introduce an auxiliary transition relation over open terms.
The notion of {\sl configuration} over terms, which stems from \cite{AFIN06}, will play an important role in their definition.

The presence of communication in CCS entails a complex definition of the semantics of configurations.
In particular, it is necessary to introduce a fresh set of variables $\Var_{\Acttau} = \{x_\mu \mid x \in \Var, \mu \in \Acttau\}$, disjoint from $\Var$, and terms.
Intuitively, the symbol $x_\mu$ denotes that the closed term substituted for an occurrence of variable $x$ has begun its execution (expressed in terms of a $\mu$-action), and it contributes thus to triggering the behaviour of the term in which $x$ occurs (see Example~\ref{ex:variabili_vi_odio} below).
Moreover, we also need to introduce special labels and subscripts for the auxiliary transitions over configurations, which will be of the form $c \trans[\ell]_\rho c'$.
Briefly, the label $\ell$ is used to keep track of the variables that trigger the transition $c \trans[\ell]_\rho c'$.
The subscript $\rho$, instead, will allow us to correctly define the semantics of communication: it will allow us to distinguish a $\tau$-action directly performed by (the term substituted for) a variable $x$ (transition $c \trans[(x)]_\tau c'$, with $\rho = \tau$), from a $\tau$-action resulting from the communication of $x$ with a subterm of the configuration (transition $c \trans[(x)]_{\alpha,\tau} c'$, with $\rho = \alpha,\tau$, where $\alpha$ is the action performed by the term substituted for $x$).

\emph{CCS configurations} are defined over the set of variables $\Var_{\Acttau}$ and CCS terms.

\begin{definition}
The collection of {\sl CCS configurations}, denoted by $\C$, is given by: 
\[
c :: = \; x_\mu \quad | \quad 
t \quad | \quad
c \mathbin{\|} c 
\enspace ,
\quad
\text{ where } t \text{ is a term, and } x_\mu \in \Var_{\Acttau}. 
\]
\end{definition}

\begin{table}[t]
\begin{gather*}
\scalebox{0.9}{$(a_1)$}\; \SOSrule{}{x \trans[(x)]_{\mu} x_{\mu}}
\qquad
\scalebox{0.9}{$(a_2)$}\; \SOSrule{t \trans[\ell]_\rho c}
{t+u \trans[\ell]_\rho c}
\qquad
\scalebox{0.9}{$(a_3)$}\; \SOSrule{t \trans[\ell]_\rho c}
{t \mathbin{\|} u \trans[\ell]_{\rho} c \mathbin{\|} u}
\\[.2cm]
\scalebox{0.9}{$(a_4)$}\; \SOSrule{t \trans[(x)]_\alpha c \quad u \trans[(y)]_{\overline{\alpha}} c'}
{t \mathbin{\|} u \trans[(x,y)]_{\tau} c \mathbin{\|} c'}
\qquad
\scalebox{0.9}{$(a_5)$}\; \SOSrule{t \trans[(x)]_\alpha c \quad u \trans[\overline{\alpha}] u'}
{t \mathbin{\|} u \trans[(x)]_{\,\alpha,\tau} c \mathbin{\|} u'}
\qquad
\scalebox{0.9}{$(a_6)$}\; \SOSrule{t \trans[\alpha] t' \quad u \trans[(x)]_{\overline{\alpha}} c}
{t \mathbin{\|} u \trans[(x)]_{\,\overline{\alpha},\tau} t' \mathbin{\|} c}
\end{gather*}
\caption{\label{tab:ell_rules} Inference rules for the transition relation $\trans[\ell]_\rho$ ($\mu \in \Acttau$, $\alpha \in \Act\cup\overline{\Act}$).}
\end{table}

The auxiliary transitions of the form $\trans[\ell]_{\rho}$ are then formally defined via the inference rules in Table~\ref{tab:ell_rules}, where we omitted the symmetric rules to ($a_2$), ($a_4$), ($a_5$) and ($a_6$).
We have that $\rho \in \Acttau \cup ((\Act \cup \overline{\Act}) \times \{\tau\})$, whereas the label $\ell$ can be either of the form $(x)$ or $(x,y)$, for some variables $x,y \in \Var$.
Given a variable $x$ and a label $\ell$, we write $x \in \ell$ if $x$ occurs in $\ell$.

The distinguished variables $x_\mu$ allow us to keep track of which variable and action trigger the behaviour of the term, and they also allow us to present substitutions in an intuitive fashion.
As explained in the following example, it is precisely because of substitutions (and communication) that we need to make the action $\mu$ explicit in $x_\mu$.

\begin{example}
\label{ex:variabili_vi_odio}
Let $x \in \Var$ and consider the term $x \mathbin{\|} x$.
By rules ($a_1$) and ($a_4$) in Table~\ref{tab:ell_rules}, we have that $x \mathbin{\|} x \trans[(x,x)]_\tau x_{\alpha} \mathbin{\|} x_{\overline{\alpha}}$ because $x \trans[(x)]_\alpha x_{\alpha}$ and $x \trans[(x)]_{\overline{\alpha}} x_{\overline{\alpha}}$.
Hence, given any substitution $\sigma$ such that $\sigma(x) \trans[\alpha] p_1$ and $\sigma(x) \trans[\overline{\alpha}] p_2$, for some terms $p_1,p_2$, we want to be able to correctly infer that $\sigma(x) \mathbin{\|} \sigma(x) \trans[\tau] p_1 \mathbin{\|} p_2$.
Since the two occurrences of $x$, $x_{\alpha}$ and $x_{\overline{\alpha}}$, can be distinguished by the subscripts, the substitution $\sigma[x_\alpha \mapsto p_1, x_{\overline{\alpha}} \mapsto p_2](x_\alpha \mathbin{\|} x_{\overline{\alpha}}) = p_1 \mathbin{\|} p_2$ is well-defined.
Without the subscripts, it would not have been possible to correctly define the substitution $\sigma$ on the configuration $c$ that is the target of $x \mathbin{\|} x \trans[(x,x)]_\tau c$.
\end{example}

\begin{lemma}
\label{lem:var_to_term}
Let $t$ be term and $\sigma$ be a closed substitution.
Let $x,y \in \Var$.
\begin{enumerate}
\item For any $\mu \in \Acttau$, if $\sigma(x) \trans[\mu] p$, for some process $p$, and $t \trans[(x)]_\mu c$, for some configuration $c \in \C$, then $\sigma(t) \trans[\mu] \sigma[x_\mu \mapsto p](c)$.
\item \label{item:x_u_to_t} For any $\alpha \in \Act\cup\overline{\Act}$, if $\sigma(x) \trans[\alpha] p$, for some process $p$, and $t \trans[(x)]_{\alpha,\tau} c$, for some configuration $c\in \C$, then $\sigma(t) \trans[\tau] \sigma[x_\alpha \mapsto p](c)$.
\item For any $\alpha \in \Act\cup\overline{\Act}$, if $\sigma(x) \trans[\alpha] p_x$, $\sigma(y) \trans[\overline{\alpha}] p_y$, for some processes $p_x,p_y$, and $t \trans[(x,y)]_{\tau} c\in \C$, for some configuration $c$, then $\sigma(t) \trans[\tau] \sigma[x_\alpha \mapsto p_x, y_{\overline{\alpha}} \mapsto p_y](c)$.
\end{enumerate}
\end{lemma}

Lemma~\ref{lem:var_to_term} shows how the auxiliary transitions can be used to derive the behaviour of $\sigma(t)$ from those of the variables in $t$.
We are now interested in analysing the converse situation: we show how a transition $\sigma(t) \trans[\mu] p$ can stem from transitions of the term $t$ and of the process $\sigma(x)$, for $x \in \var(t)$.
We limit ourselves to present the case of silent actions $\sigma(t) \trans[\tau] p$ as it requires a detailed analysis.
The case of transitions labelled with observable actions is simpler and can be found as Lemma~\ref{lem:closed2open_alpha} in Appendix~\ref{app:decomposition}.

\begin{restatable}{lemma}{lemclosedopentau}
\label{lem:closed2open_tau}
Let $t$ be a term, $\sigma$ be a closed substitution, and $p$ be a process.
If $\sigma(t) \trans[\tau] p$, then one of the following holds:
\begin{enumerate}
\item \label{lem:c2o_prefix}
There is a term $t'$ s.t.\ $t \trans[\tau] t'$ and $\sigma(t') = p$.
\item \label{lem:c2o_x}
There are a variable $x$, a process $q$, and a configuration $c$ s.t.\ $\sigma(x) \trans[\tau] q$, $t \trans[(x)]_\tau c$, and $\sigma[x_\tau \mapsto q](c) = p$.
\item \label{lem:c2o_xu}
There are a variable $x$, a process $q$, and a configuration $c$ s.t., for some $\alpha \in \Act\cup\overline{\Act}$, $\sigma(x) \trans[\alpha] q$, $t \trans[(x)]_{\alpha,\tau} c$, and $\sigma[x_\alpha \mapsto q](c) = p$.
\item \label{lem:c2o_xy}
There are variables $x,y$, processes $q_x,q_y$ and a configuration $c$ s.t., for some $\alpha \in \Act\cup\overline{\Act}$, $\sigma(x) \trans[\alpha] q_x$, $\sigma(y) \trans[\overline{\alpha}] q_y$, $t \trans[(x,y)]_\tau c$, and $\sigma[x_\alpha \mapsto q_x, y_{\overline{\alpha}} \mapsto q_y](c) = p$.
\end{enumerate}
\end{restatable}


\section{Unique parallel decomposition}
\label{sec:indecomposables}

As explained in Section~\ref{sec:roadmap}, our approach for establishing that $\mathbb{P}_n$ is invariant under equational proofs relies on processes having a unique parallel decomposition modulo $\sim_\bb$.

\begin{definition}
[Parallel decomposition modulo $\sim_\bb$]
\label{def:indecomposable}
A process $p$ is \emph{indecomposable} if $p\not\sim_\bb\nil$ and $p \sim_\bb p_1 \mathbin{\|} p_2$ implies $p_1 \sim_\bb \nil$ or $p_2 \sim_\bb \nil$, for all processes $p_1$ and $p_2$. A \emph{parallel decomposition} of a process $p$ is a finite multiset $\lbag p_1,\dots,p_k \rbag$ of indecomposable processes $p_1,\dots,p_k$ such that $p\sim_\bb p_1\mathbin{\|}\cdots\mathbin{\|}p_k$. We say that $p$ has a \emph{unique parallel decomposition} if $p$ has a parallel decomposition $\lbag p_1,\dots,p_k\rbag$ and for every other parallel decomposition $\lbag p_1',\dots,p_\ell'\rbag$ of $p$ there exists a bijection $f:\{1,\dots,k\}\rightarrow\{1,\dots,\ell\}$ such that $p_i\sim_\bb p_{f(i)}'$ for all $1\leq i \leq k$.
\end{definition}

To prove that processes have a unique parallel decomposition we shall exploit a general result stating that a partial commutative monoid has unique decomposition if it can be endowed with a \emph{weak decomposition order} that satisfies \emph{power cancellation} \cite{Bas16}; we shall define and explain the notions below. 
Note that, in view of axioms P0--P2, which are (also) sound modulo $\sim_\bb$, the set of processes $\Proc$ modulo $\sim_\bb$ is a commutative monoid with respect to the binary operation naturally induced by $\mathbin{\|}$ on $\sim_{\bb}$-equivalence classes and the $\sim_{\bb}$-equivalence class of $\nil$ as identity element. 
We permit ourselves a minor abuse in notation and use $\rightarrow$ to (also) denote the binary relation $\{(p,q)\mid \exists \mu.\ p\trans[\mu]q\}$, and proceed to argue that $\rightarrow$ induces a weak decomposition order satisfying power cancellation on the commutative monoid of processes modulo $\sim_\bb$.

Given any process $p$ and $n \ge 1$, let $p^n$ denote the $n$-fold parallel composition $p \mathbin{\|} p^{n-1}$, with $p^0 = \nil$.
We first state some properties of the reflexive-transitive closure $\rightarrow^{*}$ of $\rightarrow$:

\begin{restatable}{proposition}{rttransprops} 
\label{prop:rttransprops}
  The relation $\rightarrow^{*}$ is an inversely well-founded partial order on processes satisfying the following properties:
  \begin{enumerate}
      \item \label{item:rttranspropsleast}
        For every process $p$ there exists a process $p'$ such that $p\mathrel{\rightarrow^{*}}p'\sim_\bb \nil$.
      \item \label{item:rttranspropscompatible}
        For all processes $p$, $p'$ and $q$, if $p\rightarrow^{*}p'$, then $p\mathbin{\|}q\rightarrow^{*}p'\mathbin{\|}q$ and $q\mathbin{\|}p\rightarrow^{*}q\mathbin{\|}p'$.
      \item \label{item:rttranspropsprecompositional}
        For all processes $p$, $q$ and $r$, if $p\mathbin{\|}q\rightarrow^{*}r$, then there exist $p'$ and $q'$ such that $p\rightarrow^{*}p'$, $q \rightarrow^{*}q'$ and $r=p'\mathbin{\|}q'$.
      \item \label{item:rttranspropsArchimedean}
        For all processes $p$ and $q$, if $p\rightarrow^{*}q^n$
        for all $n\in\mathbb{N}$, then $q\sim_{\bb}\nil$.
  \end{enumerate}
\end{restatable}

The following lemma is a direct consequence of the definition of branching bisimilarity.
\begin{lemma}
  For all processes $p$, $p'$ and $q$, if $p\sim_{\bb}q$ and $p\rightarrow^{*}p'$, then there exists $q'$ such that $q\rightarrow^{*}q'$ and $p'\sim_{\bb}q'$.
\end{lemma}

By this lemma we can define a binary relation $\preceq$ on $\Proc/{\sim_{\bb}}$, the set of $\sim_{\bb}$-equivalence classes of processes, by stating that $[p]_{\sim_{\bb}}\preceq [q]_{\sim_{\bb}}$ if, and only if, there exists $p'\in [p]_{\sim_{\bb}}$ such that $q\rightarrow^{*}p'$ (here $[p]_{\sim_{\bb}}$ and $[q]_{\sim_{\bb}}$ denote the $\sim_{\bb}$-equivalence classes of $p$ and $q$, respectively). 
The following result is then a straightforward corollary of Proposition~\ref{prop:rttransprops}.

\begin{corollary}
The relation $\preceq$ is a weak decomposition order on $\Proc/{\sim_{\bb}}$, namely:
\begin{enumerate}
\item it is well-founded, i.e., every non-empty subset of $\Proc/{\sim_{\bb}}$ has a $\preceq$-minimal element;
\item the identity element $[\nil]_{\sim_{\bb}}$ of $\Proc/{\sim_{\bb}}$ is the least element of $\Proc/{\sim_{\bb}}$ with respect to $\preceq$, i.e., $[\nil]_{\sim_\bb} \preceq [p]_{\sim_\bb}$ for all $p \in \Proc$;
\item it is compatible, i.e., for all $p, q, r \in \Proc$ if $[p]_{\sim_\bb} \preceq [q]_{\sim_\bb}$, then $[p \mathbin{\|} r]_{\sim_\bb} \preceq [q \mathbin{\|} r]_{\sim_\bb}$;
\item it is precompositional, i.e., for all $p, q, r \in \Proc$ we have that $[p]_{\sim_\bb} \preceq [q \mathbin{\|} r]_{\sim_\bb}$ implies $[p]_{\sim_\bb} = [q' \mathbin{\|} r']_{\sim_\bb}$ for some $[q']_{\sim_\bb} \preceq [q]_{\sim_\bb}$ and $[r']_{\sim_\bb} \preceq [r]_{\sim_\bb}$; and
\item it is Archimedean, i.e., for all $p, q \in \Proc$ we have that $[p^n]_{\sim_\bb} \preceq [q]_{\sim_\bb}$ for all $n \in \N$ implies that $[p]_{\sim_\bb} = [\nil]_{\sim_\bb}$.
\end{enumerate}
\end{corollary}

According to \cite[Theorem 34]{Bas16} it now remains to prove that $\preceq$ satisfies power cancellation. 
The weak decomposition order $\preceq$ on the commutative monoid of processes modulo $\sim_{\bb}$ satisfies \emph{power cancellation} if for every indecomposable process $p$ and for all processes $q$ and $r$ such that $[p]_{\sim_{\bb}}\not\prec [q]_{\sim_{\bb}},[r]_{\sim_{\bb}}$, for all $k\in\mathbb{N}$, we have that $[p^k\mathbin{\|}q]_{\sim_{\bb}} = [p^k\mathbin{\|}r]_{\sim_{\bb}}$ implies $[q]_{\sim_\bb} = [r]_{\sim_\bb}$.

\begin{restatable}{proposition}{propbbcancellation}
\label{prop:bb_cancellation}
  The weak decomposition order $\preceq$ on the commutative monoid of processes modulo $\sim_{\bb}$ satisfies power cancellation.
\end{restatable}

We have now established that $\preceq$ is a weak decomposition order on the commutative monoid of processes modulo $\sim_{\bb}$ that satisfies power cancellation. Thus, with an application of \cite[Theorem 34]{Bas16} we get the following unique parallel decomposition result.

\begin{proposition}
\label{prop:unique_par}
Every process in $\Proc{}$ has a unique parallel decomposition.
\end{proposition}

In what follows, we shall make use of the following direct consequence of Proposition~\ref{prop:unique_par}.

\begin{corollary}
\label{cor:cancellation}
If $p \mathbin{\|} r \sim_\bb q \mathbin{\|} r$, then $p \sim_\bb q$.
\end{corollary}


\section{Nonexistence of a finite axiomatisation}
\label{sec:negative_result}

We devote this section to proving Theorem~\ref{thm:rbb_negative}.
Following the strategy sketched in Section~\ref{sec:roadmap}, we introduce a particular family of equations on which we will build our negative result:
\begin{align*}
& p_n  \;=\;  \sum_{i=2}^{n} a a^{\le i} & (n \ge 2)\phantom{.} \\
& \e_n \colon \quad a \mathbin{\|} p_n  \;\approx\;  a p_n + \sum_{i=2}^{n} a (a \mathbin{\|} a^{\le i}) & (n \ge 2).
\end{align*}
It is easy to check that each equation $\e_n$, for $n \ge 2$, is sound modulo rooted branching bisimilarity (as, in particular, it is sound modulo strong bisimilarity).

In order to prove Theorem~\ref{thm:rbb_negative}, we proceed to show that no finite collection of equations over CCS that are sound modulo rooted branching bisimilarity can prove all of the equations $\e_n$ ($n \ge 2$) from the family given above.
Formally, for each $n \ge 2$, we consider the property
$\mathbb{P}_n$: \textit{having a summand rooted branching bisimilar to } $a \mathbin{\|} p_n$.
Then, we prove the following: 

\begin{restatable}{theorem}{thmrbbpreservesproperty}
\label{thm:rbb_preserves_property}
Let $\E$ be a finite axiom system over CCS that is sound modulo $\sim_\rbb$, let 
$n$ be larger than the size of each term in the equations in $\E$, and let $p,q$ be closed terms such that $p,q \sim_\rbb a \mathbin{\|} p_n$.
If $\E \vdash p \approx q$ and $p$ satisfies $\mathbb{P}_n$ then so does $q$.
\end{restatable}

The crucial step in the proof of Theorem~\ref{thm:rbb_preserves_property} is delivered by the proposition below, which ensures that the property $\mathbb{P}_n$ ($n \ge 2$) is preserved by the closure under substitutions of equations in a finite, sound axiom system.
Proposition~\ref{prop:rbb_substitution_case} is proved by means of the technical results provided so far, the ones in Appendix~\ref{sec:preliminaries}, and the notion of $\nil$-\emph{factor} of a term:

\begin{definition}
\label{Lem:nil-factors}
We say that a term $t$ has a $\nil$ \emph{factor} if it contains a subterm of the form  $t' \mathbin{\|} t''$, and either $t' \sim_\rbb \nil$ or $t'' \sim_\rbb \nil$.
\end{definition}

\begin{restatable}{proposition}{proprbbsubstitutioncase}
\label{prop:rbb_substitution_case}
Let $t \approx u$ be an equation over CCS terms that is sound modulo $\sim_\rbb$.
Let $\sigma$ be a closed substitution with $p = \sigma(t)$ and $q = \sigma(u)$.
Suppose that $p$ and $q$ have neither $\nil$ summands nor $\nil$ factors, and $p,q \sim_\rbb a \mathbin{\|} p_n$ for some $n$ larger than the sizes of $t$ and $u$.
If $p$ satisfies $\mathbb{P}_n$, then so does $q$.
\end{restatable}

Theorem~\ref{thm:rbb_preserves_property} shows the property $\mathbb{P}_n$ to be an invariant under provability from finite sound axiom systems.
As the left-hand side of equation $\e_n$, i.e., the term $a \mathbin{\|} p_n$, satisfies $\mathbb{P}_n$, whilst the right-hand side, i.e., the term $a p_n + \sum_{i = 2}^{n} a (a \mathbin{\|} a^{\le i})$, does not, we can conclude that the infinite collection of equations $\e_n$ ($n \ge 2$) cannot be derived from any finite, sound axiom system.
Hence, Theorem~\ref{thm:rbb_negative} follows.


\section{Towards a positive result}
\label{sec:completeness_roadmap}

We now proceed to study the role of the auxiliary operators \emph{left merge} ($\textl$) and \emph{communication merge} ($\,\textc\,$) from \cite{BK84b} in the axiomatisation of parallel composition modulo $\sim_\rbb$.
We will show that by adding them to CCS we can obtain a complete axiomatisation of rooted branching bisimilarity over the new language.
This axiomatisation is finite if so is $\Acttau$.

We denote the language obtained by enriching CCS with $\textl$ and $\textc$ by $\ccslc$:
\begin{equation}
\tag{$\ccslc$}
t ::=\; \nil \;|\; 
x \;|\; 
\mu.t \;|\; 
t+t \;|\; 
t \mathbin{\|} t \;|\;
t \lmerge t \;|\;
t \cmerge t
\enspace ,
\end{equation}
where $x \in \Var$, and $\mu \in \Acttau$.
The SOS rules for the $\ccslc$ operators are given by the rules in Table~\ref{tab:sos_rules} plus those reported in Table~\ref{tab:sos_rules_ccslc}.

\begin{table}[t]
\begin{gather*}
\SOSrule{t \trans[\mu] t'}{t \lmerge u \trans[\mu] t' \mathbin{\|} u} 
\qquad
\SOSrule{t \trans[\alpha] t' \quad u \trans[\overline{\alpha}] u'}{t \cmerge u \trans[\tau] t' \mathbin{\|} u'}
\end{gather*}
\caption{\label{tab:sos_rules_ccslc} Additional SOS rules for $\ccslc$ operators ($\mu \in \Acttau$, $\alpha \in \Act\cup\overline{\Act}$).
} 
\end{table}

To obtain the desired completeness result, we consider the axiom system $\E_\rbb$ (see Table~\ref{tab:axioms_rbb} in Section~\ref{sec:completeness}), obtained by extending the complete axiom system for strong bisimilarity over $\ccslc$ from \cite{AFIL09} with axioms expressing the behaviour of $\textl$ and $\textc$ in the presence of $\tau$-actions (from \cite{BK85}), and with the suitable $\tau$-laws (from \cite{HM85,vGW96}) necessary to deal with rooted branching bisimilarity.
Then, we adjust the semantics of configurations given in Section~\ref{sec:decomposition} to the $\ccslc$ setting, and we use it to extend the definition of rooted branching bisimilarity to open $\ccslc$ terms (Definition~\ref{def:open_rbb}).
Usually, a behavioural equivalence $\sim$ is defined over processes and is then possibly extended to open terms by saying that $t \sim u$ if{f} $\sigma(t) \sim \sigma(u)$ for all closed substitutions $\sigma$.
However, we adopt the same approach of, e.g., \cite{Mi89b,vG93c,AvGFI96}, and present the definition of $\sim_\rbb$ directly over configurations.
We will show in Section~\ref{sec:configurations_bis} that the two approaches yield the same equivalence relation over terms (Theorem~\ref{thm:rbb_on_open}).
Finally, we apply the strategy used in \cite{AvGFI96} to obtain the completeness of the axiomatisation of prefix iteration with silent moves modulo rooted branching bisimilarity:
\begin{enumerate}
\item\label{step1} We identify \emph{normal forms} for $\ccslc$ terms (Definition~\ref{def:rbb_nf_ccslc}) and show that each term can be proven equal to a normal form using $\E_\rbb$ (Proposition~\ref{prop:rbb_nf_ccslc}).
\item\label{step2} We establish a relationship between $\sim_\bb$ and derivability in $\E_\rbb$ (Proposition~\ref{prop:rbb_provable_ccslc}).
\item\label{step3} We show that for all terms $t,u$, if $t \sim_\rbb u$, then $\E_\rbb \vdash t \approx u$ (Theorem~\ref{thm:rbb_complete_ccslc}).
\end{enumerate}


\section{Rooted branching bisimilarity over terms}
\label{sec:configurations_bis}

In this section we discuss the decomposition of the semantics of $\ccslc$ terms, and the extension of the definition of (rooted) branching bisimilarity to open $\ccslc$ terms.

The first step towards our completeness result consists in providing a semantics for open $\ccslc$ terms.
To this end, we need to extend the semantics of configurations given in Section~\ref{sec:decomposition}.
For the sake of readability, we present the syntax of $\ccslc$ configurations and the inference rules for variables and summations, even though they are identical to the corresponding ones presented in Section~\ref{sec:decomposition} for CCS.
However, we omit the explanations on the roles of labels $\ell$, $\rho$, and variables $x_\mu$, as those can be found in Section~\ref{sec:decomposition}.
In particular, the use of variables $x_\mu \in \Var_{\Acttau}$ (as explained in Example~\ref{ex:variabili_vi_odio}) remains unchanged.

\begin{definition}
[$\ccslc$ configuration]
The collection of {\sl $\ccslc$ configurations}, denoted by $\C_{\text{LC}}$, is given by:
\[
c :: = \; x_\mu \quad | \quad 
t \quad | \quad
c \mathbin{\|} c 
\enspace , \quad
\text{ where } t \text{ is a } \ccslc \text{ term, and } x_\mu \in \Var_{\Acttau}. 
\]
\end{definition}

\begin{table}[t]
\begin{gather*}
\scalebox{0.9}{$(a'_3)$}\; \SOSrule{t \trans[\ell]_\rho c}
{t \lmerge u \trans[\ell]_{\rho} c \mathbin{\|} u}
\\[.3cm]
\scalebox{0.9}{$(a'_4)$}\; \SOSrule{t \trans[(x)]_\alpha c \quad u \trans[(y)]_{\overline{\alpha}} c'}
{t \cmerge u \trans[(x,y)]_{\tau} c \mathbin{\|} c'}
\quad
\scalebox{0.9}{$(a'_5)$}\; \SOSrule{t \trans[(x)]_\alpha c \quad u \trans[\overline{\alpha}] u'}
{t \cmerge u \trans[(x)]_{\,\alpha,\tau} c \mathbin{\|} u'}
\quad
\scalebox{0.9}{$(a'_6)$}\; \SOSrule{t \trans[\alpha] t' \quad u \trans[(x)]_{\overline{\alpha}} c}
{t \cmerge u \trans[(x)]_{\,\overline{\alpha},\tau} t' \mathbin{\|} c}
\end{gather*}
\caption{\label{tab:ell_rules_ccslc} Inference rules for the transition relation $\trans[\ell]_\rho$ ($\mu \in \Acttau$, $\alpha \in \Act\cup\overline{\Act}$).}
\end{table}

The auxiliary transitions of the form $\trans[\ell]_{\rho}$ are formally defined via the inference rules in Table~\ref{tab:ell_rules_ccslc}, where we omitted the rules ($a'_1$) and ($a'_2$) for prefixing and choice (which are identical to, respectively, rules ($a_1$) and ($a_2$) in Table~\ref{tab:ell_rules}) the symmetric rules to ($a'_2$), ($a'_4$), ($a'_5$) and ($a'_6$), as well as the rules for $\mathbin{\|}$.
We remark that Lemma~\ref{lem:closed2open_alpha} and Lemma~\ref{lem:closed2open_tau} can be easily extended to $\ccslc$ to show how a transition $\sigma(t) \trans[\mu] p$ can stem from transitions of the $\ccslc$ term $t$ and of the process $\sigma(x)$, for $x \in \var(t)$.

\begin{table}[t]
\[
\scalebox{0.9}{$(c_1)$}\; \SOSrule{}{x_\mu \trans[x_\mu] x_{\mu}}
\quad
\scalebox{0.9}{$(c_2)$}\; \SOSrule{c_1 \trans[x_\mu] c_1'}
{c_1 \mathbin{\|} c_2 \trans[x_\mu] c_1' \mathbin{\|} c_2}
\quad
\scalebox{0.9}{$(c_3)$}\; \SOSrule{c_1 \trans[\mu] c_1'}{c_1 \mathbin{\|} c_2 \trans[\mu] c_1' \mathbin{\|} c_2}
\quad
\scalebox{0.9}{$(c_4)$}\; \SOSrule{c_1 \trans[\ell]_\rho c_1'}
{c_1 \mathbin{\|} c_2 \trans[\ell]_\rho c_1' \mathbin{\|} c_2}
\]
\caption{\label{tab:c_rules} Inference rules completing the operational semantics of $\ccslc$ configurations ($\mu \in \Acttau$).}
\end{table}

Since $\Var_{\Acttau}$ is disjoint from $\Var$, we also need to introduce auxiliary rules for the special configuration $x_\mu \in \Var_{\Acttau}$.
These are identified by a proper label $x_\mu$ on the transition and reported in Table~\ref{tab:c_rules} as rules ($c_1$) and ($c_2$).
To conclude our analysis of the decomposition of the semantics of terms, we then need to extend the transition relations $\trans[\mu]$ and $\trans[\ell]_\rho$ to configurations.
This is done by rules ($c_3$) and ($c_4$) in Table~\ref{tab:c_rules}, where their symmetric counterparts have been omitted.
Let $\xitrans$ range over the possible transitions over configurations, i.e., $\xitrans$ can be either $\trans[\mu]$, $\trans[\ell]_\rho$, or $\trans[x_\mu]$.
The operational semantics of $\ccslc$ configurations is then given by the LTS whose states are configurations in $\C_{\text{LC}}$, whose actions are in $\Acttau \cup \Var \cup \Var_{\Acttau}$, and whose transitions are those that are provable from the rules in Tables~\ref{tab:sos_rules},~\ref{tab:sos_rules_ccslc},~\ref{tab:ell_rules_ccslc}, and~\ref{tab:c_rules}.

Following the same approach of, e.g. \cite{Mi89b,vG93c,AvGFI96}, we now present the definitions of branching and rooted branching bisimulation equivalences directly over configurations.

\begin{definition}
[Branching bisimulation over configurations]
\label{def:open_bb}
A symmetric relation $\rel$ over $\C_{\text{LC}}$ is a \emph{branching bisimulation} if{f} whenever $c_1 \rel c_2$, if $c_1 \xitrans c_1'$ then:
\begin{itemize}
\item either $\xitrans\; =\; \trans[\tau]$ and $c_1' \rel c_2$,
\item or $c_2 \trans[\varepsilon] c_2'' \xitrans c_2'$ for some $c_2'',c_2'$ such that $c_1 \rel c_2''$ and $c_1' \rel c_2'$.
\end{itemize}
Two configurations $c_1,c_2$ are \emph{branching bisimilar}, denoted by $c_1 \sim_\bb c_2$, if{f} there exists a branching bisimulation $\rel$ such that $c_1 \rel c_2$.
\end{definition}

The definition of $\sim_\bb$ given in Definition~\ref{def:open_bb} yields the same equivalence relation over configurations that we would have obtained with the standard approach, i.e., by defining $c_1 \sim_\bb c_2$ if{f} $\sigma(c_1) \sim_\bb \sigma(c_2)$ for all closed substitutions $\sigma$.

\begin{restatable}{theorem}{thmbbonopen}
\label{thm:bb_on_open}
For all configurations $c_1, c_2 \in \C_{\text{LC}}$ it holds that $c_1 \sim_\bb c_2$ if{f} $\sigma(c_1) \sim_\bb \sigma(c_2)$ for all closed substitutions $\sigma$.
\end{restatable}

The approach for $\sim_\bb$ can be extended in a straightforward manner to $\sim_\rbb$.

\begin{definition}
[Rooted branching bisimilarity over configurations]
\label{def:open_rbb}
Let $c_1,c_2 \in \C_{\text{LC}}$.
We say that $c_1$ and $c_2$ are \emph{rooted branching bisimilar}, denoted by $c_1 \sim_\rbb c_2$, if{f}:
\begin{itemize}
\item if $c_1 \xitrans c_1'$ then $c_2 \xitrans c_2'$ for some $c_2'$ such that $c_1' \sim_\bb c_2'$;
\item if $c_2 \xitrans c_2'$ then $c_1 \xitrans c_1'$ for some $c_1'$ such that $c_1' \sim_\bb c_2'$.
\end{itemize}
\end{definition}

\begin{theorem}
\label{thm:rbb_on_open}
For all $c_1, c_2 \in \C_{\text{LC}}$ it holds that $c_1 \sim_\rbb c_2$ if{f} $\sigma(c_1) \sim_\rbb \sigma(c_2)$ for all closed substitutions $\sigma$.
\end{theorem}


\section{The equational basis}
\label{sec:completeness}

We now present the complete axiomatisation for rooted branching bisimilarity over $\ccslc$.

\begin{table*}[t]
\setlength{\tabcolsep}{15pt}
\centering
\begin{tabular}{llll}
\multicolumn{3}{l}{Equational basis modulo strong bisimilarity: $\E_\B$} \\[.2cm]
\multicolumn{2}{l}{$\scalebox{0.85}{A0}\quad x + \nil \approx x$}
&
\multicolumn{2}{l}{$\scalebox{0.85}{C0}\quad \nil \cmerge x \approx \nil$}\\
\multicolumn{2}{l}{$\scalebox{0.85}{A1}\quad x+y \approx y+x$}
&
\multicolumn{2}{l}{$\scalebox{0.85}{C1}\quad x \cmerge y \approx y \cmerge x$}\\
\multicolumn{2}{l}{$\scalebox{0.85}{A2}\quad (x+y)+z \approx x+(y+z)$} 
&
\multicolumn{2}{l}{$\scalebox{0.85}{C2}\quad (x \cmerge y) \cmerge z \approx x \cmerge (y \cmerge z)$}\\
\multicolumn{2}{l}{$\scalebox{0.85}{A3}\quad x + x \approx x$}
& 
\multicolumn{2}{l}{$\scalebox{0.85}{C3}\quad (x + y) \cmerge z \approx x \cmerge z + y \cmerge z$}\\
& &
\multicolumn{2}{l}{$\scalebox{0.85}{C4}\quad \alpha x \cmerge \beta y \approx \tau(x \mathbin{\|} y)$ \quad if $\alpha = \overline{\beta}$}\\
\multicolumn{2}{l}{$\scalebox{0.85}{L0}\quad \nil \lmerge x \approx \nil$}
&
\multicolumn{2}{l}{$\scalebox{0.85}{C5}\quad \alpha x \cmerge \beta y \approx \nil$ \quad if $\alpha \neq \overline{\beta}$}\\
\multicolumn{2}{l}{$\scalebox{0.85}{L1}\quad \mu x \lmerge y \approx \mu (x \mathbin{\|} y)$}
&
\multicolumn{2}{l}{$\scalebox{0.85}{C6}\quad (x \lmerge y) \cmerge z \approx (x \cmerge z) \lmerge y$} \\
\multicolumn{2}{l}{$\scalebox{0.85}{L2}\quad (x \lmerge y) \lmerge z \approx x \lmerge (y \mathbin{\|} z)$}
&
\multicolumn{2}{l}{$\scalebox{0.85}{C7}\quad x \cmerge y \cmerge z \approx \nil$}\\
\multicolumn{2}{l}{$\scalebox{0.85}{L3}\quad x \lmerge \nil \approx x$}
\\
\multicolumn{2}{l}{$\scalebox{0.85}{L4}\quad (x + y) \lmerge z \approx x \lmerge z + y \lmerge z$}
& \multicolumn{2}{l}{$\scalebox{0.85}{P}\quad x \mathbin{\|} y \approx x \lmerge y + y \lmerge x + x \cmerge y$} \\[.2cm]
\hline
\hline\\
\multicolumn{3}{l}{Additional axioms for $\sim_\rbb$: $\E_\rbb = \E_\B \cup \{TB,TL\}$} \\[.2cm]
\multicolumn{2}{l}{$\scalebox{0.85}{TB}\quad \mu (\tau(x+y) + y) \approx \mu (x +y)$}
&
\multicolumn{2}{l}{$\scalebox{0.85}{TL}\quad x \lmerge \tau y \approx x \lmerge y$} 
\\[.2cm]
\hline
\hline\\
\multicolumn{2}{l}{Derivable axioms} \\[.2cm]
\multicolumn{2}{l}{$\scalebox{0.85}{D1}\quad x \mathbin{\|} y \approx y \mathbin{\|} x$} 
&
\multicolumn{2}{l}{$\scalebox{0.85}{DT1}\quad \mu \tau x \approx \mu x$}\\
\multicolumn{2}{l}{$\scalebox{0.85}{D2}\quad (x \mathbin{\|} y) \mathbin{\|} z \approx x \mathbin{\|} (y \mathbin{\|} z)$}
&
\multicolumn{2}{l}{$\scalebox{0.85}{DT2}\quad x \lmerge (\tau(y+z) + y) \approx x \lmerge (y + z)$}\\
\multicolumn{2}{l}{$\scalebox{0.85}{D3}\quad (x \lmerge y) \cmerge (z \lmerge w) \approx (x \cmerge z) \lmerge (y \mathbin{\|} w)$}
&
\multicolumn{2}{l}{$\scalebox{0.85}{DT3}\quad \tau x \cmerge y \approx \nil$}
\\
\multicolumn{2}{l}{$\scalebox{0.85}{D4}\quad x \mathbin{\|} \nil \approx x$}
\end{tabular}
\caption{\label{tab:axioms_rbb} Equational basis modulo rooted branching bisimilarity.}
\end{table*}

In \cite{vGW96} it was proved that if we consider the fragment BCCS of CCS (i.e., the fragment consisting only of $\nil$, variables, prefixing, and choice), then a ground-complete axiomatisation of rooted branching bisimilarity over BCCS is given by $\E_0 \cup \{\text{TB}\}$, where $\E_0 = \{\text{A0,A1,A2,A3}\}$ from Table~\ref{tab:axioms_b} (also reported in Table~\ref{tab:axioms_rbb}), and axiom TB is in Table~\ref{tab:axioms_rbb}.
Informally, TB reflects that if executing a $\tau$-step does not discard any observable behaviour, then it is redundant.
In \cite{AFIL09} it was proved that the axiom system $\E_\B$ given in Table~\ref{tab:axioms_rbb}, is a complete axiomatisation of bisimilarity over $\ccslc$.
Starting from these works, we now study a complete axiomatisation for $\sim_\rbb$.
Our aim is to show that the axiom system $\E_\rbb = \E_\B \cup \{\text{TB,TL}\}$ presented in Table~\ref{tab:axioms_rbb} is a \emph{complete axiomatisation of rooted branching bisimilarity} over $\ccslc$.

If executing a $\tau$-move does not resolve a choice within a parallel component,
then it will also not resolve a choice of the parallel composition; axiom TL expresses this property of rooted branching bisimilarity for left merge.
Interestingly, by combining TL and TB, it is possible to derive, as shown below, equation DT2 in Table~\ref{tab:axioms_rbb}, which is the equation for the left merge corresponding to TB.
\[
x \lmerge (\tau(y+z) +y) 
\stackrel{\scalebox{0.75}{(TL)}}{\approx{}}
\;
x \lmerge \tau (\tau(y+z) +y) 
\stackrel{\scalebox{0.75}{(TB)}\phantom{1}}{\approx{}}
\;
x \lmerge \tau (y+z) 
\stackrel{\scalebox{0.75}{(TL)}}{\approx{}}
\;
x \lmerge (y+z).
\]
In Table~\ref{tab:axioms_rbb} we report also some other equations that can be derived from $\E_\rbb$, and that are useful in the technical development of our results.
We refer the reader interested in the derivation proofs of D1--D3 and DT3 to \cite{AFIL09}.
Notice that DT1 corresponds essentially to the substitution instance of TB in which $y$ is mapped to $\nil$.

First of all, it is immediate to prove the soundness of $\E_{\rbb}$ modulo $\sim_\rbb$.

\begin{theorem}
[Soundness]
The axiom system $\E_\rbb$ is sound modulo $\sim_\rbb$ over $\ccslc$.
\end{theorem}

To obtain the desired completeness result, we apply the same strategy used in \cite{AvGFI96} that consists in the three steps discussed in Section~\ref{sec:completeness_roadmap}.

Let us proceed to the first step: identifying normal forms for $\ccslc$ terms.

\begin{definition}
[Normal forms]
\label{def:rbb_nf_ccslc}
The set of normal forms over $\ccslc$ is generated by the following grammar:
\begin{align*}
S ::={} & 
\mu.N \quad | \quad
x \lmerge N \quad | \quad
(x \cmerge \alpha) \lmerge N \quad | \quad
(x \cmerge y) \lmerge N \\
N ::={} &
\nil \quad | \quad
S \quad | \quad
N + N
\end{align*}
where $x,y \in \Var$, $\mu \in \Acttau$ and $\alpha \in \Act \cup \overline{\Act}$.
Normal forms generated by $S$ are also called \emph{simple normal forms} and are characterised by the fact that they do not have $+$ as head operator.
\end{definition}

\begin{restatable}{proposition}{proprbbnfccslc}
\label{prop:rbb_nf_ccslc}
For every term $t$ there is a normal form $N$ such that $\E_\rbb \vdash t \approx N$.
\end{restatable}

We can then proceed to prove that branching bisimilar terms can be proven equal to rooted branching bisimilar terms using the axiom system $\E_\rbb$.

\begin{restatable}{proposition}{proprbbprovableccslc}
\label{prop:rbb_provable_ccslc}
For $\ccslc$ terms $t,u$,
if $t \sim_\bb u$ then $\E_\rbb \vdash \mu.t \approx \mu.u$, for any $\mu \in \Acttau$.
\end{restatable}

The completeness of the axiom system $\E_\rbb$ then follows from Proposition~\ref{prop:rbb_nf_ccslc} and Proposition~\ref{prop:rbb_provable_ccslc}.
Notice that axioms L1 and TB are actually axiom schemata that both generate $|\Acttau|$ axioms.
Similarly, the schema C4 generates $2|\Act|$ axioms, and C5 generates $2|\Act| \times (2|\Act|-1)$ axioms.
Hence, $\E_\rbb$ is \emph{finite} when so is the set of actions. 

\thmrbbcompleteccslc*


\section{Concluding remarks}
\label{sec:conclusion}

In this paper we have shown that the use of auxiliary operators left merge and communication merge is crucial to obtain a finite, complete axiomatisation of the CCS parallel composition operator modulo rooted branching bisimilarity.

A natural direction for future research is the extension of our results to other weak congruences from the spectrum \cite{vG93}.

In detail, we will investigate the existence of a general technique to lift the negative result from rooted branching bisimilarity to other weak congruences.

Regarding the positive result, we will focus on three weak congruences, namely \emph{rooted} $\eta$-\emph{bisimilarity} ($\sim_{\reb}$), \emph{rooted delay bisimilarity} ($\sim_\rdb$), and \emph{rooted weak bisimilarity} ($\sim_\rwb$), and provide the complete axiomatisations for them.
We are confident that the axiomatisation for $\sim_{\reb}$ can be obtained by exploiting a proof technique from \cite{vG95,AvGFI96} based on the notion of \emph{saturation}.
Intuitively, $\sim_{\reb}$ coincides with $\sim_\rbb$ on the class of $\eta$-saturated terms.
Hence, if we can show that each term is provably equal to an $\eta$-saturated term using the axiom system for $\sim_{\reb}$, the completeness of the considered axiom system then directly follows from that for $\sim_\rbb$ we provided in this paper.

The quest for complete axiomatisations for $\sim_\rdb$ and $\sim_\rwb$ will require a different approach, as these equivalences are not preserved by the communication merge operator.
For instance, we have that $\tau.a \sim_\rwb \tau.a + a$, but
$\nil \sim_\rwb \tau.a \cmerge \overline{a}.b \not\sim_\rwb (\tau.a + a) \cmerge \overline{a}.b \sim_\rwb \tau.b$.
For $\sim_\rdb$, we have a similar outcome (see \cite{vG11} for more details).
If we look at the seminal paper \cite{BK85}, the complete axiomatisation for observational congruence \cite{HM85} (and thus rooted weak bisimilarity) over $\text{ACP}_\tau$ comes with the axiom
\begin{equation}
\tag{TC}
\tau.x \cmerge y \approx x \cmerge y.
\end{equation}
Simialrly, in \cite{He88,Ac94}, it was argued that in order to reason compositionally, and obtain an equational theory of CCS modulo observational congruence, it is necessary to define the operational semantics of communication merge in terms of inference rules of the form 
\[
\SOSrule{t \wtrans[\alpha] t' \quad u \wtrans[\overline{\alpha}] u'}{t \cmerge u \wtrans[\tau] t' \mathbin{\|} u'}
\]
where we use $\wtrans[\mu]$ as a short-hand for the sequence of transitions $\trans[\varepsilon]\trans[\mu]\trans[\varepsilon]$.
This means that in order for $\textc$ to preserve $\sim_\rwb$ (and/or $\sim_\rdb$), we need to consider a sequence of weak transitions as a single step. 
Clearly, since $\textc$ is an auxiliary operator that we introduce specifically to obtained the axiomatisations, its semantics can be defined in the most suitable way for our purposes, i.e., so that it is consistent with the considered congruence relation.
However, it is also clear that if we modify the semantics of one operator in $\ccslc$, then we are working with a new language.
In particular, some axioms that are sound modulo \emph{strong} bisimilarity (and thus also modulo $\sim_\rbb$) over $\ccslc$ become unsound modulo rooted weak bisimilarity over the new language: this is the case of axioms C6 and C7 in Table~\ref{tab:axioms_rbb}.
As a consequence, we cannot exploit the completeness of the axiomatisation for rooted branching bisimilarity to derive complete axiomatisations for rooted weak bisimilarity and rooted delay bisimilarity, but we must provide new axiomatisations for them and prove their completeness from scratch.
Hence, we leave as future work the quest for complete axiomatisations for $\sim_\rwb$ and $\sim_\rdb$ over (recursion, relabelling, and restriction free) CCS with left merge and communication merge.

\bibliographystyle{plainurl}
\bibliography{concur_22}

\newpage
\appendix

\section{Additional background notions and results}

In this section we present some additional notions and general simple results on CCS terms that we omitted from the main text because of space limits, and that will be useful in the upcoming proofs.

Firstly, we introduce the notion of \emph{derivative} of a process.

\begin{definition}
[Derivative]
For a process $p$, the set of \emph{derivatives} of $p$, notation $\der(p)$, is the least set containing $p$ that is closed under $\trans[]$, i.e., the least set satisfying:
\begin{itemize}
\item $p \in \der(p)$, and
\item if $q \in \der(p)$ and $q \trans[\mu] q'$, for some action $\mu \in \Acttau$, then $q' \in \der(p)$.
\end{itemize}
\end{definition}

In particular, we say that $p' \in \der(p)$ is a $\mu$-\emph{derivative} of $p$, for some $\mu \in \Acttau$, if $p \trans[\mu] p'$.
Moreover, we say that $p'$ is a \emph{proper} derivative of $p$ if $p' \in \der(p) \setminus \{p\}$.

Let $p \wtrans[\mu] q$ be a shorthand for $p \trans[\varepsilon] p' \trans[\mu] q$, for some $p'$.
For a sequence of actions $\varphi = \alpha_1 \cdots \alpha_k \in (\Act\cup\overline{\Act})^*$ ($k \geq 0$), and processes $p,p'$, we write that $p \wtrans[\varphi] p'$ if and only if there exists a sequence of transitions $p = p_0 \wtrans[\alpha_1] p_1 \wtrans[\alpha_2] \cdots \wtrans[\alpha_k] p_k = p'$. 
If $p \wtrans[\varphi] p'$ holds for some process $p'$, then $\varphi$ is an {\emph observable trace} of $p$. 
Moreover, we say that $\varphi$ is a \emph{maximal} observable trace of $p$ if $\init(p') = \emptyset$.
By means of observable traces, we associate a classic notion with a process $p$, i.e., its (\emph{observable}) \emph{depth}, denoted by $\depth(p)$.
For a process $p$ whose set of traces is finite, it expresses the length of a \emph{longest} observable trace.
Formally, denoting by $|\varphi|$ the length of $\varphi$,
\[
\depth(p) = \max \{ k \mid p \wtrans[\varphi] p' \;\wedge\; |\varphi| = k\}.
\]

The case of rooted branching bisimilarity is slightly more complicated, as we need to treat a possible initial $\tau$-move as an observable action.
To this end, we define the \emph{rooted depth} of a process, denoted by $\rdepth(p)$:
\[
\rdepth(p) = \sup \{ 1 + \depth(p') \mid \exists \mu \in \Acttau, p' \in \Proc \text{ s.t. } p \trans[\mu] p' \}.
\]
Then, we notice that we can give the following inductive characterisation of the (rooted) depth of CCS processes:
\begin{itemize}
\item $\depth(\nil) = 0$;
\item $\depth(\tau.p) = \depth(p)$;
\item $\depth(\alpha.p) = 1+ \depth(p)$;
\item $\depth(p + q) = \max\{\depth(p), \depth(q)\}$;
\item $\depth(p \mathbin{\|} q) = \depth(p) + \depth(q)$.
\item $\rdepth(\nil) = 0$;
\item $\rdepth(\mu.p) = 1 + \depth(p)$;
\item $\rdepth(p + q) = \max\{\rdepth(p), \rdepth(q)\}$;
\item $\rdepth(p \mathbin{\|} q) = \rdepth(p) + \rdepth(q)$.
\end{itemize}

An immediate result is that branching bisimilarity preserves the observable depth of processes:

\begin{lemma}
\label{lem:bb_same_depth}
If $p \sim_\bb q$, then $\depth(p) = \depth(q)$.
\end{lemma}

Similarly, rooted branching bisimilarity preserves the rooted depth of processes:

\begin{lemma}
\label{lem:rbb_same_depth}
If $p \sim_\rbb q$, then $\rdepth(p) = \rdepth(q)$.
\end{lemma}

Moreover, it is well known that rooted branching bisimilarity is finer than branching bisimilarity, in the sense that whenever $p \sim_\rbb q$, then also $p \sim_\bb q$ holds.
In general, the converse implication does not hold.
However, we present here a simple case in which we can establish that two branching bisimilar processes are also rooted branching bisimilar:

\begin{lemma}
\label{lem:same_initials}
Assume that processes $p,q$ are such that $p \sim_\bb q$.
If $\tau \not \in \init(p)\cup\init(q)$, then $p \sim_\rbb q$.
Moreover, if $\tau \not\in \init(p)$, then $p \not\sim_\rbb \nil$ implies $p \not\sim_\bb \nil$.
\end{lemma}

We conclude this section by recalling that we can restrict the axiom system to a collection of equations that do not introduce unnecessary terms that are rooted branching bisimilar to $\nil$ in the equational proofs, namely $\nil$ summands and $\nil$ factors.

An example of a collection of equations over CCS that are sound modulo $\sim_\rbb$ is given by axioms A0--A3, P0--P2, in Table~\ref{tab:axioms_b}.
Interestingly, axioms A0 and P0 in Table~\ref{tab:axioms_b} (used from left to right) are enough to establish that each term that is rooted branching bisimilar to $\nil$ is also provably equal to $\nil$. 

It is well-known (see, e.g., Section~2 in \cite{Gr90}) that if an equation relating two closed terms can be proven from an axiom system $\E$, then there is a closed proof for it.  
In addition, if $\E$ satisfies a further closure property, called \emph{saturation} in \cite[Definition 5.1.1]{Mo89}, and that closed equation relates two terms containing no occurrences of $\nil$ as a summand or factor, then there is a closed proof for it in which all of the terms have no occurrences of $\nil$ as a summand or factor, as formalised in the following proposition.

\begin{proposition}
\label{Propn:proofswithout0_text}
Assume that $\E$ is a saturated axiom system.  
Then, the proof from $\E$ of an equation $p \approx q$, where $p$ and $q$ are terms not containing occurrences of $\nil$ as a summand or factor, need not use terms containing occurrences of $\nil$ as a summand or factor.
\end{proposition}

Since the proof of this result follows the same lines of that of \cite[Proposition~5.1.5]{Mo89}, we omit it from our work.
In light of Proposition~\ref{Propn:proofswithout0_text} we shall, henceforth, limit ourselves to considering saturated axiom systems.


\section{Additional results to Section~\ref{sec:decomposition}}
\label{app:decomposition}

Firstly, we present a general result on the structure of configurations.

\begin{lemma}
\label{lem:c_structure}
Let $t$ be a term and $x$ be a variable in $\Var$.
If $t \trans[\ell]_\rho c$ for some $\rho \in \Acttau \cup ((\Act \cup \overline{\Act})\times \{\tau\})$ and label $\ell$ with $x \in \ell$, then $c$ is of the form $x_\mu \mathbin{\|} c'$, for some configuration $c'$, and action $\mu \in \Acttau$ (modulo the axioms in Table~\ref{tab:axioms_b}).
\end{lemma}

The following lemma shows how transitions labelled with observable actions are triggered by variables.

\begin{lemma}
\label{lem:closed2open_alpha}
Let $\alpha \in \Act \cup \overline{\Act}$, $t$ be a term, $\sigma$ be a closed substitution, and $p$ be a process.
If $\sigma(t) \trans[\alpha] p$, then one of the following holds:
\begin{enumerate}
\item \label{c2o_prefix}
There is a term $t'$ s.t.\ $t \trans[\alpha] t'$ and $\sigma(t') = p$.
\item \label{c2o_x}
There are a variable $x$, a process $q$ and a configuration $c$ s.t.\ $\sigma(x) \trans[\alpha] q$, $t \trans[(x)]_\alpha c$, and $\sigma[x_\alpha \mapsto q](c) = p$.
\end{enumerate}
\end{lemma}


\section{Proof of Lemma~\ref{lem:closed2open_tau}}

\lemclosedopentau*

\begin{proof}
The proof follows by induction on the derivation of the transition $\sigma(t) \trans[\tau] p$.
The interesting case is the base case corresponding to a communication in the term $t = t_1 \mathbin{\|} t_2$, for some CCS terms $t_1,t_2$, that we expand below.

We can assume, without loss of generality, that $\sigma(t_1) \trans[\alpha] p_1$ and $\sigma(t_2) \trans[\overline{\alpha}] p_2$, for some processes $p_1,p_2$ such that $p_1 \mathbin{\|} p_2 = p$.
The symmetric case of $\sigma(t_1) \trans[\overline{\alpha}]$ and $\sigma(t_2) \trans[\alpha]$ follows by applying the same arguments we use below, switching the roles of $\sigma(t_1)$ and $\sigma(t_2)$.
By Lemma~\ref{lem:closed2open_alpha}, from $\sigma(t_1) \trans[\alpha] p_1$ we can infer that either $t_1 \trans[\alpha] t_1'$ for some term $t_1'$ such that $\sigma(t_1') = p_1$, or there are a variable $x$, a process $q_1$ and a configuration $c_1$ such that $\sigma(x) \trans[\alpha] q_1$, $t_1 \trans[(x)]_\alpha c_1$, and $\sigma[x_\alpha \mapsto q_1](c_1) = p_1$.
Similarly, by Lemma~\ref{lem:closed2open_alpha}, from $\sigma(t_2) \trans[\overline{\alpha}] p_2$ we can infer that either $t_2 \trans[\overline{\alpha}] t_2'$ for some term $t_2'$ such that $\sigma(t_2') = p_2$, or there are a variable $y$, a process $q_2$ and a configuration $c_2$ such that $\sigma(y) \trans[\overline{\alpha}] q_2$, $t_2 \trans[(y)]_{\overline{\alpha}} c_2$, and $\sigma[y_{\overline{\alpha}} \mapsto q_2](c_2) = p_2$.
We can then distinguish four cases, depending on how the possible derivations of the transitions of $t_1$ and $t_2$ mentioned above are combined in the derivation of the synchronisation. 
\begin{enumerate}
\item Case $t_1 \trans[\alpha] t_1'$ and $t_2 \trans[\overline{\alpha}] t_2'$.
In this case, the operational semantics of $\mathbin{\|}$ allows us to directly infer that $t_1 \mathbin{\|} t_2 \trans[\tau] t_1' \mathbin{\|} t_2'$.
Hence we have that there is $t' = t_1' \mathbin{\|} t_2'$ such that $t \trans[\tau] t'$ and $\sigma(t') = p_1 \mathbin{\|} p_2 = p$.

\item Case $\sigma(x) \trans[\alpha] q_1$, $t_1 \trans[(x)]_\alpha c_1$, $\sigma[x_\alpha \mapsto q_1](c_1) = p_1$, and $t_2 \trans[\overline{\alpha}] t_2'$.
As $t_1 \trans[(x)]_\alpha c_1$ and $t_2 \trans[\overline{\alpha}] t_2'$ we can apply the auxiliary rule ($a_5$) and obtain that $t_1 \mathbin{\|} t_2 \trans[(x)]_{\alpha,\tau} c_1 \mathbin{\|} t_2'$.
Hence we have that there are a variable $x$, a subterm $t_2$ of $t$, a process $q_1$ and a configuration $c = c_1 \mathbin{\|} t_2'$ such that $\sigma(x) \trans[\alpha] q_1$, $t_2 \trans[\overline{\alpha}] t_2'$, $t \trans[(x)]_{\alpha,\tau} c$ and $\sigma[x_\alpha \mapsto q_1](c) = \sigma[x_\alpha \mapsto q_1](c_1) \mathbin{\|} \sigma(t_2') = p_1 \mathbin{\|} p_2 = p$.

\item Case $t_1 \trans[\alpha] t_1'$, $\sigma(y) \trans[\overline{\alpha}] q_2$, $t_2 \trans[(y)]_{\overline{\alpha}} c_2$, and $\sigma[y_{\overline{\alpha}} \mapsto q_2](c_2) = p_2$.
This is the symmetrical of the previous case and it follows by applying the same reasoning used above, using the auxiliary rule ($a_6$) in place of ($a_5$).

\item Case $\sigma(x) \trans[\alpha] q_1$, $t_1 \trans[(x)]_\alpha c_1$, $\sigma[x_\alpha \mapsto q_1](c_1) = p_1$, $\sigma(y) \trans[\overline{\alpha}] q_2$, $t_2 \trans[(y)]_{\overline{\alpha}} c_2$, and $\sigma[y_{\overline{\alpha}} \mapsto q_2](c_2) = p_2$.
As $t_1 \trans[(x)]_\alpha c_1$ and $t_2 \trans[(y)]_{\overline{\alpha}} c_2$, we can apply the auxiliary rule ($a_4$) and obtain thus $t_1 \mathbin{\|} t_2 \trans[(x,y)]_{\tau} c_1 \mathbin{\|} c_2$.
Hence we have obtained that there are variables $x,y$, processes $q_1,q_2$ and a configuration $c = c_1 \mathbin{\|} c_2$ such that $\sigma(x) \trans[\alpha] q_1$, $\sigma(y) \trans[\overline{\alpha}] q_2$, $t \trans[(x,y)]_\tau c$, and $\sigma[x_\alpha \mapsto q_1, y_{\overline{\alpha}} \mapsto q_2](c) = 
\sigma[x_\alpha \mapsto q_1](c_1) \mathbin{\|} \sigma[y_{\overline{\alpha}} \mapsto q_2](c_2) = 
p_1 \mathbin{\|} p_2 = p$, where we can distribute the substitutions over $\mathbin{\|}$ since $x_\alpha$ and $y_{\overline{\alpha}}$ are unique by construction (even if $x = y$, the two subscripts allow us to distinguish them). 
\end{enumerate}
\end{proof}


\section{Proof of Proposition~\ref{prop:rttransprops}}

\rttransprops*

\begin{proof}
  Clearly, $\rightarrow^{*}$ is reflexive and transitive by definition, and since $p\trans[\mu]p'$ implies $\size(p)>\size(p')$
  it follows that $\rightarrow^{*}$ is inversely well-founded. Every inversely well-founded reflexive and transitive relation is clearly also anti-symmetric, so $\rightarrow^{*}$ is an inversely well-founded partial order.
  
  By inverse well-foundedness, for every process $p$ there exists a process $p'$ such that $p\rightarrow^{*}p'$ and $p'{\not\rightarrow}$; from $p'{\not\rightarrow}$ it follows that $p'\sim_{\bb}\nil$. Hence $\rightarrow^{*}$ satisfies property \ref{item:rttranspropsleast}.
  
  That $\rightarrow^{*}$ satisfies property~\ref{item:rttranspropscompatible} is an immediate consequence of the two rules at the bottom left of Table~\ref{tab:sos_rules}.
  
  That $\rightarrow^{*}$ satisfies property~\ref{item:rttranspropsprecompositional} is straightforwardly established with induction on the length of a transition sequence witnessing $p\mathbin{\|}q\rightarrow^{*}r$, using that the last SOS rule applied in the derivation of each individual transition must be one of the three rules for $\mathbin{\|}$.
  
  To see that $p\rightarrow^{*}q^n$ for all $n\in\mathbb{N}$ implies $q\sim_{\bb}\nil$, note that $p\rightarrow^{*}q^n$ implies $\depth(p)\geq\depth(q^n)$ and $\depth(q^n)=n\cdot\depth(q)$, from which it follows that $\depth(q)=0$ and hence $q\sim_{\bb}\nil$. This proves that $\rightarrow^{*}$ satisfies property~\ref{item:rttranspropsArchimedean}.
\end{proof}


\section{Proof of Proposition~\ref{prop:bb_cancellation}}

\propbbcancellation*

\begin{proof}
  Let $p$ be an indecomposable process and let $q$ and $r$ be processes such that $[p]_{\sim_{\bb}}\not\prec [q]_{\sim_{\bb}}, [r]_{\sim_{\bb}}$, and suppose that $[p^k\mathbin{\|}q]_{\sim_{\bb}}=[p^k\mathbin{\|}r]_{\sim_{\bb}}$ for some $k\in\mathbb{N}$. We need to prove that $[q]_{\sim_{\bb}}=[r]_{\sim_{\bb}}$, and for this it suffices to argue that the relation
  \begin{equation*}
      \mathcal{R}=\{(q,r)\}\cup{\sim_{\bb}}
  \end{equation*}
  is a branching bisimulation relation.
  
  The weak decomposition order $\preceq$, since it is well-founded, induces a well-founded order on triples of $\sim_{\bb}$-equivalence classes of processes, defining
    $([p_1]_{\sim_{\bb}},[p_2]_{\sim_{\bb}},[p_3]_{\sim_{\bb}})
      \preceq
     ([q_1]_{\sim_{\bb}},[q_2]_{\sim_{\bb}},[q_3]_{\sim_{\bb}})$
   if, and only if,
     $[p_1]_{\sim_{\bb}}\preceq [q_1]_{\sim_{\bb}}$
       and whenever $[p_1]_{\sim_{\bb}}=[q_1]_{\sim_{\bb}}$, then also $[p_2]_{\sim_{\bb}}\preceq [q_2]_{\sim_{\bb}}$ and $[p_3]_{\sim_{\bb}}\preceq [q_3]_{\sim_{\bb}}$.
  Our proof is by induction on the well-founded order $\preceq$ on triples of $\sim_{\bb}$-equivalence classes of processes: we assume, by way of induction hypothesis, that for all $p'$, $q'$ and $r'$ such that
  \begin{equation*}
     ([(p')^{k}\mathbin{\|}q']_{\sim_{\bb}},[q']_{\sim_{\bb}},[r']_{\sim_{\bb}})
      \prec
     ([p^{k}\mathbin{\|}q]_{\sim_{\bb}},[q]_{\sim_{\bb}},[r]_{\sim_{\bb}})
  \end{equation*}
  we have that $(p')^k\mathbin{\|}q'\sim_{\bb} (p')^k\mathbin{\|} r'$ implies $q'\sim_{\bb}r'$,
  and establish the property for $p$, $q$ and $r$.
  
  Since $q\rightarrow q'$ implies $\size(q)>\size(q')$, and hence is $\rightarrow$ is terminating, we may, without loss of generality, choose $q$ such that there does not exist $q'$ such that $q\rightarrow q'$ and $q\sim_{\bb}q'$, and similarly we may, without loss of generality, choose $r$ such that there does not exist an $r'\neq r$ such that $r\rightarrow r'$ and $r\sim_{\bb}r'$.
  
  Since $\sim_{\bb}$ is the largest branching bisimulation, it is immediate that all pairs in $\sim_{\bb}$ satisfy the conditions of branching bisimulations. We do need to establish that the pair $(q,r)$ also satisfies these conditions, and to this end suppose that $q\trans[\mu]q'$ for some $q'$; we establish that either $\mu=\tau$ and $q'\mathrel{\mathcal{R}} r$, or there exist processes $r'$ and $r''$ such that $r\trans[\epsilon]r'\trans[\mu]r''$, $q\mathrel{\mathcal{R}}r'$ and $q'\mathrel{\mathcal{R}}r''$. We distinguish two cases:
  \begin{enumerate}
      \item Suppose that $p^k\mathbin{\|}q'\sim_{\bb} p^k\mathbin{\|}q$. Then, since $\sim_{\bb}$ preserves depth and the depth of a parallel composition is the sum of the depths of its components, we must have that $\depth(q')=\depth(q)$ and hence $\mu=\tau$. Now, by the choice of $q$ it follows that $[q']_{\sim_{\bb}}\prec [q]_{\sim_{\bb}}$. Moreover, from $p^k\mathbin{\|}q'\sim_{\bb} p^k\mathbin{\|}q$ it follows that $[p^k\mathbin{\|}q']_{\sim_{\bb}}=[p^k\mathbin{\|}q]_{\sim_{\bb}}$. Thus we find that
      \begin{equation*}
        ([p^{k}\mathbin{\|}q']_{\sim_{\bb}},[q']_{\sim_{\bb}},[r]_{\sim_{\bb}})
      \prec
        ([p^{k}\mathbin{\|}q]_{\sim_{\bb}},[q]_{\sim_{\bb}},[r]_{\sim_{\bb}})
      \enskip,
      \end{equation*}
      so by the induction hypothesis $q'\sim_{\bb} r$, and hence $q'\mathrel{\mathcal{R}}r$.
      \item Suppose that $p^k\mathbin{\|}q'\not\sim_{\bb} p^k\mathbin{\|}q$. Then
        $[p^k\mathbin{\|}q']_{\sim_{\bb}}\prec [p^k\mathbin{\|}q]_{\sim_{\bb}}$, and,
      by the induction hypothesis, the weak decomposition order $\preceq$ on the partial commutative submonoid $\{x\mid x\preceq [p^k\mathbin{\|}q']_{\sim_{\bb}}\}$ of $\Proc$ modulo $\sim_{\bb}$ satisfies power cancellation. Hence, by \cite[Theorem 34]{Bas16}, if $s$ is any process such that $p^k\mathbin{\|}q'\sim_{\bb};\rightarrow^{*}s$ (where $;$ denotes relation composition), then $s$ has a unique parallel decomposition.
      
      From $q\trans[\mu]q'$ if follows that $p^k\mathbin{\|}q \trans[\mu]p^k\mathbin{\|}q'$, and hence, since $p^k\mathbin{\|}q\sim_{\bb}p^k\mathbin{\|}r$ there exist processes $p',p'',r',r''$ such that
      \begin{gather*}
        p^k\rightarrow^{*}p'\rightarrow^{*}p''\enskip,\\
        r\rightarrow^{*}r'\rightarrow^{*}r''\enskip,\\
        p^k\mathbin{\|}q\sim_{\bb}p'\mathbin{\|}r'\enskip,\ \text{and}\\
        p^k\mathbin{\|}q'\sim_{\bb}p''\mathbin{\|}r''\enskip.
      \end{gather*}
      Moreover, from $p''\mathbin{\|}r''\sim_{\bb}p^k\mathbin{\|}q'\not\sim_{\bb}p^k\mathbin{\|}q\sim_{\bb}p^k\mathbin{\|}r$ it follows that either $p''\not\sim_{\bb} p^k$ or $r''\not\sim_{\bb} r$; we distinguish these two subcases:
      \begin{enumerate}
          \item Suppose $r''\not\sim_{\bb} r$. Then, since $[r'']_{\sim_{\bb}}\prec[r]_{\sim_{\bb}}$ and $[p]_{\sim_{\bb}}\not\prec [r]_{\sim_{\bb}}$, the unique parallel decomposition of $r''$ cannot have occurrences of a process branching bisimilar to $p$. Since the unique decomposition of $p^k\mathbin{\|}q'$ has at least $k$ occurrences of a process branching bisimilar to $p$, it follows that
            $[p^k]_{\sim_{\bb}}\preceq [p'']_{\sim_{\bb}} \preceq [p']_{\sim_{\bb}} \preceq [p^k]_{\sim_{\bb}}$, so $p'\sim_{\bb}p''\sim_{\bb}p^k$. From $p^k\mathbin{\|}q'\sim_{\bb}p^k\mathbin{\|}r''$ it then follows by the induction hypothesis that $q'\sim_{\bb}r''$, and hence $q'\mathrel{\mathcal{R}}r''$.
          It remains to establish that $q\mathrel{\mathcal{R}}r'$. If $r'=r$, then, since $q\mathrel{\mathcal{R}}r$ this is immediate. If $r'\neq r$, then by the choice of $r$ such that there does not exist $r'$ such that $r\rightarrow r'$ and $r\sim_{\bb}r'$, it follows that $[r']_{\sim_{\bb}}\prec [r]_{\sim_{\bb}}$. Furthermore, since clearly $p^k\mathbin{\|}r \trans[\epsilon] p^k\mathbin{\|}r' \trans[\epsilon] p'\mathbin{\|}r'$ and $p^k\mathbin{\|}r \sim_{\bb} p'\mathbin{\|}r'$, it follows by the stuttering property
          that $p^k\mathbin{\|}r\sim_{bb}p^k\mathbin{\|}r'$. Then $q\sim_{\bb}r'$ follows by the induction hypothesis, and hence $q\mathrel{\mathcal{R}}r'$.
          \item Suppose $p''\not\sim_{\bb} p^k$. Then $[p'']_{\sim_{\bb}}\prec[p^k]_{\sim_{\bb}}$, so the multiplicity of $p$ in the unique decomposition of $p''$ is at most $k-1$. Hence, since $p''\mathbin{\|}r''\sim_{\bb}p^k\mathbin{\|}q'$, it follows that $p$ must be an element of the parallel decomposition of $r''$. This means that $[p]_{\sim_{\bb}}\preceq [r'']_{\sim_{\bb}}\preceq [r']_{\sim_{\bb}} \preceq [r]_{\sim_{\bb}}$ while at the same time $[p]_{\sim_{\bb}}\not\prec [r]_{\sim_{\bb}}$ by assumption. It follows that
          \begin{equation*}
             p\sim_{\bb}r\sim_{\bb}r'\sim_{\bb}r''\enskip.
          \end{equation*} Hence, since $r''$ can contribute at most $1$ to the multiplicity of $p$ in the unique parallel decomposition of $p''\mathbin{\|}r''$, while the multiplicity of $p$ is at least $k$, it now follows that the multiplicity of $p$ in the parallel decomposition of $p''$ must be $k-1$. So we can assume without loss of generality that there exist processes $p_1,\dots,p_k$ and $p_1'$ such that
          \begin{gather*}
             p' = p_1\mathbin{\|}p_2\mathbin{\|}\dots\mathbin{\|}p_k\enskip,\\
             p''= p_1' \mathbin{\|}p_2\mathbin{\|}\dots\mathbin{\|}p_k\enskip,\\
             p\trans[\epsilon] p_i\quad (1\leq i \leq k)\enskip,\\
             p\sim_{\bb}p_i\quad (2\leq i \leq k)\enskip,\ \text{and}\\
             p_1\trans[\mu]p_1'\enskip.
          \end{gather*}
          We must have that $p_1'\not\sim_{\bb}p$, for otherwise we would have $p''\sim_{\bb}p^k$, contradicting the assumption in this case. Therefore, since $p_1\sim_{\bb} p\sim_{\bb}r$, there exist $r_1'$ and $r_1''$ such that $r\trans[\epsilon]r_1'\trans[\mu]r_1''$ with $p\sim_{\bb}r_1'$ and $p_1'\sim_{\bb}r_1''$.
          By our choice of $r$ such that there does not exist an $r_1'\neq r$ such that $r\rightarrow r_1'$ and $r\sim_{\bb}r_1'$, it is immediate that $r=r_1'$ and hence $q\mathrel{\mathcal{R}}r'$.
          Moreover, from $p^k\mathbin{\|}q' \sim_{\bb} r\mathbin{\|}p_2\mathbin{\|}\cdots\mathbin{\|}p_k\mathbin{\|}p_1'$, $p\sim_{\bb}r$ and $p\sim_{\bb}p_i$ ($2\leq i \leq k$) it follows that $q'\sim_b p_1'\sim_b r_1''$, so $q'\mathrel{\mathcal{R}}r_1''$.
      \end{enumerate}
  \end{enumerate}
  
\end{proof}


\section{Properties of processes}
\label{sec:preliminaries}

In this section we discuss some properties of special processes on which we will build our proof of the negative result, alongside some general properties of terms modulo (rooted) branching bisimilarity.

First of all, we notice that the (rooted) depth of a closed instance of a term is always an upper bound to the (rooted) depth of the closed instances of the variables occurring in it.

\begin{lemma}
\label{lem:var_depth}
Let $t$ be a term, and $\sigma$ be a closed substitution.
For each $x \in \var(t)$, $\depth(\sigma(x)) \le \depth(\sigma(t))$, and $\rdepth(\sigma(x))\le \rdepth(\sigma(t))$.
\end{lemma}

Then, we recall the notion of $\nil$ \emph{factor} of a term:

\begin{definition}
We say that a term $t$ has a $\nil$ \emph{factor} if it contains a subterm of the form  $t' \mathbin{\|} t''$, and either $t' \sim_\rbb \nil$ or $t'' \sim_\rbb \nil$.
\end{definition}

Next, we prove that whenever two terms are provably equal modulo rooted branching bisimilarity, then they have the same variables as summands.

\begin{remark}
\label{rmk:summands}
Whenever a process term $t$ has neither $\nil$ summands nor factors then we can assume that, for some finite non-empty index set $I$, $t = \sum_{i \in I} t_i$ for some terms $t_i$ such that none of them has $+$ as head operator and moreover, none of them has $\nil$ summands nor factors.
\end{remark}

\begin{restatable}{proposition}{propsamevar}
\label{prop:same_var}
Let $t,u$ be two terms without $\nil$ summands and $\nil$ factors, and let $t \approx u$ be sound modulo $\sim_\rbb$.
If $t$ has a summand $x$, then so does $u$.
\end{restatable}

\begin{proof}
Observe, first of all, that since $t$ and $u$ have no $\nil$ summands or factors, by Remark~\ref{rmk:summands} we can assume that $t  =  \sum_{i\in I} t_i$ and $u  =  \sum_{j\in J} u_j$ for some finite non-empty index sets $I, J$, where none of the $t_i$ ($i\in I$) and $u_j$ ($j\in J$) has $+$ as its head operator, and none of the $t_i$ ($i\in I$) and $u_j$ ($j\in J$) have $\nil$ summands or factors.
Therefore, if $t$ has a summand $x$, then that there is some index $i \in I$ such that $t_i = x$.
We aim to show that $u$ has a summand $x$ as well.
Consider the substitution $\sigma_{\nil}$ mapping each variable to $\nil$.  
Pick an integer $m$ larger than the rooted depth of $\sigma_{\nil}(t)$ and of $\sigma_{\nil}(u)$. 
Let $\sigma$ be the substitution mapping $x$ to the term $a^{m+1}$ and agreeing with $\sigma_{\nil}$ on all the other variables.

As $t \approx u$ is sound modulo rooted branching bisimilarity, we have that
$
\sigma(t) \sim_\rbb \sigma(u).
$
Moreover, the term $\sigma(t)$ affords the transition $\sigma(t) \trans[a] a^{m}$, for $t_i=x$ and $\sigma(x)=a^{m+1} \trans[a] a^{m}$. 
Hence, for some closed term $p$,
\[
\sigma(u)= \sum_{j\in J} \sigma(u_j) \trans[a] p \sim_\bb a^{m} \enspace .
\]
This means that there is a $j \in J$ such that $\sigma(u_j) \trans[a] p$. 
We claim that this $u_j$ can only be the variable $x$. 
To see that this claim holds, observe, first of all, that $x\in \var(u_j)$.
In fact, if $x$ did not occur in $u_j$, then we would reach a contradiction thus:
\[
m = 
\depth(p) < \rdepth(\sigma(u_j)) = 
\rdepth(\sigma_{\nil}(u_j)) \le
\rdepth(\sigma_{\nil}(u)) < 
m \enspace .
\]
Using this observation and Lemma~\ref{lem:var_depth}, it is not hard to show that, for each of the other possible forms $u_j$ may
have, $\sigma(u_j)$ does not afford an $a$-labelled transition leading to a term of depth $m$.  
We may therefore conclude that $u_j=x$, which was to be shown.
\end{proof}

We now proceed to analyse the properties of the following family of processes, which will play a crucial role in the proof of our first main result:
\[
p_n = \sum_{i = 2}^n a.a^{\le i}
\qquad (n \ge 2)
\]
where $a^{\le i} = \sum_{j = 1}^i a^j$, and $a^j = a.a^{j-1}$, for each $j\in \{i,\dots,i\}$, $i \in \{2,\dots,n\}$.

\begin{restatable}{lemma}{lemaiindecomposable}
\label{lem:ai_indecomposable}
Let $i\ge 2$.
Process $a^{\le i}$ is indecomposable.
\end{restatable}

\begin{proof}
Assume, towards a contradiction, that there are processes $p,q \not\sim_\bb \nil$ such that $a^{\le i} \sim_\bb p \mathbin{\|} q$.
As $a^{\le i} \trans[a] \nil$, we have that $p \mathbin{\|} q \trans[\varepsilon] r \trans[a] r'$ for some processes $r,r'$ such that $r \sim_\bb a^{\le i}$ and $r' \sim_\bb \nil$.
Notice that $p \not\sim_\bb \nil$ implies that there is at least one $p' \in \der(p)$ such that $p' \trans[\mu] p''$ for some process $p''$ and action $\mu \neq \tau$.
A similar property holds for $q$.
Since, moreover, $a^{\le i}$ can only perform sequences of $a$-moves, $p \mathbin{\|} q \sim_\bb a^{\le i}$ implies that also $p$ and $q$ can perform (weak) sequences of $a$-moves.
In particular, it follows that there is no derivative of $p$ or $q$ that can perform action $\overline{a}$.
As a consequence, $p$ and $q$ cannot synchronise.
This implies that there are processes $p',q'$ such that $p \trans[\varepsilon] p'$, $q \trans[\varepsilon] q'$ and $r = p' \mathbin{\|} q'$.
Then $r \trans[a] r'$ can follow either from $p' \trans[a] p''$, or $q' \trans[a] q''$, for some $p'',q''$.
Assume, without loss of generality, that $p' \trans[a] p''$ and $r' = p'' \mathbin{\|} q'$ (the case $q' \trans[a] q''$ is analogous).
Now $p'' \mathbin{\|} q' \sim_\bb \nil$ implies $p''\sim_\bb \nil$ and $q' \sim_\bb \nil$.
Hence, $p' \sim_\bb p' \mathbin{\|} q' = r \sim_\bb a^{\le i}$ and, thus, $\depth(p) \ge i$.
Since $q \not\sim_\bb \nil$ gives $\depth(q) \ge 1$, we get that $\depth(p \mathbin{\|} q) > i$.
This gives a contradiction with $p \mathbin{\|} q \sim_\bb a^{\le i}$, since $\depth(a^{\le i}) = i \neq \depth(p \mathbin{\|} q)$ (cf.\ Lemma~\ref{lem:bb_same_depth}).
\end{proof}

\begin{remark}
\label{rmk:norm}
We could not directly infer that $a^{\le i}$ is indecomposable from the fact that the norm of the process is $1$ (where the norm of a process is the length of the shortest maximal trace that it can perform).
This is due to the fact that we are considering the full merge operator.
Consider, for instance, the process $\tau.a + \alpha.a.\overline{\alpha} + \overline{\alpha}.\alpha.a + \alpha.\overline{\alpha}.a$.
This process has observable norm $1$, yet it is not indecomposable since it is branching bisimilar to $\alpha.a \mathbin{\|} \overline{\alpha}$.
Similarly, the process $\tau + \overline{\alpha}.\alpha + \alpha.\overline{\alpha}$ has norm $1$, but it is (rooted) branching bisimilar to $\alpha \mathbin{\|} \overline{\alpha}$.
\end{remark}

\begin{restatable}{lemma}{lempnindecomposable}
\label{lem:pn_indecomposable}
Let $n\ge 2$.
Process $p_n$ is indecomposable.
\end{restatable}

\begin{proof}
Assume, towards a contradiction, that there are processes $p,q \not\sim_\bb \nil$ such that $a^{\le i} \sim_\bb p \mathbin{\|} q$.
Notice that $p \not\sim_\bb \nil$ implies that there is at least one $p' \in \der(p)$ such that $p' \trans[\mu] p''$ for some process $p''$ and action $\mu \neq \tau$.
A similar property holds for $q$.
Given any $i \in \{2,\dots,n\}$, we have that $p_n \trans[a] a^{\le i}$.
Hence, from $p_n \sim_\bb p \mathbin{\|} q$, we get that $p \mathbin{\|} q \trans[\varepsilon] r \trans[a] r'$ for some processes $r,r'$ such that $r \sim_\bb p_n$ and $r' \sim_\bb a^{\le i}$.
We remark that, since $p_n$ can only perform sequences of $a$-moves, $p \mathbin{\|} q \sim_\bb p_n$ implies that also $p$ and $q$ can perform (weak) sequences of $a$-moves.
In particular, it follows that there is no derivative of $p$ or $q$ that can perform action $\overline{a}$.
As a consequence, $p$ and $q$ cannot synchronise.
We can then distinguish four cases, according to how the sequence of transitions $p \mathbin{\|} q \trans[\varepsilon] r \trans[a] r'$ is derived:
\begin{itemize}
\item $p \mathbin{\|} q \trans[\varepsilon] p' \mathbin{\|} q \trans[a] p'' \mathbin{\|} q$, with $p' \mathbin{\|} q \sim_\bb p_n$ and $p'' \mathbin{\|} q \sim_\bb a^{\le i}$.
Since $q \not\sim_\bb \nil$ and $a^{\le i}$ is indecomposable by Lemma~\ref{lem:ai_indecomposable}, we get that $p'' \sim_\bb \nil$ and $q \sim_\bb a^{\le i}$.
Hence, given any $j \in \{0,\dots,i-1\}$, $a^{\le i} \trans[a] a^j$ implies that $q \trans[\varepsilon] q' \trans[a] q''$ for some $q',q''$ such that $q' \sim_\bb a^{\le i}$ and $q'' \sim_\bb a^j$.
Therefore, $p \mathbin{\|} q \trans[\varepsilon] p \mathbin{\|} q' \trans[a] p \mathbin{\|} q''$ and $p_n \sim_\bb p \mathbin{\|} q$ imply that $p_n \trans[a] r$ for some process $r$ such that $r \sim_\bb q'' \sim_\bb a^j$ ($p_n$ cannot perform any $\tau$-move, hence it must match any initial sequence of $\tau$-transitions by doing nothing).
However, $p_n \trans[a] a^{\le k}$, for any $k \in \{2,\dots,n\}$, and there are no $j,k$ such that $a^j \sim_\bb a^{\le k}$, thus giving a contradiction with $p_n \sim_\bb p \mathbin{\|} q$.
\item $p \mathbin{\|} q \trans[\varepsilon] p \mathbin{\|} q' \trans[a] p \mathbin{\|} q''$, with $p \mathbin{\|} q' \sim_\bb p_n$ and $p \mathbin{\|} q'' \sim_\bb a^{\le i}$. 
This case is symmetric to the previous one and therefore omitted.
\item $p \mathbin{\|} q \trans[\varepsilon] p' \mathbin{\|} q' \trans[a] p'' \mathbin{\|} q'$, with $p' \mathbin{\|} q' \sim_\bb p_n$ and $p'' \mathbin{\|} q'\sim_\bb a^{\le i}$.
Since $a^{\le i}$ is indecomposable (Lemma~\ref{lem:ai_indecomposable}), we can distinguish two cases:
\begin{itemize}
\item $p'' \sim_\bb \nil$ and $q' \sim_\bb a^{\le i}$.
Then we can proceed as in the previous case, and show that $p \mathbin{\|} q \trans[\varepsilon] p \mathbin{\|} q' \trans[a] a^j$, for some $j \le i-1$, and thus $p_n \not\sim_\bb p \mathbin{\|} q$. 
\item $p'' \sim_\bb a^{\le i}$ and $q' \sim_\bb \nil$.
In this case we have that $p \mathbin{\|} q \trans[\varepsilon] p' \mathbin{\|} q \trans[a] p'' \mathbin{\|} q$.
Hence, since $p_n$ cannot perform any $\tau$-move, $p_n \sim_\bb p \mathbin{\|} q$ requires that $p_n \trans[a] a^{\le k}$ for some $k \in \{2,\dots,n\}$ such that $a^{\le k} \sim_\bb a^{\le i} \mathbin{\|} q$.
As $a^{\le k}$ is indecomposable (Lemma~\ref{lem:ai_indecomposable}), $a^{\le i} \not\sim_\bb \nil$ for any $i \in \{2,\dots,n\}$, and by the proviso of the lemma $q \not\sim_\bb \nil$, there is no index $k \in \{2,\dots,n\}$ realising the required equivalence.
Hence, also this case gives a contradiction with $p_n \sim_\bb p \mathbin{\|} q$.
\end{itemize}
\item $p \mathbin{\|} q \trans[\varepsilon] p' \mathbin{\|} q' \trans[a] p' \mathbin{\|} q''$, with $p' \mathbin{\|} q' \sim_\bb p_n$ and $p' \mathbin{\|} q''\sim_\bb a^{\le i}$.
This case is symmetric to the previous one and therefore omitted.
\end{itemize}
\end{proof}

The next lemma characterises the parallel components of a process that is branching bisimilar to the process $a \mathbin{\|} p_n$, for some $n \ge 2$.

\begin{restatable}{lemma}{lemapnpar}
\label{lem:apn_par}
Let $n \ge 2$.
Assume that $p \mathbin{\|} q \sim_\bb a \mathbin{\|} p_n$ for some processes $p,q \not\sim_\bb \nil$.
Then either $p \sim_\bb a$ and $q \sim_\bb p_n$, or $p \sim_\bb p_n$ and $q \sim_\bb a$.
\end{restatable}

\begin{proof}
First of all, we notice that we can apply the same reasoning used in the proof of Lemma~\ref{lem:ai_indecomposable} to argue that no $\tau$-move performed by any derivative of $p \mathbin{\|} q$ can follow from a communication between the two parallel components, for otherwise we get a contradiction with $p \mathbin{\|} q \sim_\bb a \mathbin{\|} p_n$.

We have that $a \mathbin{\|} p_n \trans[a] \nil \mathbin{\|} p_n \sim_\bb p_n$.
Hence, as $p \mathbin{\|} q \sim_\bb a \mathbin{\|} p_n$, there are processes $r,r'$ such that $p \mathbin{\|} q \trans[\varepsilon] r \trans[a] r'$, $r \sim_\bb a \mathbin{\|} p_n$, and $r' \sim_\bb p_n$.
Since no communication can occur, we have that there are processes $p',q'$ such that $p \trans[\varepsilon] p'$, $q \trans[\varepsilon] q'$, and $r = p' \mathbin{\|} q'$.
As $r \trans[a] r'$, assume, without loss of generality, that $p' \trans[a] p''$ and $r' = p'' \mathbin{\|} q'$.
(Under this assumption we will obtain that $p \sim_\bb a$ and $q \sim_\bb p_n$. 
The symmetric case is obtained when $q' \trans[a] q''$ and $r' = p' \mathbin{\|} q''$, by applying the same reasoning.)
Since $p'' \mathbin{\|} q' \sim_\bb p_n$ and $p_n$ is indecomposable by Lemma~\ref{lem:pn_indecomposable}, we can distinguish the following two cases:
\begin{itemize}
\item $p'' \sim_\bb p_n$ and $q' \sim_\bb \nil$.
As $p' \trans[a] p''$, this implies that $\depth(p) \ge n+2$.
At the same time, $q \not\sim_\bb \nil$ gives $\depth(q) \ge 1$.
Consequently, we have that $\depth(p \mathbin{\|} q) = \depth(p) + \depth(q) \ge n+3 > n+2 = \depth(a \mathbin{\|} p_n)$, giving thus a contradiction with $p \mathbin{\|} q \sim_\bb a \mathbin{\|} p_n$.
\item $p'' \sim_\bb \nil$ and $q' \sim_\bb p_n$.
In this case, we have that $r = p' \mathbin{\|} q' \sim_\bb p' \mathbin{\|} p_n$.
Moreover, $q \trans[\varepsilon] q'$ and $\depth(q') = n+1$ give that $\depth(q) \ge n+1$.
Now, as $p \not\sim_\bb \nil$ gives $\depth(p) \ge 1$, from $\depth(p \mathbin{\|} q) = \depth(a \mathbin{\|} p_n) = n+2$ we infer that $\depth(q) = n+1$ and $\depth(p) = 1$.
As $a \mathbin{\|} p_n$ can perform only sequences of $a$-moves, and $\depth(p) = 1$, we obtain that whenever $p \wtrans[\alpha] p'$ then $\alpha = a$ and $p' \sim_\bb \nil$.

We can also show that there is no $p'$ such that $p \trans[\varepsilon] p'$ and $p' \sim_\bb \nil$.
Assume, towards a contradiction, that such a process $p'$ exists.
Then $p \mathbin{\|} q \trans[\varepsilon] p' \mathbin{\|} q \sim_\bb q$.
As $a \mathbin{\|} p_n$ cannot perform any $\tau$-transition, the branching bisimulation relation requires that $a \mathbin{\|} p_n \sim_\bb q$.
However, this is not possible since $\depth(a \mathbin{\|} p_n) \neq \depth(q)$.

We have therefore obtained that:
\begin{itemize}
\item whenever $p \wtrans[\alpha] p'$ then $\alpha = a$ and $p' \sim_\bb \nil$, and
\item there is no $p'$ such that $p \trans[\varepsilon] p'$ and $p' \sim_\bb \nil$.
\end{itemize}
This is enough to infer that $p \sim_\bb a$ (as the second condition guarantees that whenever $p \trans[\varepsilon] p' \trans[a] p''$ then $p' \sim_\bb a$ also holds.)

Let $p'$ be any process such that $p \trans[\varepsilon] p' \trans[a] p''$ for some $p'' \sim_\bb \nil$.
Then $p \mathbin{\|} q \trans[\varepsilon] p' \mathbin{\|} q \trans[a] p'' \mathbin{\|} q \sim_\bb q$.
As $p \mathbin{\|} q \sim_\bb a \mathbin{\|} p_n$, then $a \mathbin{\|} p_n \trans[\varepsilon] r \trans[a] r'$ for some processes $r,r'$ such that $r \sim_\bb p' \mathbin{\|} q$ and $r' \sim_\bb q$.
Since, however, $a \mathbin{\|} p_n$ cannot perform any $\tau$-transition, we get that $a \mathbin{\|} p_n \trans[a] r'$ for some $r' \sim_\bb q$.
Moreover, as $\depth(q) = n+1$, we have that either $r' = p_n$ or $r' = a \mathbin{\|} a^{\le n}$, as not other $a$-derivative of $a \mathbin{\|} p_n$ has observable depth $n+1$.
Hence, to conclude, it is enough to show that the case $r' = a \mathbin{\|} a^{\le n}$ yields a contradiction with $p \mathbin{\|} q \sim_\bb a \mathbin{\|} p_n$.
Assume that $q \sim_\bb a \mathbin{\|} a^{\le n}$.
Then, as $p \sim_\bb a$, we get that $p \mathbin{\|} q \sim_\bb a \mathbin{\|} (a \mathbin{\|} a^{\le n})$.
Now $a \mathbin{\|} (a \mathbin{\|} a^{\le n}) \trans[a] a \mathbin{\|} (a \mathbin{\|} a^{n-1}) \sim_\bb a^{n+1}$.
However, there is no process $r$ such that $a \mathbin{\|} p_n \trans[a] r$ and $r \sim_\bb a^{n+1}$.
Hence $a \mathbin{\|} p_n \not\sim_\bb a \mathbin{\|} (a \mathbin{\|} a^{\le n})$, giving thus a contradiction with $a \mathbin{\|} p_n \sim_\bb p \mathbin{\|} q$.
\end{itemize}
We have therefore obtained that $p \sim_\bb a$ and $q \sim_\bb p_n$, thus concluding the proof.
\end{proof}

We also need Lemma~\ref{lem:apn_par} to hold modulo rooted branching bisimilarity.

\begin{restatable}{lemma}{lemapnparrbb}
\label{lem:apn_par_rbb}
Let $n \ge 2$.
Assume that $p \mathbin{\|} q \sim_\rbb a \mathbin{\|} p_n$ for some processes $p,q \not\sim_\rbb \nil$.
Then either $p \sim_\rbb a$ and $q \sim_\rbb p_n$, or $p \sim_\rbb p_n$ and $q \sim_\rbb a$.
\end{restatable}

\begin{proof}
First of all, we notice that $\tau \not \in \init(a \mathbin{\|} p_n)$ implies $\tau \not \in \init(p)$ and $\tau \not \in \init(q)$.
In particular, this guarantees that $p,q \not\sim_\rbb \nil$ implies $p,q \not\sim_\bb \nil$ (Lemma~\ref{lem:same_initials}).

We have that $a \mathbin{\|} p_n \trans[a] \nil \mathbin{\|} p_n \sim_\bb p_n$.
Hence, as $p \mathbin{\|} q \sim_\rbb a \mathbin{\|} p_n$, there is a process $r$ such that $p \mathbin{\|} q \trans[a] r$ and $r \sim_\bb p_n$.
We can assume, without loss of generality, that $p \trans[a] p'$ and $r = p' \mathbin{\|} q$, for some process $p'$.
(Under this assumption we will obtain that $p \sim_\rbb a$ and $q \sim_\rbb p_n$. 
The symmetric case is obtained when $q \trans[a] q'$ and $r = p \mathbin{\|} q'$, by applying the same reasoning.)
Since $p' \mathbin{\|} q \sim_\bb p_n$, $q \not\sim_\bb \nil$, and $p_n$ is indecomposable by Lemma~\ref{lem:pn_indecomposable}, it follows that $p' \sim_\bb \nil$ and $q \sim_\bb p_n$.
Moreover, as $\init(q) = \init(p_n)$ and $\tau \not\in \init(p_n)$, by Lemma~\ref{lem:same_initials} we get $q \sim_\rbb p_n$.

Hence, $p \mathbin{\|} q \sim_\rbb p \mathbin{\|} p_n$.
Assume now that $p \trans[\mu] p'$ for some action $\mu$ and process $p'$.
Form $p \mathbin{\|} p_n \sim_\rbb a \mathbin{\|} p_n$ we infer $\mu = a$, and we can distinguish two cases, according to $a$-move is performed by $a \mathbin{\|} p_n$ to match the $a$-transition by $p \mathbin{\|} p_n$:
\begin{itemize}
\item $p' \mathbin{\|} p_n \sim_\bb p_n$.
As $p_n$ is indecomposable (Lemma~\ref{lem:pn_indecomposable}), we can directly conclude that $p' \sim_\bb \nil$.
\item $p' \mathbin{\|} p_n \sim_\bb a \mathbin{\|} a^{\le i}$ for some $i \in \{2,\dots,n\}$.
Due to the form of $p_n$, this case is not possible.
\end{itemize}
Thus, whenever $p \trans[a] p'$ we have that $p' \sim_\bb \nil$, which, together with $\init(p) = \{a\}$, gives $p \sim_\rbb a$.
\end{proof}

We conclude this section with the following result, showing a particular case in which we can establish that not only a variable $x$ triggers the behaviour of a term $t$, but that $t$ actually coincides with $x$.

\begin{restatable}{lemma}{lemaisum}
\label{lem:ai_sum}
Let $t$ be a term that does not have $+$ as head operator, and let $\sigma$ be a closed substitution.
Assume that $\sigma(t)$ has neither $\nil$ summands nor $\nil$ factors, and that $\sigma(t) \sim_\rbb \sum_{k = 1}^m aa^{\le i_k}$, for some $m > 1$, and $2 \le i_1 < \dots < i_m$.
Then $t = x$ for some variable $x$.
\end{restatable}

\begin{proof}
Assume towards a contradiction that $t$ is not a variable.
We proceed by a case analysis on the other possible forms $t$ may have:
\begin{itemize}
\item $t = \mu.t'$ for some $\mu \in \Acttau$ and term $t'$.
Then $\mu = a$ and $a^{\le i_1} \sim_\bb \sigma(t') \sim_\bb a^{\le i_m}$.
However, this is a contradiction because, as $i_1 < i_m$, the terms $a^{\le i_1}$ and $a^{\le i_m}$ have different observable depths, and therefore they cannot be branching bisimilar.
\item $t = t' \mathbin{\|} t''$.
Since $\sigma(t)$ has no $\nil$ factors we have that $\sigma(t'),\sigma(t'') \not \sim_\rbb \nil$.
In particular, notice that since $\tau \not\in \init(\sigma(t')) \cup \init(\sigma(t''))$ (as $\tau \not\in \init(\sum_{k = 1}^m aa^{\le i_k})$), it follows that $\sigma(t'),\sigma(t'') \not \sim_\bb \nil$ as well (Lemma~\ref{lem:same_initials}).
Moreover, observe that $\sum_{k =1}^m a.a^{\le i_k} \trans[a] a^{\le i_m}$.
As $\sigma(t) = \sigma(t') \mathbin{\|} \sigma(t'') \sim_\rbb \sum_{k=1}^m a.a^{\le i_k}$, there is a process $r$ such that $\sigma(t') \mathbin{\|} \sigma(t'') \trans[a] r$ and $r \sim_\bb a^{\le i_m}$.
Assume, without loss of generality, that $\sigma(t') \trans[a] p$ and $r = p \mathbin{\|} \sigma(t'')$, for some process $p$.
As $\sigma(t'') \not\sim_\bb \nil$ and $a^{\le i_m}$ is indecomposable (Lemma~\ref{lem:ai_indecomposable}), it follows that $p \sim_\bb \nil$ and $\sigma(t'') \sim_\bb a^{\le i_m}$.
Now notice that $\sum_{k=1}^m a.a^{\le i_k} \trans[a] a^{\le i_1}$.
We can distinguish two cases according to how this transition is matched by $\sigma(t') \mathbin{\|} \sigma(t'')$ in the rooted branching bisimulation game:
\begin{itemize}
\item $\sigma(t') \trans[a] p'$, for some process $p'$, and $p' \mathbin{\|} \sigma(t'') \sim_\bb a^{\le i_1}$.
Then we can proceed as above to show that $\sigma(t'') \sim_\bb a^{\le i_1}$.
As $a^{\le i_1} \not\sim_\bb a^{\le i_m}$, this gives a contradiction.
\item $\sigma(t'') \trans[a] q$, for some process $q$, and $\sigma(t') \mathbin{\|} q \sim_\bb a^{\le i_1}$.
Since $\sigma(t') \not\sim_\bb \nil$ and $a^{\le i_1}$ is indecomposable (Lemma~\ref{lem:ai_indecomposable}), it follows that $q \sim_\bb \nil$ and $\sigma(t') \sim_\bb a^{\le i_1}$.
This implies that $\sigma(t') \mathbin{\|} \sigma(t'') \sim_\bb a^{\le i_1} \mathbin{\|} a^{\le i_m}$, and thus $\depth(\sigma(t') \mathbin{\|} \sigma) = i_1 + i_m > i_m +1$ (since $i_1 \ge 2$).
In particular, since $\tau \not \in \init(\sigma(t') \mathbin{\|} \sigma(t'')) \cup \init(a^{\le i_1} \mathbin{\|} a^{\le i_m})$, we can infer that $\depth(\sigma(t') \mathbin{\|} \sigma) = \rdepth(\sigma(t') \mathbin{\|} \sigma(t''))$.
Therefore, we get a contradiction with $\sigma(t') \mathbin{\|} \sigma(t'') \sim_\rbb \sum_{k=1}^m a.a^{\le i_k}$, since $\rdepth(\sum_{k=1}^m a.a^{\le i_k}) = i_m+1$.
\end{itemize}
\end{itemize}
\end{proof}


\section{Proof of Proposition~\ref{prop:rbb_substitution_case}}

\proprbbsubstitutioncase*

\begin{proof}
Observe, first of all, that since $\sigma(t)=p$ and $\sigma(u)=q$ have no $\nil$ summands or factors, then neither do $t$ and $u$.  
Hence, by Remark~\ref{rmk:summands}, we have that for some finite non-empty index sets $I, J$,
$
t = \sum_{i\in I} t_i 
\text{ and } 
u = \sum_{j\in J} u_j,
$
where none of the $t_i$ ($i\in I$) and $u_j$ ($j\in J$) is $\nil$, has $+$ as its head operator, has $\nil$ summands and factors.

As $p=\sigma(t)$ has a summand rooted branching bisimilar to $a \mathbin{\|} p_n$, there is an index $i\in I$ such that $\sigma(t_i) \sim_\rbb a \mathbin{\|} p_n$. 

Our aim is now to show that there is an index $j\in J$ such that $\sigma(u_j) \sim_\rbb a \mathbin{\|} p_n$, proving that $q=\sigma(u)$ also has a summand rooted branching bisimilar to $a \mathbin{\|} p_n$. 

We proceed by a case analysis on the form $t_i$ may have.

\begin{enumerate}
\item \label{Case:t-var} 
\textcolor{blue}{\sc Case $t_i = x$ for some variable $x$}.
In this case, we have $\sigma(x) \sim_\rbb a \mathbin{\|} p_n$, and $t$ has $x$ as a summand. 
As $t \approx u$ is sound modulo rooted branching bisimilarity and neither $t$ nor $u$ have $\nil$ summands or factors, it follows that $u$ also has $x$ as a summand (Proposition~\ref{prop:same_var}). 
Thus there is an index $j\in J$ such that $u_j = x$, and, modulo rooted branching bisimulation, $\sigma(u)$ has $a \mathbin{\|} p_n$ as a summand, which was to be shown.

\item \label{Case:t-pre} 
\textcolor{blue}{\sc Case $t_i = \mu t'$ for some term $t'$}. 
This case is vacuous because, since the root condition implies $\mu = a$, and $a \sigma(t') \trans[a] \sigma(t')$ is the only transition afforded by $\sigma(t_i)$, this term cannot be bisimilar to $a \mathbin{\|} p_n$. 
Indeed $a \mathbin{\|} p_n$ can perform both, $a \mathbin{\|} p_n \trans[a] \nil \mathbin{\|} p_n \sim_\bb p_n$, and $a \mathbin{\|} p_n \trans[a] a \mathbin{\|} a^{\le i}$, for some $i \in \{2,\dots,n\}$.
Clearly, $p_n \not\sim_\bb a \mathbin{\|} a^{\le i}$ for any $i \in \{2,\dots,n\}$, and thus $\sigma(t')$ cannot be branching bisimilar to both of them.

\item \label{Case:t-par} 
\textcolor{blue}{\sc Case $t_i = t' \mathbin{\|} t''$ for some terms $t',t''$}. 
In this case, we have $\sigma(t') \mathbin{\|} \sigma(t'') \sim_\rbb a \mathbin{\|} p_n$.
As $\sigma(t_i)$ has no $\nil$ factors, it follows that $\sigma(t') \not\sim_\rbb \nil$ and $\sigma(t'') \not\sim_\rbb \nil$.  
Thus by Lemma~\ref{lem:apn_par_rbb} (and without loss of generality) $\sigma(t') \sim_\rbb a$ and $\sigma(t'') \sim_\rbb p_n$. 
Now, $t''$ can be written as $t'' = v_1 + \cdots + v_M,  (M > 0)$, where none of the summands $v_i$ is $\nil$ or a sum. 
Observe that, since $n$ is larger than the size of $t$, we have that $M < n$.
Hence, since $\sigma(t'') \sim_\rbb p_n$, there must be some $h \in \{1,\ldots,M\}$ such that $\sigma(v_h) \sim_\rbb \sum_{k=1}^{m} a a^{\le i_k}$ for some $m>1$ and $2 \le i_1 < \ldots < i_m \le n$. 
The term $\sigma(v_h)$ has no $\nil$ summands or factors, or else so would $\sigma(t'')$, and thus $p=\sigma(t)$. 
By Lemma~\ref{lem:ai_sum}, it follows that $v_h$ can only be a variable $x$ and thus that
\begin{equation}
\label{eq:sigma(x)}
\sigma(x) \sim_\rbb \sum_{k=i}^{m} a a^{\le i_k} \enspace .
\end{equation}
Observe, for later use, that, since $t'$ has no $\nil$ factors, the above equation yields that $x \not\in \var(t')$, or else $\sigma(t') \not\sim_\rbb a$ (Lemma~\ref{lem:var_depth}). 
So, modulo rooted branching bisimilarity, $t_i$ has the form $t' \mathbin{\|} (x+t''')$, for some term $t'''$, with $x \not\in \var(t')$ and $\sigma(t') \sim_\rbb a$.
    
Our order of business will now be to use the information collected so far to argue that $\sigma(u)$ has a summand rooted branching bisimilar to $a \mathbin{\|} p_n$. 
To this end, consider the substitution 
\[
\sigma' = \sigma[x \mapsto \overline{a} (a \mathbin{\|} p_n)] \enspace . 
\]
We have that 
\begin{align*}
\sigma'(t_i) &  ={} \sigma'(t') \mathbin{\|} \sigma'(t'') \\
& ={} \sigma(t') \mathbin{\|} \sigma'(t'') & \text{(As $x\not\in\var(t')$)} \\
& \sim_\rbb a \mathbin{\|} (\overline{a} (a \mathbin{\|} p_n) +\sigma'(t''')) & \text{(As $t'' = x + t'''$).}
\end{align*}
Thus, there is some process $p'$ such that $\sigma'(t_i) \trans[\tau] p' \sim_\bb a \mathbin{\|} p_n$, so that
\[
\sigma'(t) \trans[\tau] p' \sim_\bb a \mathbin{\|} p_n
\] 
also holds.
Since $t\approx u$ is sound modulo $\sim_\rbb$, it follows that 
\[
\sigma'(t) \sim_\rbb \sigma'(u) \enspace . 
\]
Hence, we can infer that there are a $j\in J$ and a $q'$ such that
\begin{equation}
\label{eq:A.14}
\sigma'(u_j) \trans[\tau] q' \sim_\bb a \mathbin{\|} p_n \enspace .
\end{equation}
Recall that, by one of the assumptions of the proposition, $\sigma(u) \sim_\rbb a \mathbin{\|} p_n$, and thus $\sigma(u)$ has rooted depth $n+2$. 
On the other hand, by Equation (\ref{eq:A.14}),
\[
\rdepth(\sigma'(u_j)) \ge n+3 \enspace . 
\]
Since $\sigma$ and $\sigma'$ differ only in the closed term they map variable $x$ to, it follows that
\begin{equation}
\label{eq:A.15}
x \in \var(u_j) \enspace .
\end{equation}
We now proceed to show that $\sigma(u_j) \sim_\rbb a \mathbin{\|} p_n$ by a further case analysis on the form a term $u_j$ satisfying Equation (\ref{eq:A.14}) and (\ref{eq:A.15}) may have.
\begin{enumerate}
\item \label{Case:uj-var} 
\textcolor{red}{\sc Case $u_j = x$}. 
This case is vacuous because $\sigma' (x) = \overline{a} (a \mathbin{\|} p_n) \ntrans[\tau]$, and thus this possible form for $u_j$ does not meet Equation (\ref{eq:A.14}).

\item \label{Case:uj-pre} 
\textcolor{red}{\sc Case $u_j = \mu u'$ for some term $u'$}. 
In light of Equation (\ref{eq:A.14}), we have that $\mu=\tau$ and $q'=\sigma'(u') \sim_\bb a \mathbin{\|} p_n$. 
Using Equation (\ref{eq:A.15}), we get $x \in \var(u')$ and, thus, $\depth(\sigma'(u')) \geq n+3$ (Lemma~\ref{lem:var_depth}).  
Since $a \mathbin{\|} p_n$ has depth $n+2$, this contradicts $q' \sim_\bb a \mathbin{\|} p_n$.

\item \label{Case:uj-hmerge} 
\textcolor{red}{\sc Case $u_j = u' \mathbin{\|} u''$ for some terms $u',u''$}. 
Our assumption that $\sigma(u)$ has no $\nil$ factors yields that none of the terms $u',u'',\sigma(u')$ and $\sigma(u'')$ is rooted branching bisimilar to $\nil$. 
In addition, $\init(\sigma(u)) = \{a\}$, gives $\tau \not\in \init(\sigma(u')) \cup \init(\sigma(u''))$, so that, by Lemma~\ref{lem:same_initials}, $\sigma(u'),\sigma(u'') \not\sim_\bb \nil$ holds.
Moreover, by Equation (\ref{eq:A.15}), either $x\in\var(u')$ or $x\in\var(u'')$ (and a simple observation on the depth of terms guarantees that it cannot occur in both).
      
By Equation (\ref{eq:A.14}), we have that $\sigma(u') \mathbin{\|} \sigma'(u'') \trans[\tau] q' \sim_\bb a \mathbin{\|} p_n$. 
Hence, we can distinguish three cases, according to the possible origins for such a transition according to Lemma~\ref{lem:closed2open_tau}:

\begin{enumerate} 
\item \textcolor{ForestGreen}{There is a term $w$ such that $u' \mathbin{\|} u'' \trans[\tau] w$.}
This case is vacuous, since $u' \mathbin{\|} u'' \trans[\tau] w$ implies $u \trans[\tau] w$, and, thus, $\sigma(u) \trans[\tau] \sigma(w)$, giving a contradiction with $\sigma(u) \sim_\rbb a \mathbin{\|} p_n$.

\item \textcolor{ForestGreen}{$u' \trans[\alpha] w$ and there is a variable $y \in \var(u'')$ such that $\sigma'(y) \trans[\overline{\alpha}] r_y$, $u' \mathbin{\|} u'' \trans[(y)]_{\overline{\alpha},\tau} c$, and $\sigma'[y \mapsto r_y](c) = q'$.}
(The symmetric case with $u'' \trans[\alpha] w$ and $y \in \var(u')$ can be treated in the same way and it is therefore omitted.)

Notice that if $\alpha \neq a$ then $\sigma(u') \trans[\alpha]$ implies $\sigma(u) \trans[\alpha]$, and thus $\init(\sigma(u)) \neq \init(a \mathbin{\|} p_n)$, giving a contradiction with $\sigma(u) \sim_\rbb a \mathbin{\|} p_n$.
Hence, $\alpha = a$ and, consequently, $\overline{\alpha} = \overline{a}$.

Moreover, if $y \neq x$ then $\sigma'(y) = \sigma(y)$ giving $\sigma(y) \trans[\overline{a}]$ and thus $\sigma(u) \trans[\overline{a}]$.
Also in this case we obtain a contradiction with $\sigma(u) \sim_\rbb a \mathbin{\|} p_n$.
Hence $y = x$.

Since $\sigma'(x) \trans[\overline{a}] a \mathbin{\|} p_n$, we have $\sigma'(u') \mathbin{\|} \sigma'(u'') \trans[\tau] \sigma'(w) \mathbin{\|} (a \mathbin{\|} p_n) \mathbin{\|} \sigma'(w')$ (where the general term $w'$ is included, as from $x \in \var(u'')$, and $u'' \trans[(x)]_{\mu}$, we can infer that $u''$ can be written in the general form $((x+w_1) \mathbin{\|} w') + w_2$).
From $\sigma'(w) \mathbin{\|} (a \mathbin{\|} p_n) \mathbin{\|} \sigma'(w') \sim_\bb a \mathbin{\|} p_n$ it follows (reasoning on the depth of terms) that $\sigma'(w) \sim_\bb \nil \sim_\bb \sigma'(w')$ and, in particular, $\sigma'(w) = \sigma(w)$ and $\sigma'(w') = \sigma(w')$.
Moreover, as $x \not \in \var(u')$, we have that $\sigma'(u') = \sigma(u')$.

Summarising, so far we have obtained that there is a process $r$ such that $\sigma(u') \trans[a] r$ and $r \sim_\bb \nil$, and that $x \in \var(u'')$ with $u \trans[(x)]_\mu$.
Next, we notice that $\sigma(u') \trans[a] r$ implies that $\sigma(u_j) \trans[a] r \mathbin{\|} \sigma(u'')$, so that $\sigma(u) \trans[a] r \mathbin{\|} \sigma(u'') \sim_\bb \sigma(u'')$.
From $\sigma(u) \sim_\rbb a \mathbin{\|} p_n$, it then follows that $a \mathbin{\|} p_n \trans[a] r'$ for some process $r'$ such that $r' \sim_\bb \sigma(u'')$.
Given the structure of $a \mathbin{\|} p_n$, we can distinguish two major cases, according to the possible form of $r'$:
\begin{itemize}
\item $r' = a \mathbin{\|} a^{\le i}$ for some $i \in \{2,\dots,n\}$.
We proceed to show that in this case it is not possible to have $r' \sim_\bb \sigma(u'')$.
Recall that, by Equation (\ref{eq:sigma(x)}), $\sigma(x) \sim_\rbb \sum_{k=1}^m aa^{\le i_k}$ for $m>1$ and $2 \le i_1 < \dots < i_m \le n$.
As $u'' \trans[(x)]_a$, by Lemma~\ref{lem:var_to_term} we have that $\sigma(u'') \trans[a] a^{\le i_1}$ and $\sigma(u'') \trans[a] a^{\le i_m}$
(where no parallel composition is considered because we have seen in the analysis above that any potential parallel component would be branching bisimilar to $\nil$).
In particular, $i_1 < i_m$ gives $a^{\le i_1} \not\sim_\bb a^{\le i_m}$ as the two processes have different depths (besides different branching structures).
On the one hand, we have that $r' \trans[a] a \mathbin{\|} a^j$ for some $j \in \{0,\dots,i-1\}$, and $a \mathbin{\|} a^j \sim_\bb a^{j+1} \not\sim_\bb a^{\le i_1},a^{\le i_m}$ for any $j \in \{0,\dots,i-1\}$ and $i_1,i_m \ge 2$.
On the other hand, $r' \trans[a] \nil \mathbin{\|} a^{\le i} \sim_\bb a^{\le i}$, but $a^{\le i}$ cannot be branching bisimilar to both $a^{\le i_1}$ and $a^{\le i_m}$.
Hence, $r' \not\sim_\bb \sigma(u'')$ in this case.

\item $r' = \nil \mathbin{\|} p_n$.
Then $\sigma(u'') \sim_\bb p_n$.
Notice that $\sigma(u) \sim_\rbb a \mathbin{\|} p_n$ implies $\init(\sigma(u'')) = \{a\}$.
Hence, by Lemma~\ref{lem:same_initials}, $\tau \not\in \init(\sigma(u'')) \cup \init(p_n)$ and $\sigma(u'') \sim_\bb p_n$ imply $\sigma(u'') \sim_\rbb p_n$.
In particular, it follows that $\rdepth(\sigma(u'')) = n+1$.
Since $\rdepth(\sigma(u)) = n+2 \ge \rdepth(\sigma(u_j)) = \rdepth(\sigma(u')) + \rdepth(\sigma(u''))$, we obtain that $\rdepth(\sigma(u')) = 1$.
In addition, $\init(\sigma(u')) = \{a\}$.
It follows that whenever $\sigma(u') \trans[\mu] q'$, for some process $q'$, then $\mu = a$ and $q' \sim_\bb \nil$.
Hence, $\sigma(u') \sim_\rbb a$.

Summarising, we have obtained that $\sigma(u_j) = \sigma(u') \mathbin{\|} \sigma(u'') \sim_\rbb a \mathbin{\|} p_n$, and it is thus the summand we were looking for.
\end{itemize}

\item \textcolor{ForestGreen}{There are variables $y,z$ such that $\sigma'(y) \trans[\alpha] r_y$, $\sigma'(z) \trans[\overline{\alpha}] r_z$, $u' \mathbin{\|} u'' \trans[(y,z)]_{\tau} c$, and $\sigma'[y \mapsto r_y, z \mapsto r_z](c) = q'$.}
We can proceed as in the previous case to argue that if both $y$ and $z$ are different from $x$, then we obtain a contradiction with $\sigma(u) \sim_\rbb a \mathbin{\|} p_n$.
Similarly, it cannot be the case that both $y$ and $z$ are equal to $x$, for otherwise they would not be able to synchronise as $\init(\sigma'(x)) = \{\overline{a}\}$.
Hence, we can assume, without loss of generality, that $y \neq x$ and $z = x$.
Thus, $\sigma'(y) = \sigma(y)$ and $\alpha = a$.
As we can also assume, without loss of generality, that $y \in \var(u')$ and $x \in \var(u'')$ (having both in one term gives a contradiction either with $q' \sim_\bb a \mathbin{\|} p_n$ or $u',u'' \not\sim_\rbb \nil$).
We can then proceed exactly as in the previous case to show that $\sigma(u') \sim_\rbb a$ and $\sigma(u'') \sim_\rbb p_n$, giving that $u_j$ is the desired summand.
\end{enumerate}
This completes the proof for the case $u_j = u' \mathbin{\| }u''$ for some terms $u',u''$.
\end{enumerate}
This completes the analysis for the case $t_i = t' \mathbin{\| }t''$ for some terms $t',t''$.
\end{enumerate}
The proof of Proposition~\ref{prop:rbb_substitution_case} is now complete. 
\end{proof}


\section{Proof of Theorem~\ref{thm:rbb_preserves_property}}

\thmrbbpreservesproperty*

\begin{proof}
Assume that $\E$ is a finite axiom system over the language CCS that is sound modulo rooted branching bisimilarity.
Recall that, without loss of generality, we may assume that the closed terms involved in the proof of an equation $p \approx q$ have no $\nil$ summands or factors (by Proposition~\ref{Propn:proofswithout0_text}, as $\E$ may be assumed to be saturated), and that applications of symmetry happen first in equational proofs
(that is, $\E$ is closed with respect to symmetry). 
Hence, by the proviso of the theorem, the following hold, for some closed terms $p$ and $q$ and positive integer $n$ larger than the size of each term in the equations in $\E$:
\begin{enumerate}
\item $E \vdash p \approx q$,
\item $p \sim_\rbb q \sim_\rbb a \mathbin{\|} p_n$, 
\item $p$ and $q$ contain no occurrences of $\nil$ as a summand or factor, and
\item $p$ has a summand rooted branching bisimilar to $a \mathbin{\|} p_n$.
\end{enumerate}
We prove that $q$ also has a summand rooted branching bisimilar to $a \mathbin{\|} p_n$ by induction on the depth of the closed proof of the equation $p \approx q$ from $\E$. 
  
We proceed by a case analysis on the last rule used in the proof of $p \approx q$ from $\E$. 
The case of reflexivity is trivial, and that of transitivity follows immediately by using the inductive hypothesis twice. 
Below we only consider the other possibilities.
\begin{itemize}
\item \textcolor{blue}{\sc Case $E \vdash p \approx q$, because $\sigma(t) = p$ and $\sigma(u) = q$ for some equation $(t\approx u)\in E$ and closed substitution $\sigma$}. 
Since $\sigma(t) = p$ and $\sigma(u) = q$ have no $\nil$ summands or factors, and $n$ is larger than the size of each term mentioned in equations in $\E$, the claim follows by Proposition~\ref{prop:rbb_substitution_case}.
\item \textcolor{blue}{\sc Case $E \vdash p \approx q$, because $p=\mu p'$ and $q=\mu q'$ for some $p',q'$ such that $E \vdash p' \approx q'$}.
This case is vacuous because $p=\mu p'\not \sim_\rbb a \mathbin{\|} p_n$, and thus $p$ does not have a summand rooted branching bisimilar to $a \mathbin{\|} p_n$.
\item \textcolor{blue}{\sc Case $E \vdash p \approx q$, because $p = p'+p''$ and $q = q'+q''$ for some $p',q',p'',q''$ such that $E \vdash p' \approx q'$ and $E \vdash p'' \approx q''$}.
Since $p$ has a summand rooted branching bisimilar to $a \mathbin{\|} p_n$, we have that so does either $p'$ or $p''$. 
Assume, without loss of generality, that $p'$ has a summand rooted branching bisimilar to $a \mathbin{\|} p_n$. 
Since $p$ is bisimilar to $a \mathbin{\|} p_n$, so is $p'$. 
Using the soundness of $\E$ modulo $\sim_\rbb$, it follows that $q' \sim_\rbb a \mathbin{\|} p_n$. 
The inductive hypothesis now yields that $q'$ has a summand rooted branching bisimilar to $a \mathbin{\|} p_n$. 
Hence, $q$ has a summand rooted branching bisimilar to $a \mathbin{\|} p_n$, which was to be shown.
\item \textcolor{blue}{\sc Case $E \vdash p \approx q$, because $p = p' \mathbin{\|} p''$ and $q = q' \mathbin{\|} q''$ for some $p',q',p'',q''$ such that $E \vdash p' \approx q'$ and $E \vdash p'' \approx q''$}. 
Since the proof involves no uses of $\nil$ as a summand or a factor, we have that $p',p''\not\sim_\rbb \nil$ and $q',q''\not\sim_\rbb \nil$. 
It follows that $q$ is a summand of itself. 
By our assumptions, 
$
a \mathbin{\|} p_n \sim_\rbb q.
$
Therefore we have that $q$ has a summand bisimilar to $a \mathbin{\|} p_n$, and we are done.
\end{itemize}
This completes the proof of Theorem~\ref{thm:rbb_preserves_property}.
\end{proof}


\section{Additional properties of terms and auxiliary transitions}

We discuss here some properties of configurations and auxiliary transitions, that we omitted from the main text in Section~\ref{sec:configurations_bis} because of space limits.
The proofs of the results in this section follow by simple inductions on the derivation of transitions and/or on the structure of terms (configurations), and are therefore omitted.

We recall that thanks to the axioms in Table~\ref{tab:axioms_b} we can prove the following classic result for $\ccslc$ terms, which will be useful in the rest of paper:

\begin{lemma}
\label{lem:summands}
For every $\ccslc$ term $t$ there are terms $t_1,\dots,t_n$ ($n \ge 0$) that do not have $+$ as head operator such that $t \approx \sum_{i = 1}^n t_i$ (by A0--A3).
\end{lemma}
The terms $t_i$ are also called the \emph{summands} of $t$.

Moreover, notice that Lemmas~\ref{lem:c_structure} and~\ref{lem:var_to_term} presented for CCS configurations, can be extended in a straightforward manner to $\ccslc$ terms.

\begin{remark}
For Lemma~\ref{lem:var_to_term}.\ref{item:x_u_to_t}, Lemma~\ref{lem:substitution} guarantees that applying $\sigma$ to the term $t'$ that triggers, together with $x$, the transition $t \trans[(x)]_{\alpha,\tau}$, does not affect the possibility of synchronisation between $\sigma(x)$ and $\sigma(t')$.
\end{remark}

We also underline the following properties of variables $x_\mu \in \Var_{\Acttau}$.
We extend the definition of $\var(c)$ in accordance with the newly introduced set of variables:
$\var(c) = \{x \in \Var \mid x \text{ occurs in } c\} \cup \{x_\mu \in \Var_{\Acttau} \mid x_\mu \text{ occurs in } c\}$.

\begin{lemma}
\label{lem:xmu_occurrence}
Let $c \in \C_{\text{LC}}$ and $x \in \Var$.
If $c \trans[\ell]_\rho c'$, for some $c'$ and $\ell$ with $x \in \ell$,
then $c'$ is of the form $x_\mu \mathbin{\|} c''$ for some (possibly null) configuration $c''$ (modulo the axioms in Table~\ref{tab:axioms_b}).
Similarly, if $x_\mu \in \var(c)$, then $c$ is of the form $x_\mu \mathbin{\|} c''$ for some (possibly null) configuration $c''$ (modulo the axioms in Table~\ref{tab:axioms_b}).
\end{lemma}

The following two lemmas, that correspond to Lemma~\ref{lem:closed2open_alpha} and Lemma~\ref{lem:closed2open_tau} respectively, allow us to show how a transition $\sigma(t) \trans[\mu] p$ can stem from transitions of the $\ccslc$ term $t$ and of the process $\sigma(x)$, for $x \in \var(t)$.

\begin{lemma}
\label{lem:closed2open_alpha_ccslc}
Let $\alpha \in \Act \cup \overline{\Act}$, $t$ be a $\ccslc$ term, $\sigma$ be a closed substitution, and $p$ be a process.
If $\sigma(t) \trans[\alpha] p$, then one of the following holds:
\begin{enumerate}
\item \label{c2o_prefix_ccslc}
There is a $\ccslc$ term $t'$ s.t.\ $t \trans[\alpha] t'$ and $\sigma(t') = p$.
\item \label{c2o_x_ccslc}
There are a variable $x$, a process $q$ and a configuration $c$ s.t.\ $\sigma(x) \trans[\alpha] q$, $t \trans[(x)]_\alpha c$, and $\sigma[x_\alpha \mapsto q](c) = p$.
\end{enumerate}
\end{lemma}

\begin{lemma}
\label{lem:closed2open_tau_ccslc}
Let $t$ be a $\ccslc$ term, $\sigma$ be a closed substitution, and $p$ be a process.
If $\sigma(t) \trans[\tau] p$, then one of the following holds:
\begin{enumerate}
\item \label{lem:c2o_prefix_ccslc}
There is a $\ccslc$ term $t'$ s.t.\ $t \trans[\tau] t'$ and $\sigma(t') = p$.
\item \label{lem:c2o_x_ccslc}
There are a variable $x$, a process $q$, and a configuration $c$ s.t.\ $\sigma(x) \trans[\tau] q$, $t \trans[(x)]_\tau c$, and $\sigma[x_\tau \mapsto q](c) = p$.
\item \label{lem:c2o_xu_ccslc}
There are a variable $x$, a process $q$, and a configuration $c$ s.t., for some $\alpha \in \Act\cup\overline{\Act}$, $\sigma(x) \trans[\alpha] q$, $t \trans[(x)]_{\alpha,\tau} c$, and $\sigma[x_\alpha \mapsto q](c) = p$.
\item \label{lem:c2o_xy_ccslc}
There are variables $x,y$, processes $q_x,q_y$ and a configuration $c$ s.t., for some $\alpha \in \Act\cup\overline{\Act}$, $\sigma(x) \trans[\alpha] q_x$, $\sigma(y) \trans[\overline{\alpha}] q_y$, $t \trans[(x,y)]_\tau c$, and $\sigma[x_\alpha \mapsto q_x, y_{\overline{\alpha}} \mapsto q_y](c) = p$.
\end{enumerate}
\end{lemma}

\begin{proof}
The proof follows by induction on the derivation of the transition $\sigma(t) \trans[\tau] p$.
The interesting case is the base case $t = t_1 \cmerge t_2$, for some CCS terms $t_1,t_2$, that we expand below.

According to the operational semantics of the communication merge, a $\tau$-move by $\sigma(t_1 \cmerge t_2)$ can only be generated by a communication between $\sigma(t_1)$ and $\sigma(t_2)$.
We can assume, without loss of generality, that $\sigma(t_1) \trans[\alpha] p_1$ and $\sigma(t_2) \trans[\overline{\alpha}] p_2$, for some processes $p_1,p_2$ such that $p_1 \mathbin{\|} p_2 = p$.
The symmetric case of $\sigma(t_1) \trans[\overline{\alpha}]$ and $\sigma(t_2) \trans[\alpha]$ follows by applying the same arguments we use below, switching the roles of $\sigma(t_1)$ and $\sigma(t_2)$.
By Lemma~\ref{lem:closed2open_alpha_ccslc}, from $\sigma(t_1) \trans[\alpha] p_1$ we can infer that either $t_1 \trans[\alpha] t_1'$ for some term $t_1'$ such that $\sigma(t_1') = p_1$, or there are a variable $x$, a process $q_1$ and a configuration $c_1$ such that $\sigma(x) \trans[\alpha] q_1$, $t_1 \trans[(x)]_\alpha c_1$, and $\sigma[x_\alpha \mapsto q_1](c_1) = p_1$.
Similarly, by Lemma~\ref{lem:closed2open_alpha_ccslc}, from $\sigma(t_2) \trans[\overline{\alpha}] p_2$ we can infer that either $t_2 \trans[\overline{\alpha}] t_2'$ for some term $t_2'$ such that $\sigma(t_2') = p_2$, or there are a variable $y$, a process $q_2$ and a configuration $c_2$ such that $\sigma(y) \trans[\overline{\alpha}] q_2$, $t_2 \trans[(y)]_{\overline{\alpha}} c_2$, and $\sigma[y_{\overline{\alpha}} \mapsto q_2](c_2) = p_2$.
We can then distinguish four cases, depending on how the possible derivations of the transitions of $t_1$ and $t_2$ mentioned above are combined in the derivation of the synchronisation. 
\begin{enumerate}
\item Case $t_1 \trans[\alpha] t_1'$ and $t_2 \trans[\overline{\alpha}] t_2'$.

In this case, the operational semantics of $\cmerge$ allows us to directly infer that $t_1 \cmerge t_2 \trans[\tau] t_1' \mathbin{\|} t_2'$.
Hence we have that there is $t' = t_1' \mathbin{\|} t_2'$ such that $t \trans[\tau] t'$ and $\sigma(t') = p_1 \mathbin{\|} p_2 = p$.

\item Case $\sigma(x) \trans[\alpha] q_1$, $t_1 \trans[(x)]_\alpha c_1$, $\sigma[x_\alpha \mapsto q_1](c_1) = p_1$, and $t_2 \trans[\overline{\alpha}] t_2'$.

As $t_1 \trans[(x)]_\alpha c_1$ and $t_2 \trans[\overline{\alpha}] t_2'$ we can apply the auxiliary rule ($a_5$) and obtain that $t_1 \cmerge t_2 \trans[(x)]_{\alpha,\tau} c_1 \mathbin{\|} t_2'$.
Hence we have that there are a variable $x$, a subterm $t_2$ of $t$, a process $q_1$ and a configuration $c = c_1 \mathbin{\|} t_2'$ such that $\sigma(x) \trans[\alpha] q_1$, $t_2 \trans[\overline{\alpha}] t_2'$, $t \trans[(x)]_{\alpha,\tau} c$ and $\sigma[x_\alpha \mapsto q_1](c) = \sigma[x_\alpha \mapsto q_1](c_1) \mathbin{\|} \sigma(t_2') = p_1 \mathbin{\|} p_2 = p$.

\item Case $t_1 \trans[\alpha] t_1'$, $\sigma(y) \trans[\overline{\alpha}] q_2$, $t_2 \trans[(y)]_{\overline{\alpha}} c_2$, and $\sigma[y_{\overline{\alpha}} \mapsto q_2](c_2) = p_2$.

This is the symmetrical of the previous case and it follows by applying the same reasoning used above, using the auxiliary rule ($a_6$) in place of ($a_5$).

\item Case $\sigma(x) \trans[\alpha] q_1$, $t_1 \trans[(x)]_\alpha c_1$, $\sigma[x_\alpha \mapsto q_1](c_1) = p_1$, $\sigma(y) \trans[\overline{\alpha}] q_2$, $t_2 \trans[(y)]_{\overline{\alpha}} c_2$, and $\sigma[y_{\overline{\alpha}} \mapsto q_2](c_2) = p_2$.

As $t_1 \trans[(x)]_\alpha c_1$ and $t_2 \trans[(y)]_{\overline{\alpha}} c_2$, we can apply the auxiliary rule ($a_4$) and obtain thus $t_1 \cmerge t_2 \trans[(x,y)]_{\tau} c_1 \mathbin{\|} c_2$.
Hence we have obtained that there are variables $x,y$, processes $q_1,q_2$ and a configuration $c = c_1 \mathbin{\|} c_2$ such that $\sigma(x) \trans[\alpha] q_1$, $\sigma(y) \trans[\overline{\alpha}] q_2$, $t \trans[(x,y)]_\tau c$, and $\sigma[x_\alpha \mapsto q_1, y_{\overline{\alpha}} \mapsto q_2](c) = 
\sigma[x_\alpha \mapsto q_1](c_1) \mathbin{\|} \sigma[y_{\overline{\alpha}} \mapsto q_2](c_2) = 
p_1 \mathbin{\|} p_2 = p$, where we can distribute the substitutions over $\mathbin{\|}$ since $x_\alpha$ and $y_{\overline{\alpha}}$ are unique by construction (even if $x = y$, the two subscripts allow us to distinguish them). 
\end{enumerate}
\end{proof}

\begin{remark}
\label{rmk:single_step}
It is worth noticing that in Lemmas~\ref{lem:closed2open_alpha_ccslc} and~\ref{lem:closed2open_tau_ccslc}, we only consider the derivation of a \emph{single computation step}.
In particular, each (occurrence of a) variable can trigger at most one (possibly nondeterministic) computation step of a term.
In order words, we cannot use a single (occurrence of a) variable to derive a sequence of transitions.
This step-by-step approach is applied the substitutions as well.
Consider the term $x \mathbin{\|} x$ and let $\sigma$ be a closed substitution such that $\sigma(x) = \mu.p_1 + \mu.p_2$, with $p_1 \neq p_2$.
It is immediate that $\sigma(x) \mathbin{\|} \sigma(x) \trans[\mu] p_1 \mathbin{\|} \sigma(x) \trans[\mu] p_1 \mathbin{\|} p_2$ is a valid computation.
At the same time, we have that $x \mathbin{\|} x \trans[x] x_{\dd} \mathbin{\|} x$, and $x_{\dd} \mathbin{\|} x \trans[x] x_\dd \mathbin{\|} x_{\dd}$.
As you can see, there is no way in which we can distinguish the two occurrences of $x_\dd$ in the last term.
Thus, the substitution $\sigma[x_\dd \mapsto p_1, x_\dd \mapsto p_2](x_\dd \mathbin{\|} x_\dd)$ is not well defined.
However, this is not a problem, since the substitution is applied at each step, thus $\sigma(x) \mathbin{\|} \sigma(x) \trans[\mu] \sigma[x_\dd \mapsto p_1] (x_\dd \mathbin{\|} x) = p_1 \mathbin{\|} \sigma(x)$ and $p_1 \mathbin{\|} \sigma(x) \trans[\mu] \sigma[x_\dd \mapsto p_2](p_1 \mathbin{\|} x_\dd) = p_1 \mathbin{\|} p_2$.
\end{remark}


\section{Proof of Theorem~\ref{thm:bb_on_open}}

Before proceeding to the proof, we list some properties that can be directly inferred from the operational semantics given in Tables~\ref{tab:sos_rules},~\ref{tab:sos_rules_ccslc},~\ref{tab:ell_rules_ccslc}, and~\ref{tab:c_rules}.
They can be seen as an extension of Lemma~\ref{lem:var_to_term} from terms to configurations.

\begin{lemma}
\label{lem:config_properties}
Let $c$ be a $\ccslc$ configuration and $\sigma$ be a closed substitution.
Then:
\begin{enumerate}
\item For all $\mu \in \Acttau$, whenever $c \trans[\mu] c'$ for some configuration $c'$, then $\sigma(c) \trans[\mu] \sigma(c')$.

\item For all variables $x,y \in \Var$:
\begin{enumerate}
\item For all $\mu \in \Acttau$, whenever $c \trans[(x)]_\mu c'$ for some configuration $c'$, then $\sigma(x) \trans[\mu] p_x$ implies $\sigma(c) \trans[\mu] \sigma[x_\mu \mapsto p_x](c')$.
\item Whenever $c \trans[(x,y)]_\tau c'$ for some configuration $c'$, then $\sigma(x) \trans[\alpha] p_x$ and $\sigma(y) \trans[\overline{\alpha}] p_y$ imply $\sigma(c) \trans[\tau] \sigma[x_\alpha \mapsto p_x, y_{\overline{\alpha}} \mapsto p_y](c')$, for any $\alpha \in \Act\cup\overline{\Act}$. 
\item For all $\alpha \in \Act\cup\overline{\Act}$, whenever $c \trans[(x)]_{\alpha,\tau} c'$ for some configuration $c'$, then $\sigma(x) \trans[\alpha] p_x$ implies $\sigma(c) \trans[\tau] \sigma[x_\alpha \mapsto p_x](c')$.
\end{enumerate}

\item Whenever $c \trans[x_\mu] c'$ for some configuration $c'$, then $c \not\sim_\bb \nil$, and $\sigma(c) = \sigma(c')$ for all closed substitutions $\sigma$.
\end{enumerate}
\end{lemma}

The following lemma proves a fundamental, albeit immediate, property of branching bisimilar configurations: they contain the same variables.

\begin{lemma}
\label{lem:same_variables}
Assume that $c_1 \sim_\bb c_2$.
Then, given any variable $x \in \Var$, $x \in \var(c_1)$ if{f} $x \in \var(c_2)$.
Similarly, for any $\mu \in \Acttau$, $x_\mu \in \var(c_1)$ if{f} $x_\mu \in \var(c_2)$.
\end{lemma}

Moreover, we introduce the notion of \emph{derivative} to configurations.

\begin{definition}
[Derivative]
For a configuration $c$, the set of \emph{derivatives} of $c$, notation $\der(c)$, is the least set containing $c$ that is closed under $\twoheadrightarrow$, i.e., the least set satisfying:
\begin{itemize}
\item $c \in \der(c)$, and
\item if $c' \in \der(c)$ and $c' \xitrans c''$, for some label $\xi$, then $c'' \in \der(c)$.
\end{itemize}
\end{definition}

The following lemma can be proved by induction over the structure of configurations (noticing that $\der(x_\dd) = \{x_\dd\}$ and $\der(x_\mu) = \{x_\mu\}$, for all $x_\dd \in \Var_\dd$ and $x_\mu \in \Var_{\Acttau}$).

\begin{lemma}
\label{lem:der_finite}
For every $\ccslc$ configuration $c$, the set $\der(c)$ is finite.
\end{lemma}

We can now prove Theorem~\ref{thm:bb_on_open}.

\thmbbonopen*

\begin{proof}
We prove the statement only for branching bisimilarity $\sim_\bb$.
The proofs for the other equivalences can be carried out in a similar fashion.

We prove the two implications separately.

$(\Rightarrow)$.
Assume that $c_1 \sim_\bb c_2$.
Our aim is to show that $\sigma(c_1) \sim_\bb \sigma(c_2)$ for all closed substitutions $\sigma$.
To this end it is enough to prove that the relation
\[
\rel = \{ (\sigma(c), \sigma(c')) \mid c \sim_\bb c', \sigma \text{closed substitution}\}
\]
is a branching bisimulation.
This is an immediate consequence of Lemma~\ref{lem:substitution}, Lemma~\ref{lem:config_properties} and Definition~\ref{def:open_bb}.\\

$(\Leftarrow)$.
Assume now that $\sigma(c_1) \sim_\bb \sigma(c_2)$ for all closed substitutions $\sigma$.
We aim to show that $c_1 \sim_\bb c_2$ according to Definition~\ref{def:open_bb}.
To this end we proceed by induction on the number of variables occurring in $c_1$ and $c_2$, i.e. on $|\var(c_1) \cup \var(c_2)|$.
\begin{itemize}
\item Base case: $\var(c_1) \cup \var(c_2) = \emptyset$.

In this case, we have that $c_1$ and $c_2$ are closed CCS terms, giving that $c_i = \sigma(c_i)$ for all closed substitutions $\sigma$, $i = 1,2$, and also that the all the transitions from $c_1,c_2$ and their derivatives are of the form $\trans[\mu]$ for some $\mu \in \Acttau$.
It is then immediate to verify that $c_1 \sim_\bb c_2$.

\item Inductive step: $\var(c_1) \cup \var(c_2) \neq \emptyset$.
This means that there is at least one variable $x \in \Var$ or $y_\mu \in \Var_{\Acttau}$ in the union.
Assume that we actually have both, i.e., $x,y_\mu \in \var(c_1) \cup \var(c_2)$.
This assumption does not invalidate the generality of our approach: the cases in which only $x$ or only $y_\mu$ occur in $\var(c_1) \cup \var(c_2)$ can be easily obtained from the combined case that we present here. 

Let $\alpha \in \Act\cup\overline{\Act}$.
It is immediate to verify that, for $n,m > 0$, $\alpha^n \sim_\bb \alpha^m$ if{f} $n = m$ (notice that $\alpha \neq \tau$).
Since the set of derivatives of $c_1$ and $c_2$ are finite (Lemma~\ref{lem:der_finite}), we can find $n > 0$ such that 
\begin{equation}
\label{eq:alphan_not_bb_c}
\alpha^n \not\sim_\bb c \quad \text{ for each } c \in \der(c_1) \cup \der(c_2).
\end{equation}
In particular, this implies that 
\begin{equation}
\label{eq:alphan_not_bb_c_sub}
\begin{split}
& \alpha^n \not\sim_\bb [x \mapsto \alpha^{n+2}](c), \\
& \alpha^n \not\sim_\bb [y_\mu \mapsto \alpha^{n+1}](c), \text{ and} \\
& \alpha^n \not\sim_\bb [x \mapsto \alpha^{n+2}, y_\mu \mapsto \alpha^{n+1}](c).
\end{split}
\end{equation}
In fact, we have that:
\begin{itemize}
\item If $x,y_\mu \not \in \var(c)$, then $[x \mapsto \alpha^{n+2}](c) = [y_\mu \mapsto \alpha^{n+1}](c) = [x \mapsto \alpha^{n+2}, y_\mu \mapsto \alpha^{n+1}](c) = c$, and the relation immediately follows from Equation (\ref{eq:alphan_not_bb_c}).
The same reasoning applies to all derivatives $c$ such that $c \ntrans[\ell]_\rho $ for any $\ell,\rho$ with $x \in \ell$, and for all those such that $c \ntrans[y_\mu]$. 
\item If $y_\mu \in \var(c)$, then $c$ is of the form $y_\mu \mathbin{\|} c'$ for some configuration $c'$ (Lemma~\ref{lem:xmu_occurrence}), and it then immediate to verify that $[y_\mu \mapsto \alpha^{n+1}](c)$ can perform a sequence of $n+1$ $\alpha$-moves, whereas $\alpha^n$ can perform at most $n$ of such transitions.
\item If $x \in \var(c)$, then the only interesting case to analyse is that of an occurrence of $x$ in $c$ within the scope of communication merge. 
All other cases can be treated as done for $y_\mu$.
We remark here that a potential occurrence of $x$ in the scope of $\cmerge$ is precisely the reason why $x$ is mapped to $\alpha^{n+2}$ whereas $y_\mu$ is mapped to $\alpha^{n+1}$.

According to the operational semantics in Table~\ref{tab:ell_rules_ccslc}, we can distinguish three cases:
\begin{itemize}
\item $c \trans[(x,y)]_\tau c'$ for some variable $y$ and configuration $c'$.
Hence, the occurrence of $x$ is in a subterm of the form $x \cmerge y$.
Consequently, no synchronisation can occur in $[x \mapsto \alpha^{n+2}](c)$ between $\alpha^{n+2}$ and $y$, and this case can be treated as the case in which $x \not \in \var(c)$.

\item $c \trans[(x)]_{\beta,\tau} c'$ for some configuration $c'$ and $\beta \neq \alpha$.
Hence, $x$ occurs in a subterm of the form $x \cmerge \overline{\beta}.t$ for some $t$, and we can proceed as in the previous case.

\item $c \trans[(x)]_{\alpha,\tau} c'$ for some configuration $c'$.
The occurrence of $x$ is then in a subterm of the form $x \cmerge \overline{\alpha}.t$.
Hence, in this case $[x \mapsto \alpha^{n+2}](c)$ has a derivative (the one obtained from the synchronisation of $\alpha^{n+2}$ with $\overline{\alpha}.t$) that can perform a trace with at least $n+1$ $\alpha$-transitions.
Clearly, $\alpha^n$ has no such a derivative.
\end{itemize}
\end{itemize}
We now notice that for every closed substitution $\sigma$ it holds that
\[
\sigma[x \mapsto \alpha^{n+2}, y_\mu \mapsto \alpha^{n+1}](c_1)
\sim_\bb
\sigma[x \mapsto \alpha^{n+2}, y_\mu \mapsto \alpha^{n+1}](c_2).
\]
Since 
\[
\left|
\begin{array}{c}
\var([x \mapsto \alpha^{n+2}, y_\mu \mapsto \alpha^{n+1}](c_1))\; \cup \\ \var([x \mapsto \alpha^{n+2}, y_\mu \mapsto \alpha^{n+1}](c_2))
\end{array}
\right|
< |\var(c_1) \cup \var(c_2)|,
\]
we can apply the inductive hypothesis and obtain that
\begin{equation}
\label{eq:derivatives_all_bb}
[x \mapsto \alpha^{n+2}, y_\mu \mapsto \alpha^{n+1}](c_1)
\sim_\bb
[x \mapsto \alpha^{n+2}, y_\mu \mapsto \alpha^{n+1}](c_2).
\end{equation}
Hence, to conclude our proof, we proceed to show that this implies that $c_1 \sim_\bb c_2$.
This can be done simply by showing that the relation
\begin{align*}
\rel = \{ (c,c') \mid & (c,c') \in \der(c_1) \times \der(c_2) \text{ and } \\
& [x \mapsto \alpha^{n+2}, y_\dd \mapsto \alpha^{n+1}](c) \sim_\bb [x \mapsto \alpha^{n+2}, y_\dd \mapsto \alpha^{n+1}](c')\}
\end{align*}
is a branching bisimulation.
This follows from Equation (\ref{eq:derivatives_all_bb}), Lemma~\ref{lem:substitution}, Lemma~\ref{lem:config_properties} and Equation (\ref{eq:alphan_not_bb_c_sub}), which ensures that whenever $c \rel c'$, then $x \in \var(c)$ if{f} $x \in \var(c')$, (and similarly for $y_\mu$).
\end{itemize}
\end{proof}


\section{Proof of Proposition~\ref{prop:rbb_nf_ccslc}}

Since we defined $\sim_\rbb$ over the set of $\ccslc$ configurations (that include $\ccslc$ terms), before proceeding we prove that, in equational proofs from $\E_\rbb$, terms are never transformed into configurations, in the sense that if a configuration $x_\mu$ is not introduced in the proof by a substitution instance of an axiom, then it cannot be introduced by any application of an axiom or of a rule of equational logic.
Hence, we can prove our results on the set of $\ccslc$ terms (thus simplifying the technical development).

\begin{lemma}
\label{lem:only_terms_ccslc}
\begin{enumerate}
\item \label{lem:only_terms_ccslc_1}
If $\E_\rbb \vdash t \approx u$ and $t$ is a $\ccslc$ term (i.e. it does not contain any occurrence of a variable from $\Var_{\Acttau}$), then also $u$ is a $\ccslc$ term.
\item \label{lem:only_terms_ccslc_2}
If $\E_\rbb \vdash t \approx u$ and $t$ is a $\ccslc$ term, then a proof of $t \approx u$ from $\E_\rbb$ uses only equations over $\ccslc$ terms. 
\end{enumerate}
\end{lemma}

\begin{proof}
The first item directly follows from the soundness of $\E_\rbb$ modulo $\sim_\rbb$. 
In fact, as $t \approx u$ implies $t \sim_\rbb u$, by Lemma~\ref{lem:same_variables} we get that $\var(t) = \var(u)$.
Hence, $\var(t) \cap \Var_\dd = \emptyset$ implies that $\var(u) \cap \Var_\dd = \emptyset$ as well.

Let us now deal with the second item.
First of all, we notice that since $t$ is a $\ccslc$ term, we can apply Lemma~\ref{lem:only_terms_ccslc}.\ref{lem:only_terms_ccslc_1} and obtain that $u$ is a $\ccslc$ term as well.
The proof then proceeds by induction on the length of the proof of $t \approx u$ from $\E_\rbb$, where the inductive step is carried out by a case analysis on the last rule of equational logic that is used in the proof.
The proof is standard and therefore we omit it.
We only want to point out that the only way in which a variable $x_\mu \in \Var_{\Acttau}$ can occur in an equational proof is by means of an application of a substitution.
However, we remark that, in our setting, substitutions can only be applied to axioms in equational proofs.
We can then distinguish two cases:
\begin{itemize}
\item Either no substitution instance happen in the proof for $t \approx u$, or it happens, but it does not map any variable into a configuration containing an occurrence of a variable in $\Var_{\Acttau}$.
In this case, the proof follows by induction.
\item A substitution instance is applied and it introduces a configuration with (at least) an occurrence of a variable $x_\mu$ from $\Var_{\Acttau}$.
Then we can prove that $x_\mu$ is not removed by the rules of equational logic.
More precisely, one can show that if $\E_\rbb \vdash c_1 \approx c_2$ for some configurations $c_1,c_2$, and $x_\mu \in \var(c_1)$, then, not only $x_\mu \in \var(c_2)$, but $x_\mu$ occurs in all the equations in the equational proof. 
This property can be proved by induction on the length of the proof of $c_1 \approx c_2$ from $\E_\rbb$.
Hence, in this case we get a contradiction with the proviso of the Lemma stating that $t$ is a $\ccslc$ term.
\end{itemize}
\end{proof}

The property of terms in Remark~\ref{rmk:summands} can be extended to normal forms: every normal form can be rewritten modulo A0--A3 as a summation of simple normal forms.

\begin{lemma}
\label{lem:simple_sums}
For each normal form $N$ there is a sequence (possibly empty) of simple normal forms $S_1,\dots,S_n$ such that $N \approx \sum_{i=1}^n S_i$ (by A0--A3).
\end{lemma}

\begin{remark}
\label{rmk:rbb_nf_ccslc}
Notice that variables in $\Var_{\Acttau}$ cannot occur in normal forms.
\end{remark}

Each term can be proven equal using $\E_\rbb$ to a normal form.

\proprbbnfccslc*

\begin{proof}
The proof can be carried out by induction over the size of term exactly like in the proof of Lemma 4.4 in \cite{AFIL09}, and therefore we omit it.
\end{proof}

\begin{remark}
As direct consequence of Proposition~\ref{prop:rbb_nf_ccslc}, we can henceforth assume that each $\ccslc$ term $t$ can be expressed in the general form
\[
t \approx
\sum_{i \in I} \mu_i N_i + 
\sum_{j \in J} x_j \lmerge N_j +
\sum_{h \in H} (x_h \cmerge \alpha_h) \lmerge N_h +
\sum_{k \in K} (x_k \cmerge x'_k) \lmerge N_k
.
\]
\end{remark}


\section{Proof of Proposition~\ref{prop:rbb_provable_ccslc}}

Before proceeding to the proof, we present an auxiliary result.

\begin{lemma}
\label{lem:cancellation_xmu}
Let $x_\mu \in \Var_{\Acttau}$.
Let $c_1,c_2$ be two configurations such that either $x_\mu \not\in \var(c_1) \cup \var(c_2)$, or $x_\mu \in \var(c_1) \cap \var(c_2)$.
If $x_\mu \mathbin{\|} c_1 \sim_\bb x_\mu \mathbin{\|} c_2$, then $c_1 \sim_\bb c_2$.
\end{lemma}

\begin{proof}
To prove the statement it is enough to show that the relation
\[
\rel = \{(c,c') \mid x_\mu \mathbin{\|} c \sim_\bb x_\mu \mathbin{\|} c' \text{ and  either } x_\mu \not\in \var(c) \cup \var(c') \text{, or } x_\mu \in \var(c) \cap \var(c')\}
\]
is a branching bisimulation.

Let $c_1,c_2 \in \rel$. 
We consider first the case of $c_1,c_2$ such that $x_\mu \not \in \var(c_1) \cup \var(c_2)$.
Notice that, in this case, we have that there are no $c_1',c_2'$ such that $c_1 \trans[\varepsilon] c_1' \trans[x_\mu]$ and $c_2 \trans[\varepsilon] c_2' \trans[x_\mu]$.
Assume that $c_1 \xitrans c_1'$.
Clearly, this gives $x_\mu \mathbin{\|} c_1 \trans[\xi] x_\mu \mathbin{\|} c_1'$.
We proceed by a case analysis on the possible forms of $\xi$ and on how $x_\mu \mathbin{\|} c_2'$ matches the transition from $x_\dd \mathbin{\|} c_1$, according to Definition~\ref{def:open_bb}:
\begin{itemize}
\item $\xi = \tau$ and $x_\mu \mathbin{\|} c_1' \sim_\bb x_\mu \mathbin{\|} c_2$.
Clearly, in this case, $x_\mu \not \in \var(c_1') \cup \var(c_2)$, and thus $(c_1',c_2) \in \rel$.
\item $c_1 \xitrans c_1'$ either because $c_1 \trans[\mu] c_1'$ for some $\mu \in \Acttau$, or $c_1 \trans[\ell]_\rho c_1'$ for some $\ell$ with $x \not\in \ell$, $x \in\ell$ but $\rho \neq \mu,(\mu,\tau)$.
(Notice that this includes also the case of $\xi = \tau$ and $x_\mu \mathbin{\|} c_2$ matching the $\tau$-transition from $x_\mu \mathbin{\|} c_1$ through a sequence of silent moves.)
Then $x_\mu \mathbin{\|} c_1 \xitrans x_\mu \mathbin{\|} c_1'$, and $x_\mu \mathbin{\|} c_1 \sim_\bb x_\mu \mathbin{\|} c_2$ implies that there are configurations $r,r'$ such that $x_\mu \mathbin{\|} c_2 \trans[\varepsilon] r \xitrans r'$, $r \sim_\bb x_\mu \mathbin{\|} c_1$, and $r' \sim_\bb x_\mu \mathbin{\|} c_1'$.
Since none of the transitions in the sequence $x_\mu \mathbin{\|} c_2 \trans[\varepsilon] r \xitrans r'$ can be performed by $x_\mu$, we have that there are $c_2',c_2''$ such that $c_2 \trans[\varepsilon] c_2' \xitrans c_2''$, $r = x_\mu \mathbin{\|} c_2'$ and $r' = x_\mu \mathbin{\|} c_2''$.
Moreover, from the considerations above, we can also infer that $x_\mu \not \in \var(c_2'),\var(c_2'')$.
Therefore, we have obtained that there are configurations $c_2',c_2''$ such that $c_2 \trans[\varepsilon] c_2' \xitrans c_2''$, $(c_1,c_2') \in \rel$, and $(c_1',c_2'') \in \rel$.
\item $c_1 \xitrans c_1'$ because $c_1 \trans[\ell]_\rho c_1'$ with $x \in \ell$ and $\rho\in\{\mu,(\mu,\tau)\}$.
Then $c_1 \xitrans c_1'$ where, by Lemma~\ref{lem:xmu_occurrence}, $c_1'$ is of the form $x_\mu \mathbin{\|} c$ for some (possibly null) configuration $c$.
As $x_\mu \mathbin{\|} c_1 \sim_\bb x_\mu \mathbin{\|} c_2$, it follows that $x_\mu \mathbin{\|} c_2 \trans[\varepsilon] r \xitrans r'$ for some configurations $r,r'$ such that $r \sim_\bb x_\mu \mathbin{\|} c_1$ and $r' \sim_\bb x_\mu \mathbin{\|} c_1'$.
Also in this case, none of the transitions in the sequence $x_\mu \mathbin{\|} c_2 \trans[\varepsilon] r \xitrans r'$ can be due to $x_\mu$.
Hence, there are configurations $c_2',c_2''$ such that $c_2 \trans[\varepsilon] c_2' \xitrans c_2''$, $r = x_\mu \mathbin{\|} c_2'$, and $r' = x_\mu \mathbin{\|} c_2''$.
Moreover, we notice that $x_\mu \not \in \var(c_2')$, whereas, by Lemma~\ref{lem:xmu_occurrence}, $x_\mu \in \var(c_2'')$.
Summarising, we have obtained that there are configurations $c_2',c_2''$ such that $c_2 \trans[\varepsilon] c_2' \xitrans c_2''$, $(c_1,c_2') \in \rel$, and $(c_1',c_2'') \in \rel$. 
\end{itemize}
Symmetrically, we can prove that any transition from $c_2$ is matched by $c_1$ according to the definition of branching bisimilarity. 

Consider now the case in which $c_1 \xitrans c_1'$ and $x_\mu \in \var(c_1) \cap \var(c_2)$.
The proof follows as in the previous case, by noticing that, by the operational semantics of $\mathbin{\|}$ defined in Tables~\ref{tab:sos_rules_ccslc},~\ref{tab:ell_rules_ccslc}, and~\ref{tab:c_rules}, for any configuration $c$ it holds that $x_\mu \in \var(c)$ implies $x_\mu \in \var(c')$ for all $c' \in \der(c)$.
Moreover, there is an extra case that we need to consider in the case analysis over the possible forms of the label $\xi$, namely:
\begin{itemize}
\item $\xi = x_\mu$.
This case is trivial as $x_\mu \in \var(c_i)$ implies $c_i \trans[x_\mu] c_i$, for $i \in \{1,2\}$.
\end{itemize}
\end{proof}

\begin{remark}
\label{rmk:index_notation}
We are aware that, in the following proof, the notation used for the sets of indexes is a bit heavy. 
We did our best to to explain every step in the proof so that the reader has always a clear idea of what we are doing.
\end{remark}

\proprbbprovableccslc*

\begin{proof}
By Proposition~\ref{prop:rbb_nf_ccslc} and Lemma~\ref{lem:only_terms_ccslc} it is enough to prove the statement for normal forms.
So let $t$ and $u$ be branching bisimilar normal forms.
In particular we can write them in the general forms
\begin{align*}
t \approx{} &
\sum_{i \in I} \mu_i.N_i +
\sum_{j \in J} x_j \lmerge N_j +
\sum_{h \in H} (x_h \cmerge \alpha_h) \lmerge N_h +
\sum_{k \in K} (x_k \cmerge x'_k) \lmerge N_k
\\
u \approx{} &
\sum_{\bar{i} \in \bar{I}} \nu_{\bar{i}}.M_{\bar{i}} +
\sum_{\bar{j} \in \bar{J}} y_{\bar{j}} \lmerge M_{\bar{j}} +
\sum_{\bar{h} \in \bar{H}} (y_{\bar{h}} \cmerge \beta_{\bar{h}}) \lmerge M_{\bar{h}} +
\sum_{\bar{k} \in \bar{K}} (y_{\bar{k}} \cmerge y'_{\bar{k}}) \lmerge M_{\bar{k}},
\end{align*}
where all the $N_i,N_j,N_h,N_k,M_{\bar{i}},M_{\bar{j}},M_{\bar{h}},M_{\bar{k}}$ are themselves in normal form.
We proceed to prove that $\E_\rbb \vdash \mu.t \approx \mu.u$ by induction over the sum of the sizes of $t$ and $u$.

As $t \sim_\bb u$, we can distinguish three cases, according to the form of the actions $\mu_i,\nu_{\bar{i}}$, for $i \in I, \bar{i} \in \bar{I}$.
\begin{enumerate}
\item\label{case:uno} There is no $i \in I$ such that $\mu_i = \tau$ and $N_i \sim_\bb u$, and there is no $\bar{i} \in \bar{I}$ such that $\nu_{\bar{i}} = \tau$ and $t \sim_\bb M_{\bar{i}}$.

In this case, as $t \sim_\bb u$, each transition $t \xitrans t'$ must be matched by a transition $u \xitrans u'$ for some $u'$ such that $t' \sim_\bb u'$.
Given the characterisation of the semantics of terms, and the definition of $\sim_\bb$ (Definition~\ref{def:open_bb}), we can infer the following:
\begin{itemize}
\item For each summand $\mu_i.N_i$ there is a summand $\nu_{\bar{i}}.M_{\bar{i}}$ such that $\nu_{\bar{i}} = \mu_i$ and $N_i \sim_\bb M_{\bar{i}}$.
Symmetrically, for each summand $\nu_{\bar{i}}.M_{\bar{i}}$ there is a summand $\mu_i.N_i$ such that $\mu_i = \nu_{\bar{i}}$ and $N_i \sim_\bb M_{\bar{i}}$.
As the sum of the sizes of $N_i$ and $M_{\bar{i}}$ is strictly smaller than $\size(t)+\size(u)$, by induction we get that 
\begin{equation}
\label{eq:rbb_provable_I}
\E_\rbb \vdash \mu_i.N_i \approx \mu_i.M_{\bar{i}} = \nu_{\bar{i}}.M_{\bar{i}}.
\end{equation}
\item For each $k \in K$ there is a $\bar{k} \in \bar{K}$ such that, for any $\alpha \in \Act\cup \overline{\Act}$: 
\begin{enumerate}
\item $t \trans[(x_k,x'_k)]_\tau x_{k,\alpha} \mathbin{\|} x'_{k,\overline{\alpha}} \mathbin{\|} N_k$;
\item $u \trans[(y_{\bar{k}},y'_{\bar{k}})]_\tau y_{\bar{k},\alpha} \mathbin{\|} y'_{\bar{k},\overline{\alpha}} \mathbin{\|} M_{\bar{k}}$;
\item\label{item:tre} $x_k \cmerge x'_k = y_{\bar{k}} \cmerge y'_{\bar{k}}$ modulo C1;
\item\label{item:quattro} $x_{k,\alpha} \mathbin{\|} x'_{k,\overline{\alpha}} \mathbin{\|} N_k \sim_\bb y_{\bar{k},\alpha} \mathbin{\|} y'_{\bar{k},\overline{\alpha}} \mathbin{\|} M_{\bar{k}}$.
\end{enumerate}
Symmetrical relations hold for each $\bar{k} \in \bar{K}$.
As $x_{k,\alpha},x'_{k,\overline{\alpha}},y_{\bar{k},\alpha},y'_{\bar{k},\overline{\alpha}}$ are univocally determined by, respectively, $x_k,x'_k,y_{\bar{k}},y'_{\bar{k}}$, from item (\ref{item:tre}) above we infer that $x_{k,\alpha} \mathbin{\|} x'_{k,\overline{\alpha}} = y_{\bar{k},\alpha} \mathbin{\|} y'_{\bar{k},\overline{\alpha}}$ modulo D1.
Hence, by item (\ref{item:quattro}) and two applications of Lemma~\ref{lem:cancellation_xmu}, we obtain that $N_k \sim_\bb M_{\bar{k}}$.
As the sum of the sizes of $N_k$ and $M_{\bar{k}}$ is strictly smaller than $\size(t)+\size(u)$, by induction we get that $\E_\rbb \vdash \tau.N_k \approx \tau.M_{\bar{k}}$.
(Please notice that since $N_k$ and $M_{\bar{k}}$ are normal forms, they are in particular CCS terms, and induction is therefore well defined on them.)
Then, we have 
\begin{align*}
(x_k \cmerge x'_k) \lmerge N_k &
\stackrel{\scalebox{0.7}{(TL)}}{\approx{}} (x_k \cmerge x'_k) \lmerge \tau.N_k \\
& \approx{} (x_k \cmerge x'_k) \lmerge \tau.M_{\bar{k}} \\
& \stackrel{\scalebox{0.7}{(TL)}}{\approx{}} (x_k \cmerge x'_k) \lmerge M_{\bar{k}} \\
& \approx{} (y_{\bar{k}} \cmerge y'_{\bar{k}}) \lmerge M_{\bar{k}}
\end{align*}
Summarising, we have obtained that for each $k \in K$ (respectively, $\bar{k} \in \bar{K}$) there is a $\bar{k} \in \bar{K}$ (respectively, $k \in K$) such that:
\begin{equation}
\label{eq:rbb_provable_K}
\E_\rbb \vdash (x_k \cmerge x'_k) \lmerge N_k \approx (y_{\bar{k}} \cmerge y'_{\bar{k}}) \lmerge M_{\bar{k}}.
\end{equation}
\item For each $j \in J$ there is a $\bar{j} \in \bar{J}$ such that, given any $\mu \in \Acttau$:
\begin{enumerate}
\item $t \trans[(x_j)]_\mu x_{j,\mu} \mathbin{\|} N_j$;
\item $u \trans[(y_{\bar{j}})]_\mu y_{\bar{j},\mu} \mathbin{\|} M_{\bar{j}}$;
\item $x_j = y_{\bar{j}}$;
\item $x_{j,\mu} \mathbin{\|} N_j \sim_\bb y_{\bar{j},\mu} \mathbin{\|} M_{\bar{j}}$.
\end{enumerate}
Symmetrical relations hold for each $\bar{j} \in \bar{J}$.
We can then proceed as the case of indexes in $K, \bar{K}$ and obtain that for each $j \in J$ (respectively, $\bar{j} \in \bar{J}$) there is a $\bar{j} \in \bar{J}$ (respectively, $j \in J$) such that:
\begin{equation}
\label{eq:rbb_provable_J}
\E_\rbb \vdash x_j \lmerge N_j \approx y_{\bar{j}} \lmerge M_{\bar{j}}.
\end{equation}
\item For each $h \in H$ there is a $\bar{h} \in \bar{H}$ such that, given any $\alpha \in \Act\cup\overline{\Act}$:
\begin{enumerate}
\item $t \trans[(x_h)]_{\overline{\alpha_h},\tau} x_{h,\overline{\alpha_h}} \mathbin{\|} N_h$;
\item $u \trans[(y_{\bar{h}})]_{\overline{\beta_{\bar{h}}},\tau} y_{\bar{h},\overline{\beta_{\bar{h}}}} \mathbin{\|} M_{\bar{h}}$;
\item $x_h = y_{\bar{h}}$ and $\alpha_h = \beta_{\bar{h}}$;
\item $x_{h,\overline{\alpha_h}} \mathbin{\|} N_h \sim_\bb y_{\bar{h},\overline{\beta_{\bar{h}}}} \mathbin{\|} M_{\bar{h}}$.
\end{enumerate}
Symmetrical relations hold for each $\bar{h} \in \bar{H}$.
We can then proceed as the case of indexes in $K, \bar{K}$ and obtain that for each $h \in H$ (respectively, $\bar{h} \in \bar{H}$) there is a $\bar{h} \in \bar{H}$ (respectively, $h \in H$) such that:
\begin{equation}
\label{eq:rbb_provable_H}
\E_\rbb \vdash (x_h \cmerge \alpha_h) \lmerge N_h \approx (y_{\bar{h}} \cmerge \beta_{\bar{h}}) \lmerge M_{\bar{h}}.
\end{equation}
\end{itemize}
Equations (\ref{eq:rbb_provable_I})--(\ref{eq:rbb_provable_H}) taken together give $\E_\rbb \vdash t \approx u$, from which it is immediate to infer $\E_\rbb \vdash \mu.t \approx \mu.u$, for any $\mu \in \Acttau$, and the proof is complete is this case.

\item Assume now that $\mu_i = \tau$ and $N_i \sim_\bb u$ for some $i \in I$, and that $\nu_{\bar{i}} = \tau$ and $t \sim_\bb M_{\bar{i}}$ for some $\bar{i} \in \bar{I}$.
Clearly, we have that $N_i \sim_\bb u \sim_\bb t \sim_\bb M_{\bar{i}}$, and $\size(N_i) + \size(M_{\bar{i}}) < \size(t) + \size(u)$, so that by induction we obtain
\[
\E_\rbb \vdash \mu.N_i \approx \mu.M_{\bar{i}}
\qquad
\E_\rbb \vdash \mu.t \approx \mu.M_{\bar{i}}
\qquad
\E_\rbb \vdash \mu.N_i \approx \mu.u
\]
from which $\E_\rbb \vdash \mu.t \approx \mu.u$ can be inferred, and the proof is complete in this case.

\item Assume that there is an index $i \in I$ such that $\mu_i = \tau$ and $N_i \sim_\bb u$, but there is no $\bar{i} \in \bar{I}$ such that $\nu_{\bar{i}} = \tau$ and $t \sim_\bb M_{\bar{i}}$.
(The symmetric case can be treated similarly and it is therefore omitted.)
For every summand $\tau.N_i$ of $t$ with $N_i \sim_\bb u$ we have that the sum of the sizes of $N_i$ and $u$ is strictly smaller than $\size(t)+\size(u)$.
Hence, by induction we obtain that $\E_\rbb \vdash \tau.N_i \approx \tau.u$ for all such summands.
Thus, possibly applying axioms A0--A3, we can infer that
\begin{equation}
\E_\rbb \vdash t \approx \tau.u + N
\end{equation}
where 
\[
N = \sum_{i \in I_t} \mu_i.N_i +
\sum_{j \in J} x_j \lmerge N_j +
\sum_{h \in H} (x_h \cmerge \alpha_h) \lmerge N_h +
\sum_{k \in K} (x_k \cmerge x'_k) \lmerge N_k
\]
with $I_t = \{i \in I \mid \mu_i \neq \tau \vee N_i \not\sim_\bb u\}$.
Given the condition on the indexes in $I_t$, and considering that but there is no $\bar{i} \in \bar{I}$ such that $\nu_{\bar{i}} = \tau$ and $t \sim_\bb M_{\bar{i}}$, it is immediate to verify that whenever $N \xitrans C$ then $u \xitrans c$ for some $c$ such that $C \sim_\bb c$.
In particular, by applying the same reasoning used in the analysis of case~\ref{case:uno} above, we have:
\begin{itemize}
\item for each $i \in I_t$ there is a $\bar{i}_i \in \bar{I}$ such that
\[
\E_\rbb \vdash \mu_i.N_i \approx \nu_{\bar{i}_i}.M_{\bar{i}_i};
\]
\item for each $j \in J$ there is a $\bar{j}_j \in \bar{J}$ such that 
\[
\E_\rbb \vdash x_j \lmerge N_j \approx y_{\bar{j}_j} \lmerge M_{\bar{j}_j};
\]
\item for each $h \in H$ there is a $\bar{h}_h \in \bar{H}$ such that 
\[
\E_\rbb \vdash (x_h \cmerge \alpha_h) \lmerge N_h \approx (y_{\bar{h}_h} \cmerge \beta_{\bar{h}_h}) \lmerge M_{\bar{h}_h};
\]
\item for each $k \in K$ there is a $\bar{k}_k \in \bar{K}$ such that 
\[
\E_\rbb \vdash (x_k \cmerge x'_k) \lmerge N_k \approx (y_{\bar{k}_k} \cmerge y'_{\bar{k}_k}) \lmerge M_{\bar{k}_k}.
\]
\end{itemize}
Summarising, we have obtained that
\[
\E_\rbb \vdash u \approx N + M
\]
where
\begin{align*}
M ={} &
\sum \{\nu_{\bar{i}}.M_{\bar{i}} \mid \bar{i} \neq \bar{i}_i \text{ for all } i\} +
\sum \{y_{\bar{j}} \lmerge M_{\bar{j}} \mid \bar{j} \neq \bar{j}_j \text{ for all } j \} + \\
& \sum \{(y_{\bar{h}} \cmerge \beta_{\bar{h}} \mid \lmerge M_{\bar{h}} \mid \bar{h} \neq \bar{h}_h \text{ for all } h\} +
\sum \{(y_{\bar{k}} \cmerge y'_{\bar{k}} \mid \lmerge M_{\bar{k}} \mid \bar{k} \neq \bar{k}_k \text{ for all } k\}.
\end{align*}
Then:
\begin{align*}
\E_\rbb \vdash \mu.t \approx{} &
\mu.(\tau.u + N) \\
\approx{} & 
\mu.(\tau.(N + M) + N) \\
\stackrel{\scalebox{0.7}{(TB)}}{\approx{}} &
\mu.(N + M) \\
\approx{} &
\mu.u
\end{align*}
and the proof follows also in this case.
\end{enumerate}
\end{proof}


\section{Proof of Theorem~\ref{thm:rbb_complete_ccslc}}

\thmrbbcompleteccslc*

\begin{proof}
To prove the claim, it is enough to prove that $\E_\rbb \vdash t \approx t + u$, since, by symmetry of $\sim_\rbb$, this also gives $\E_\rbb \vdash u \approx t + u$ and thus that $\E_\rbb \vdash t \approx u$.

By Proposition~\ref{prop:rbb_nf_ccslc} and Lemma~\ref{lem:only_terms_ccslc}, $t$ and $u$ can be written in normal form as follows:
\begin{align*}
t \approx{} &
\sum_{i \in I} \mu_i.N_i +
\sum_{j \in J} x_j \lmerge N_j +
\sum_{h \in H} (x_h \cmerge \alpha_h) \lmerge N_h +
\sum_{k \in K} (x_k \cmerge x'_k) \lmerge N_k
\\
u \approx{} &
\sum_{\bar{i} \in \bar{I}} \nu_{\bar{i}}.M_{\bar{i}} +
\sum_{\bar{j} \in \bar{J}} y_{\bar{j}} \lmerge M_{\bar{j}} +
\sum_{\bar{h} \in \bar{H}} (y_{\bar{h}} \cmerge \beta_{\bar{h}}) \lmerge M_{\bar{h}} +
\sum_{\bar{k} \in \bar{K}} (y_{\bar{k}} \cmerge y'_{\bar{k}}) \lmerge M_{\bar{k}},
\end{align*}
where all the $N_i,N_j,N_h,N_k,M_{\bar{i}},M_{\bar{j}},M_{\bar{h}},M_{\bar{k}}$ are themselves in normal form.
As $t \sim_\rbb u$ and the equations in $\E_\rbb$ are sound modulo rooted branching bisimilarity, we can apply the same reasoning used in case~\ref{case:uno} of the proof of Proposition~\ref{prop:rbb_provable_ccslc} and obtain that:
\begin{itemize}
\item For each index $\bar{i} \in \bar{I}$ there is an index $i_{\bar{i}} \in I$ such that $\mu_{i_{\bar{i}}} = \nu_{\bar{i}}$ and $N_{i_{\bar{i}}} \sim_\bb M_{\bar{i}}$.
Then, for each $\bar{i} \in \bar{I}$, by Proposition~\ref{prop:rbb_provable_ccslc} we obtain that 
\begin{equation}
\label{eq:rbb_complete_I}
\E_\rbb \vdash \mu_{i_{\bar{i}}}.N_{i_{\bar{i}}} \approx \mu_{i_{\bar{i}}}.M_{\bar{i}} = \nu_{\bar{i}}.M_{\bar{i}}.
\end{equation}
\item For each $\bar{j} \in \bar{J}$ there is a $j_{\bar{j}} \in J$ such that $y_{\bar{j}} = x_{j_{\bar{j}}}$ and $M_{\bar{j}} \sim_\bb N_{j_{\bar{j}}}$.
Then, for each $\bar{j} \in \bar{J}$, by Proposition~\ref{prop:rbb_provable_ccslc} and axiom TL we obtain that
\begin{equation}
\label{eq:rbb_complete_J}
\E_\rbb \vdash x_{j_{\bar{j}}} \lmerge N_{j_{\bar{j}}} \approx
x_{j_{\bar{j}}} \lmerge \tau.N_{j_{\bar{j}}} \approx y_{\bar{j}} \lmerge \tau.M_{\bar{j}} \approx y_{\bar{j}} \lmerge M_{\bar{j}}.
\end{equation}
\item For each $\bar{h} \in \bar{H}$ there is a $h_{\bar{h}} \in H$ such that $y_{\bar{h}} = x_{h_{\bar{h}}}$, $\beta_{\bar{h}} = \alpha_{h_{\bar{h}}}$ and $M_{\bar{h}} \sim_\bb N_{h_{\bar{h}}}$.
Then, for each $\bar{h} \in \bar{H}$, by Proposition~\ref{prop:rbb_provable_ccslc} and axiom TL we obtain that
\begin{equation}
\label{eq:rbb_complete_H}
\E_\rbb \vdash (x_{h_{\bar{h}}} \cmerge \alpha_{h_{\bar{h}}}) \lmerge N_{h_{\bar{h}}} \approx (x_{h_{\bar{h}}} \cmerge \alpha_{h_{\bar{h}}}) \lmerge \tau.N_{h_{\bar{h}}} \approx (y_{\bar{h}} \cmerge \beta_{\bar{h}}) \lmerge \tau.M_{\bar{h}} \approx (y_{\bar{h}} \cmerge \beta_{\bar{h}}) \lmerge M_{\bar{h}}.
\end{equation}
\item For each $\bar{k} \in \bar{K}$ there is a $k_{\bar{k}} \in K$ such that $y_{\bar{k}} \cmerge y'_{\bar{k}} = x_{k_{\bar{k}}} \cmerge x'_{k_{\bar{k}}}$, modulo C1, and $M_{\bar{k}} \sim_\bb N_{k_{\bar{k}}}$.
Then, for each $\bar{k} \in \bar{K}$, by Proposition~\ref{prop:rbb_provable_ccslc} and axiom TL we obtain that
\begin{equation}
\label{eq:rbb_complete_K}
\E_\rbb \vdash (x_{k_{\bar{k}}} \cmerge x'_{k_{\bar{k}}}) \lmerge N_{k_{\bar{k}}} \approx (x_{k_{\bar{k}}} \cmerge x'_{k_{\bar{k}}}) \lmerge \tau.N_{k_{\bar{k}}} \approx (y_{\bar{k}} \cmerge y'_{\bar{k}}) \lmerge \tau.M_{\bar{k}} \approx (y_{\bar{k}} \cmerge y'_{\bar{k}}) \lmerge M_{\bar{k}}.
\end{equation}
\end{itemize}
The fact that $\E_\rbb \vdash t \approx t + u$ then immediately follows from Equations (\ref{eq:rbb_complete_I})--(\ref{eq:rbb_complete_K}).
\end{proof}

\end{document}